\documentclass[11pt,a4paper]{article}
\usepackage[hmargin=1.25in,vmargin=1.25in]{geometry}
\usepackage{appendix,enumitem,setspace,titlesec}
\usepackage{amssymb,amsthm,mathrsfs,mathtools}
\usepackage[linesnumbered,ruled,vlined]{algorithm2e}
\usepackage[dvipsnames,svgnames,x11names]{xcolor} % Must load before tikz
\usepackage{natbib,bibunits,float,subcaption,tabularray,xurl}
\UseTblrLibrary{booktabs,siunitx}
\usepackage{pgfplots,pgfplotstable}
\pgfplotsset{compat=1.18}
\usetikzlibrary{arrows.meta,calc,plotmarks,positioning}
\usepgfplotslibrary{fillbetween,groupplots,colorbrewer}
\usepackage[hyperindex,breaklinks,hypertexnames=false]{hyperref} % hyperindex,breaklinks needed for arXiv
\usepackage[open,openlevel=2,numbered]{bookmark}  % Must load after hyperref

\AtBeginDocument{%
  \setlength{\abovedisplayskip}{10.3pt plus 3pt minus 6pt}%
  \setlength{\belowdisplayskip}{10.3pt plus 3pt minus 6pt}%
}

\hypersetup{colorlinks, linkcolor={NavyBlue}, citecolor={NavyBlue}, urlcolor={NavyBlue}}
\DontPrintSemicolon
\SetAlgoCaptionSeparator{.} 
\SetKwComment{Comment}{\#\ }{}
\AtBeginEnvironment{figure}{\captionsetup[figure]{font=small,labelsep=period,labelfont=sc}}
\AtBeginEnvironment{table}{\captionsetup[table]{font=small,labelsep=period,labelfont=sc}}
\DeclareTblrTemplate{caption-sep}{default}{.\space}
\DeclareTblrTemplate{remark-sep}{default}{.\space}
\SetTblrStyle{caption-tag}{font=\small\scshape}
\SetTblrStyle{remark-tag}{font=\footnotesize\itshape}
\SetTblrStyle{caption-text}{font=\small}
\SetTblrStyle{remark-text}{font=\footnotesize}
\SetTblrInner[tblr,talltblr,longtblr]{cells={font=\footnotesize}}
\SetTblrStyle{remark}{halign=l}
\pgfplotsset{title style={font=\small}}

\newtheorem{thm}{Theorem}[section]

\newtheorem{lem}{Lemma}[section]
\newtheorem{pro}{Proposition}[section]
\newtheorem{ass}{Assumption}[section]
\theoremstyle{definition}

\newtheorem{ex}{Example}[section]
\newtheorem{rem}{Remark}[section]
\newtheoremstyle{exctd}
{\topsep} {\topsep}%
{\upshape}% Body font
{}% Indent amount (empty = no indent, \parindent = para indent)
{\bfseries}% Thm head font
{.}% Punctuation after thm head
{1em}% Space after thm head (\newline = linebreak)
{\thmname{#1} \thmnumber{ #2}\thmnote{#3} (cont.)}% Thm head spec
\theoremstyle{exctd}
\newtheorem*{exctd}{Example}

\begin{document}
\makeatletter
\gdef\@extra@b@citeb{@bu1}
\gdef\@extra@binfo{@bu1}
\makeatother
\begin{bibunit}[ecta]
\pdfbookmark[1]{Title}{title}
\title{Debiased Machine Learning: Identification, Estimation, and Shape Constraints\thanks{We are grateful to Whitney K.\ Newey for his insights and encouragement. We also thank St\'{e}phane Bonhomme, Florian F.\ Gunsilius, Essie Maasoumi, Mikkel Plagborg-M{\o}ller, Pedro H.\ C.\ Sant'Anna, and participants at seminars and conferences for their helpful comments and discussions. We acknowledge the use of OpenAI's ChatGPT for language editing, literature searches, and coding assistance. }} 
\author{
Qihui Chen\\ School of Management and Economics \\CUHK--Shenzhen\\ qihuichen@cuhk.edu.cn
\and
Ka Yan Cheng \\ Department of Economics \\ Emory University\\ ka.yan.cheng@emory.edu
\and
Zheng Fang \\ Department of Economics \\ Emory University\\ zheng.fang@emory.edu}
\date{July 29, 2026}
\maketitle

\begin{abstract}
We develop a general framework of identification and estimation for automatic debiased machine learning (DML) where the parameter of interest $\theta_0$ is identified by a moment condition involving a nuisance $\gamma_0$ that may be high dimensional. We establish conditions under which the Riesz representer $\alpha_0$, which is at the core of DML, is identified, and show that the identification occurs precisely when $\alpha_0$ uniquely optimizes a quadratic functional. This characterization enables us to develop a general estimation procedure for $\alpha_0$ that allows for generic $\gamma_0$ including those defined by models with endogeneity and encompasses both classical sieves and modern architectures such as deep neural networks. To improve estimation precision and mitigate the curse of dimensionality, we incorporate shape constraints on $\gamma_0$ by embedding them into a possibly nonlinear parameter space. We illustrate our estimation procedure through simulations and empirical applications. 
\end{abstract}

\begin{center}
\end{center}

\newpage

\section{Introduction}

Debiased machine learning (DML) has emerged as a powerful framework for statistical inference in empirical research, particularly in applications with big data \citep{AhrensChernozhukovHansenKozburSchafferWiemann2026DML}. In these settings, the parameter of interest $\theta_0$ is identified by a moment restriction that depends on a first step nuisance parameter $\gamma_0$. Thus, inference on $\theta_0$ must account for estimation error in $\gamma_0$, a challenge exacerbated when $\gamma_0$ is high dimensional. While machine learning is well suited for estimating such high dimensional objects, it tends to induce regularization and overfitting biases that may lead to biased and $\sqrt{n}$-inconsistent estimation of $\theta_0$ \citep{ChernozhukovChetverikovDemirerDufloHansenNeweyRobins2018Double,ChernozhukovEscancianoIchimuraNewey2022LocalRobust}. DML alleviates these biases by orthogonalizing the moment with respect to $\gamma_0$, coupled with cross-fitting (a form of sample splitting). At the core of the general DML framework in \citet{ChernozhukovEscancianoIchimuraNewey2022LocalRobust,ChernozhukovNeweySingh2022Automatic} is an additional nuisance parameter $\alpha_0$, known as the Riesz representer, which is in general necessary to achieve orthogonality through the so-called first step influence function. As argued in the literature, since $\alpha_0$ may lack a closed-form expression or involve inverting unknowns, it is important to estimate $\alpha_0$ in an automatic way without the knowledge of its analytic form. 

Despite several existing strategies for the automatic estimation of $\alpha_0$ \citep{ChernozhukovNeweyQuintasSyrgkanis2024RieszReg}, the identification of $\alpha_0$, to the best of our knowledge, remains undeveloped at a general level, particularly for first steps defined by models with endogeneity. As our first main contribution, we establish conditions under which $\alpha_0$ is identified by the moment conditions underlying the automatic estimation procedures in \citet{ChernozhukovEscancianoIchimuraNewey2022LocalRobust,ChernozhukovNeweySingh2022Automatic}. The conditions we impose are mild in the sense that they consist of standard smoothness requirements, a linearity condition with respect to $\alpha_0$ that reflects a common feature of the first step influence function, and a {\it coercivity} condition that is not only sufficient but also necessary in a certain sense.  

While the identifying equation may in principle be exploited to directly estimate $\alpha_0$, it is important to obtain a tractable extremal reformulation amenable to penalization as common in machine learning. Our second contribution shows that $\alpha_0$ is identified precisely when it uniquely optimizes a quadratic functional $\mathfrak C_0$. Crucially, the optimization only depends on the knowledge of the functional forms of the original moment restriction and the first step influence function but not the analytic expression of $\alpha_0$, thereby enabling automatic estimation. Our result also reveals that the quadratic characterization recurring in the literature (see, e.g., \citet{ChernozhukovEscancianoIchimuraNewey2022LocalRobust,ChernozhukovNeweySingh2022Automatic,ChernozhukovNeweySinghSyrgkanis2024Adversarial,ChernozhukovNeweyQuintasSyrgkanis2024RieszReg} \citet{Singh2024KRRR}, and \citet{LaanBibautKallusLuedtke2026AutoDML}) arises from a more general level. 

As our third contribution, we develop an automatic estimation procedure of $\alpha_0$ based on the empirical analog of $\mathfrak C_0$, known as Riesz regression, which generalizes \citet{ChernozhukovEscancianoIchimuraNewey2022LocalRobust,ChernozhukovNeweySingh2022Automatic,ChernozhukovNeweyQuintasSyrgkanis2024RieszReg} to settings where $\gamma_0$ may be defined by models with endogeneity, e.g., nonparametric instrumental variable (NPIV) models. The generalization appears nontrivial and poses nontrivial challenges because the parameter space of $\alpha_0$ may be unknown. We derive the convergence rates of the resulting estimator under conditions that allow for penalization and nonlinear sieve approximations (e.g., neural networks). At the technical level, a key ingredient in characterizing the rates is the modulus of continuity of an empirical process, which connects to the critical radius in statistical learning \citep{ChernozhukovNeweySinghSyrgkanis2024Adversarial}.

Our final contribution aims to integrate shape constraints into DML. While machine learning is powerful for processing big data, it is no panacea in delivering statistical guarantees and equally suffers from the curse of dimensionality dictated by the optimal convergence rates \citep{Stone1982RateGlobal,HallHorowitz2005NPIV,ChenReiss2011Rate}. Therefore, in high dimensional settings, additional structures must be exploited to ensure that the convergence rates are fast enough. Following \citet{Stone1985Additive,Stone1994Spline}, some recent studies in deep learning exploit ``low dimensional'' structures such as hierarchical interaction \citep{Schmidt2020ReLU,KohlerLanger2021RateDNN}. These structures, however, may not always be easy to motivate. Shape constraints, on the other hand, are deeply grounded in economics because extensive economic theories are formulated as shape constraints \citep{Matzkin1994Handbook,ChetverikovSantosAzeem2018Shape}. Thus, they may be viewed as alternative structures that help improve estimation precision. In developing our general theory, we operate under the restriction $\gamma_0\in\Gamma$ for some possibly nonlinear parameter space $\Gamma$ which in turn restricts $\alpha_0$. Our identification, characterization, and estimation results for $\alpha_0$ thus allow for general nonlinear shape constraints, though we stress that these results appear novel even without shape constraints.

We complement our contributions on Riesz regression and shape constraints by simulation studies and empirical applications. We focus on implementation via deep learning, not only to maintain a coherent treatment but also because neural networks provide a relatively flexible and tractable way to incorporate shape constraints. Our simulation designs include endogenous as well as exogenous first steps. In both cases, incorporating shape constraints overall improves the performance of point estimates and confidence intervals, with larger gains as more constraints are imposed. Our first application is a difference-in-differences (DiD) analysis of Medicaid expansions and mortality, where monotonicity is imposed based on institutional knowledge well documented in the literature. The second application studies wage growth and working hours, imposing convexity as an implication of economic theory. 

% Arguments for incorporating shape constraints: p.50 in \citet{IchimuraNewey2022IF}.

% [ADD DISCUSSIONS ON NUMERICAL FINDINGS; ADJUST THE WHOLE INTRODUCTION AND THE ABSTRACT ACCORDINGLY]
% [ADJUST THE WHOLE INTRODUCTION AND THE ABSTRACT ACCORDINGLY BASED ON PARTIAL ID EXAMPLE]

This paper contributes to the extensive and rapidly evolving literature on debiased machine learning. Building on the semiparametric literature \citep{Newey1994AsymptoticVar,IchimuraNewey2022IF}, \citet{ChernozhukovEscancianoIchimuraNewey2022LocalRobust,ChernozhukovNeweySingh2022Automatic} developed a general framework for automatic DML while leaving the identification of $\alpha_0$ for future work. We thus complement their work by laying the identification foundation. The automatic estimation of $\alpha_0$ in \citet{ChernozhukovEscancianoIchimuraNewey2022LocalRobust,ChernozhukovNeweySingh2022Automatic} is based on a LASSO minimum distance criterion where $\gamma_0$ is defined by a ``generalized linear regression'' without endogeneity and subject to linear constraints. \citet{ChernozhukovNeweyQuintasSyrgkanis2024RieszReg} obtained the quadratic characterization of $\alpha_0$ for the same class of first steps and developed the Riesz regression that facilitates the use of nonlinear machine learners---see also \citet{ChernozhukovNeweySingh2022GlobalLocalDML}, \citet{ChernozhukovNeweySinghSyrgkanis2024Adversarial}, \citet{Singh2024KRRR}, and \citet{LaanBibautKallusLuedtke2026AutoDML} for constructions in related settings. Our quadratic characterization is more general in that $\theta_0$ may be defined by nonlinear moment conditions while $\gamma_0$ may be defined by models with endogeneity subject to general nonlinear shape constraints.

There is relatively less work on automatic DML for endogenous first steps, even under linear constraints. While the setup in \citet{ChernozhukovEscancianoIchimuraNewey2022LocalRobust} does accommodate endogenous $\gamma_0$, they primarily focus on exogenous $\gamma_0$ after showing how first step influence functions can be used to construct orthogonal moment conditions. We build on their work by taking the orthogonality as given. \citet{Bakhitov2026PGMM} developed a LASSO minimum distance estimator of $\alpha_0$ for functionals of the NPIV regression. Existing simulation evidence, however, suggests that allowing for nonlinear machine learners can materially affect inference; see, e.g., \citet{ChernozhukovChetverikovDemirerDufloHansenNeweyRobins2018Double} and \citet{AhrensChernozhukovHansenKozburSchafferWiemann2026DML}. \citet{BennettKallusMaoNeweySyrgkanisUehara2025StrongID} studied linear functionals of first steps that may be partially identified by linear conditional moment restrictions under a ``strong identification'' condition. When specialized to this setup, our general theory implies that the strong identification condition is stronger than necessary for the identification of $\alpha_0$---see also Appendix \ref{Sec: NPIV extensions} for more comparisons. \citet{Bruns2025TSML} developed a two-stage machine learning estimator of the NPIV regression by adaptively learning the instrument basis.  Our identification, characterization, and estimation results appear novel and more general while also allowing for nonlinear constraints on $\gamma_0$. 

% Based on their characterization of $\alpha_0$, \citet{BennettKallusMaoNeweySyrgkanisUehara2025StrongID} designed a minimax estimator of $\alpha_0$, which may be computationally intensive to implement. 

Our work is also related to the literature on shape constraints which have long played prominent roles in economics. Statistical studies on shape constraints emerged in the 1950s \citep{Hildreth1954Concave,Grenander1956II,BrunkEwingUtz1957Minimize} and have since developed into an influential area of research \citep{Matzkin1994Handbook,ChetverikovSantosAzeem2018Shape}. Our use of shape restrictions is driven by their potential in improving estimation \citep{BlundellHorowitzParey2012Measuring,BlundellHorowitzParey2017Slutsky,ChetverikovWlihelm2017MonNPIV,HorowitzLee2017Shape}. Shape constraints may also be viewed as another form of regularization, as emphasized in \citet{ChetverikovWlihelm2017MonNPIV}. This perspective is particularly helpful in deep learning which often depends on various forms of hyperparameters and regularization and yet a general theory on data-driven tuning is currently unavailable. Our framework brings shape constraints into DML and our simulation results ahead show that shape constraints may improve estimation and inference even with crude tuning. 
 
%  and \citet{Zhu2020Shape}, \citet{FangSeo2019Shape}, \citet{ChernozhukovNeweySantos2023CCMM}, \citet{BreunigChen2024Adaptive}, and references therein for more recent contributions. 
% \footnote{A seminorm $\|\cdot\|_{\mathbf H}\colon \mathbf H\to\mathbf R$ is such that $\|a+b\|_{\mathbf H}\le \|a\|_{\mathbf H}+\|b\|_{\mathbf H}$ and $\|ta\|_{\mathbf H}=|t|\|a\|_{\mathbf H}$ for all $t\in\mathbf R$ and $a,b\in\mathbf H$. If in addition $a=0$ whenever $\|a\|_{\mathbf H}=0$, then $\|\cdot\|_{\mathbf H}$ is a norm.}
 
Finally, we introduce some notation and concepts. For $a,b\in\mathbf R$, set $a\vee b\equiv\max\{a,b\}$ and let $a\lesssim b$ mean $a\le cb$ for some constant $c>0$. For a set $A$ in a vector space $\mathbf H$ with seminorm $\|\cdot\|_{\mathbf H}$, its interior $A^\circ$ is the largest open set contained in $A$, its convex conical hull $\mathrm{con}(A)$ is the smallest convex cone in $\mathbf H$ containing $A$, its linear span $\mathrm{lin}(A)$ is the smallest subspace containing $A$, and its closed linear span  $\overline{\mathrm{lin}}(A)$ is the closure of $\mathrm{lin}(A)$. The distance from $a\in\mathbf H$ to $A$ is $d_{\mathbf H}(a,A)\equiv \inf_{a'\in A}\|a-a'\|_{\mathbf H}$. A map $\phi\colon \Gamma\subset\mathbf H\to\mathbf R$ is Gateaux differentiable at $\gamma_0\in\Gamma$ tangentially to a set $H_0\subset\mathbf H$ if there is a linear map $\nabla_\gamma\phi(\gamma_0)\colon\mathrm{lin}(H_0)\to\mathbf R$ such that $\lim_{t\downarrow 0}\{\phi(\gamma_0+th)-\phi(\gamma_0)\}/t=\nabla_\gamma\phi(\gamma_0)[h]$ whenever $h\in H_0$ and $\gamma_0+th\in\Gamma$ for all small $t\ge 0$. For $X\in\mathcal X$ with law $P$, let $L^2(X)\equiv\{f\colon\mathcal X\to\mathbf R\colon \|f\|_{P,2}<\infty\}$ with $\|f\|_{P,2}\equiv\{E[|f(X)|^2]\}^{1/2}$ and $L^\infty(X)\equiv\{f\colon\mathcal X\to\mathbf R\colon \|f\|_{P,\infty}<\infty\}$ with $\|f\|_{P,\infty}\equiv\inf\{a\ge 0\colon P(|f(X)|>a)=0\}$. For $\{\phi(\cdot,\gamma)\colon\gamma\in\Gamma\}\subset L^2(X)$, we say that $\phi(X,\cdot)$ is mean-square continuous at $\gamma_0$ if the map $\gamma\mapsto\phi(\cdot,\gamma)\in L^2(X)$ is continuous at $\gamma_0$. 

The remainder of the paper is organized as follows. Section \ref{Sec: Setup} introduces the setup and related examples. Section \ref{Sec: Main theory} develops our general framework of automatic DML under shape constraints. Section \ref{Sec: Implementation and Sims} discusses implementation details and presents simulation studies, while Section \ref{Sec: Applications} showcases empirical applications. Section \ref{Sec: Conclusion} concludes. All proofs and supporting results are relegated to the appendix and supplement.
 
\section{The Setup and Examples}\label{Sec: Setup}

As in the literature \citep{Newey1994AsymptoticVar,ChernozhukovEscancianoIchimuraNewey2022LocalRobust}, we work with a generalized method of moments (GMM) model that accommodates numerous causal and structural parameters. Let $X\in\mathcal X$ be a vector of observables with $\mathcal X$ a sample space and $\theta_0\in\Theta\subset\mathbf R^{d_\theta}$ be a parameter of interest such that
\begin{align}\label{Eqn: GMM}
E[g(X,\theta,\gamma_0)]=0
\end{align}
admits a unique solution at $\theta=\theta_0$, where $g\colon \mathcal X\times \Theta\times\Gamma^\dag\to\mathbf R^{d_g}$ is a known map with $d_g\ge d_\theta$ and $\gamma_0\in\Gamma^\dag\subset \mathbf H$ is a nuisance parameter in a vector space $\mathbf H$. The generality of our setup also stems from the flexibility in specifying $\gamma_0$ and its parameter space $\mathbf H$. First, we do not take a stand on the nature of $\gamma_0$ as it may be parametric or nonparametric, causal or structural, or a vector of nuisance parameters. Second, $\mathbf H$ may be equipped with a strong or weak topology depending on the application. There are settings where it is more appropriate to work with a (weak) seminorm than a (strong) norm,  and vice versa---see Remark \ref{Rem: weak norm and partial ID} for details.

Our general theory extends the existing literature \citep{ChernozhukovEscancianoIchimuraNewey2022LocalRobust,ChernozhukovNeweySingh2022Automatic,ChernozhukovNeweyQuintasSyrgkanis2024RieszReg} by accommodating two important additional features of the setup in \eqref{Eqn: GMM}. First, the first step $\gamma_0$ entering the nonlinear moment conditions \eqref{Eqn: GMM} may be defined by a model with endogeneity; e.g., $\gamma_0$ may be an NPIV regression. Such endogeneity poses substantial challenges for identifying and estimating the debiasing nuisance $\alpha_0$, because the relevant parameter space for $\alpha_0$ is then unknown. Second, $\gamma_0$ may be subject to nonlinear shape constraints, which is encoded by restricting $\gamma_0\in\Gamma\subset\Gamma^\dag$. For example, if $\gamma_0$ is monotone, then $\Gamma$ may be taken to be the set of monotone functions in $\Gamma^\dag$. These shape constraints, in turn, induce corresponding restrictions on $\alpha_0$, as our theory below makes clear. Although a few studies allow for endogenous first steps in special settings (see the introduction), they focus on estimation. By contrast, we study not only the estimation of $\alpha_0$ but also its identification and characterization, and we do so in a more general framework in the sense that $\theta_0$ is defined implicitly by a general GMM model \eqref{Eqn: GMM} and $\gamma_0$ is a generic first step subject to possibly nonlinear shape constraints. To our knowledge, the incorporation of general nonlinear shape constraints appears to be new to the DML literature. We also note that our identification result below for $\alpha_0$ appears novel even for exogenous first steps. 

We follow the blueprint for automatic DML laid out in \citet{ChernozhukovEscancianoIchimuraNewey2022LocalRobust}, which is based on the notion of a first step influence function. Intuitively, this function captures the effect of first step estimation errors on the moment condition in \eqref{Eqn: GMM}. Since it plays a fundamental role in our analysis, we next provide a brief review.

\subsection{The First Step Influence Function}\label{Sec: FSIF}

First step estimation errors may be viewed as finite sample counterparts of local perturbations in $\gamma_0$ at the population level. In order to describe these perturbations, we assume that $X$ is distributed according to $P\in\mathbf P$ where $\mathbf P$ is a collection of distributions that possibly generate $X$ while identifying $\theta_0$ via the model \eqref{Eqn: GMM}. Let $P_t=(1-t)P+tQ$ for a distribution $Q$ such that $P_t\in\mathbf P$ and the associated first step, denoted $\gamma (P_t)$, belongs to $\Gamma^\dag$ for each $t\in[0,1]$. Then the first step influence function is a function $\phi\colon\mathcal X\times\Theta\times\Gamma^\dag\times\mathbf A\to\mathbf R^{d_g}$ with $\mathbf A$ a suitable space that satisfies, for some $\alpha_0\in\mathbf A$ and all $Q$ such that each $P_t\in\mathbf P$ and $\gamma (P_t)\in\Gamma^\dag$, the restrictions (i) all entries of $\phi(\cdot,\theta_0,\gamma_0,\alpha_0)$ are in $L^2(X)$ and (ii)
\begin{align}\label{Eqn: FSIF}
\frac{\mathrm d}{\mathrm dt}E[g(X,\theta_0,\gamma(P_t))]\Big|_{t=0+}= \int\phi(x,\theta_0,\gamma_0,\alpha_0)\,Q(\mathrm dx)~,
\end{align}
where $|_{t=0+}$ indicates the right derivative at $0$. Note that \eqref{Eqn: FSIF} implies $E[\phi(X,\theta_0,\gamma_0,\alpha_0)]=0$. The presence of an additional nuisance $\alpha_0$ is a common feature in DML. 

The parameter $\alpha_0$ plays a key role in DML because it helps orthogonalize the influences of the nuisance parameters through the first step influence function, without compromising the identification of $\theta_0$. Specifically, the modified moment condition 
\begin{align}\label{Eqn: moment function augmented}
E[g(X,\theta,\gamma_0)+\phi(X,\theta,\gamma_0,\alpha_0) ]=0
\end{align}
continues to uniquely identify $\theta_0$ and is insensitive to small changes in both $\gamma_0$ and $\alpha_0$. In particular, \citet{ChernozhukovEscancianoIchimuraNewey2022LocalRobust} establish under regularity conditions the following local orthogonality condition with respect to $\gamma_0$:
\begin{align}\label{Eqn: Local Orthogonal wtr gamma}
\frac{\mathrm d}{\mathrm dt}E[g(X,\theta_0,\gamma_0+t\delta)+\phi(X,\theta_0,\gamma_0+t\delta,\alpha_0)]\Big|_{t=0+}=0
\end{align}
with $\delta$ ranging over $\mathbf H$ such that $\gamma_0+t\delta\in\Gamma^\dag$. A stronger orthogonality condition holds for $\alpha_0$ but it is \eqref{Eqn: Local Orthogonal wtr gamma} that matters for our discussions below. 

We do not aim to derive first step influence functions but instead take them as given by leveraging existing work; see, e.g., \citet{ChernozhukovEscancianoIchimuraNewey2022LocalRobust}, \citet{IchimuraNewey2022IF}, and references therein. We stress that, in line with \citet{ChernozhukovEscancianoIchimuraNewey2022LocalRobust}, we view first step influence functions as means of correcting first step biases rather than as a route to deriving efficient influence functions. Indeed, the asymptotic theory in Section \ref{Sec: AN} merely relies on the orthogonality---see the discussions of Assumption \ref{Ass: FSIF}. A common structure of first step influence functions is that they are linear in $\alpha_0$, which we illustrate through examples below.

\subsection{Examples}\label{Sec: Examples}

To fix ideas, we introduce examples that play predominant roles in empirical work. Importantly, the first step influence functions in these settings are all linear in $\alpha$.  While we focus on these examples for conciseness, we note that such linearity also appears in related settings such as moment conditions with density as first steps \citep{Newey1994AsymptoticVar}, dynamic discrete choice models \citep{ChernozhukovEscancianoIchimuraNewey2022LocalRobust}, inference on support functions in a class of partially identified models \citep{Semenova2023SetDML}, and models with generated regressors \citep{EscancianoPerez2025DML}. % endogenous orthgonality conditions under misspecification \citep{IchimuraNewey2022IF}, 

The first example concerns a parameter $\theta_0$ that depends on a first step $\gamma_0$ defined by exogenous orthogonal conditions. % \citep{ChernozhukovNeweySingh2022Automatic,ChernozhukovNeweyQuintasSyrgkanis2024RieszReg}.

\begin{ex}(Average effects)\label{Ex: average linear effects}
Let $Y\in\mathbf R$ be an outcome variable and $Z\in \mathbf R^{d_z}$ a vector of covariates. The first step $\gamma_0$ is defined by:  for $X= (Y,Z)$,
\begin{align}\label{Eqn: average linear effects gamma}
E[\rho(X,\gamma_0) h(Z)]=0 \text{ for all } h\in \Gamma^\dag\subset\mathbf H= L^2(Z)~,
\end{align}
where $\rho(X,\gamma_0)$ is a generalized residual such as $Y-\gamma_0(Z)$ for the mean regression and $u-1\{Y< \gamma_0(Z)\}$ for the quantile regression with $u\in(0,1)$, and $\Gamma^\dag$ is a closed subspace of $L^2(Z)$. The parameter of interest is: for some known $m\colon\mathcal X\times \Gamma^\dag\to\mathbf R$,
\begin{align}\label{Eqn: average linear effects}
\theta_0=E[m(X,\gamma_0)]~.
\end{align}
Special cases of \eqref{Eqn: average linear effects} include various treatment effects \citep{Imbens2004TreatReview}, policy effects \citep{ChernozhukovNeweySingh2022Automatic}, and marginal effects \citep{Stoker1986Scaled}. For $g(X,\theta,\gamma)=m(X,\gamma)-\theta$, the first step influence function is given by:
\begin{align}\label{Eqn: average linear effects, FSIF}
\phi(X,\theta,\gamma,\alpha)= \alpha(Z) \rho(X,\gamma)~,
\end{align}
for all $\gamma\in\Gamma^\dag$ and $\alpha\in\mathbf A\equiv\Gamma^\dag$ \citep{IchimuraNewey2022IF}. \qed
\end{ex}

Our second example is a special case of Example \ref{Ex: average linear effects}, which we single out due to its importance in the study of welfare analysis \citep{HausmanNewey2017Welfare}.

% \citet{FanHsuLieliZhang2022CATE}: Estimation of CATE with high dimensional controls
% \citet{XuOtsu2022Isotonic} and \cite{Xu2022Isotonic}:
% Empirical papers using MTR: Empirical papers: \citet{Gonzalez2005Language}, \citet{KreiderHill2009Uninsured}, \citet{AngristBettingerKremer2006Vouchers}, \citet{Haan2011Schooling}, \citet{ChenFloresFlores2018JCT}, \citet{JunLee2023Persuasion}, \citet{KongDubeDaljord2024Demand}
% Empirical papers using MIV: \citet{BlundellGoslingIchimuraMeghir2007Wages}, \citet{KreiderPepper2007Disability}. See also \citet{MourifieHenryMeango2020Roy}.

\begin{ex}(Average welfare)\label{Ex: welfare}
Empirical welfare analysis of consumer demand often relies on bounds to account for individual heterogeneity \citep{HausmanNewey2016Individual}. Let $Y$ be the expenditure share of a continuous good, $T$ its price, $S$ income, and $Q$ a vector of covariates including prices of other goods. Thus, in the notation of Example \ref{Ex: average linear effects}, $Z=(T,S,Q)$, $X=(Y,Z)$, and $\gamma_0(Z)=E[Y|Z]$. If $b$ is a known lower (upper) bound on the income effect, then the average equivalent variation as $T$ varies from $t_0$ to $t_1$ may be bounded above (below) by \eqref{Eqn: average linear effects} where
\begin{align}
m(x,\gamma) = \varpi(s,q)\int_{t_0}^{t_1} \frac{s}{u}\gamma(u,s,q) \exp\{-b(u-t_0)\}\,\mathrm du
\end{align}
for $x=(y,t,s,q)$ and $\varpi(s,q)$ a known weight function. \qed
\end{ex}

Our third example generalizes Example \ref{Ex: average linear effects} to settings where the first step is defined by endogenous orthogonality conditions.

\begin{ex}[Endogenous first steps]\label{Ex: NPIV}
Let $Y\in\mathbf R$ be an outcome, $Z\in \mathbf R^{d_z}$ a vector of covariates, and $W\in\mathbf R^{d_w}$ a vector of instruments, and set $X=(Y,Z,W)$. For a closed subspace $\Gamma^\dag\subset \mathbf H=L^2(Z)$  and  $\mathbf A\subset L^2(W)$, consider the parameter in \eqref{Eqn: average linear effects} with $\gamma_0\in\Gamma^\dag$ characterized by: for $\rho(X,\gamma_0)$ a generalized residual and all $b\in \mathbf A$,
\begin{align}\label{Eqn: NPIV}
E[\rho(X,\gamma_0) b(W)]=0~.
\end{align}
Special cases of \eqref{Eqn: NPIV} include nonparametric instrumental variable models (NPIV) and other conditional moment restriction models \citep{NeweyPowell2003NPIV,ChernozhukovHansen2005IVQR,ChenPouzo2012EstimationNonsmooth}. The first step influence function is
\begin{align}\label{Eqn: NPIVinfluence}
\phi(X,\theta,\gamma,\alpha)=\alpha(W) \rho(X,\gamma)~,
\end{align}
for $\gamma\in\Gamma^\dag$ and $\alpha\in\mathbf A$ \citep{IchimuraNewey2022IF}. Here, we embed $\gamma_0$ into the space $\mathbf H=L^2(Z)$ which is defined by the norm $\|\cdot\|_{P_Z,2}$ with $P_Z$ the distribution of $Z$. In certain ill-posed problems, however, the convergence rate for estimating $\gamma_0$ may be too slow under $\|\cdot\|_{P_Z,2}$. Nonetheless, it is possible to derive a fast enough convergence rate under a weaker seminorm; see Remark \ref{Rem: weak norm and partial ID} where we also emphasize that our framework allows the first step model to be partially identified. \qed
\end{ex}

\begin{rem}\label{Rem: weak norm and partial ID}
Conditional moment restrictions such as NPIV models often feature a first step $\gamma_0$ that depends on endogenous variables. In these settings, it can be difficult to obtain sufficiently fast convergence rates under a strong norm such as $\|\cdot\|_{P_Z,2}$. This may occur if the model is severely ill-posed \citep{ChenPouzo2009EfficientNonsmooth,ChenPouzo2015SieveWald}. Nonetheless, fast rates may still be achievable under a weaker seminorm when the generalized residual is smooth \citep{AiChen2003Efficient}. A second, conceptually distinct complication is that identification of $\gamma_0$ in these models typically relies on assumptions that are restrictive and may not be testable \citep{Santos2012NPIV,CanaySantosShaikh2013Testability}. Consistent with \citet{Santos2011Recover}, \citet{SeveriniTripathi2012Efficiency}, and \citet{Chen2021RobustPLM}, our framework permits debiased machine learning of an identified linear functional $\theta_0$ of $\gamma_0$ without requiring the point identification of $\gamma_0$. In doing so, we complement \citet{BennettKallusMaoNeweySyrgkanisUehara2025StrongID} by accommodating general shape constraints on $\gamma_0$ and by studying identification, characterization, and estimation of $\alpha_0$ without imposing their strong identification condition---see Appendix \ref{Sec: NPIV extensions} for more detailed discussions and comparisons. \qed 
\end{rem}

\section{The General Framework}\label{Sec: Main theory}

We develop our general theory in three steps. In Section \ref{Sec: Riesz ID}, we establish the identification of $\alpha_0$ when the nuisance parameter $\gamma_0$ is subject to possibly nonlinear shape constraints in addition to those already incorporated in $\Gamma^\dag$, and then characterize $\alpha_0$ as the unique solution to a quadratic optimization problem. In Section \ref{Sec: Riesz regression}, we study automatic estimation of $\alpha_0$ based on the quadratic characterization. In Section \ref{Sec: AN}, we obtain the asymptotic distribution of the debiased GMM estimator $\hat\theta_n$. 

\subsection{Identification and Characterization}\label{Sec: Riesz ID}

The parameter $\alpha_0$ originates from deriving the first step influence function in \eqref{Eqn: FSIF}. As mentioned in the introduction, $\alpha_0$ may not admit a closed-form expression, especially under shape constraints on $\gamma_0$. When it does have a closed-form expression, it may involve inverting unknown objects (e.g., densities) that lead to estimators with poor finite sample performance. These concerns motivate automatic estimators, which tend to perform better, as shown in, e.g., \citet{ChernozhukovNeweyQuintasSyrgkanis2022RieszNet} and \citet{AhrensChernozhukovHansenKozburSchafferWiemann2026DML}. Thus, it is important to estimate $\alpha_0$ in an automatic way without requiring knowledge of its analytical form. A general strategy for the automatic estimation of $\alpha_0$ is proposed by \citet{ChernozhukovEscancianoIchimuraNewey2022LocalRobust} who developed a LASSO minimum distance estimator of $\alpha_0$ based on sample analogs of the moment conditions in \eqref{Eqn: Local Orthogonal wtr gamma}. This, however, raises an important question: Does the system \eqref{Eqn: Local Orthogonal wtr gamma} admit a unique solution?

To address this question, we must choose the set of feasible perturbations that reflect the shape restriction(s) incorporated in $\Gamma$. Our theory in Section \ref{Sec: AN} suggests requiring \eqref{Eqn: Local Orthogonal wtr gamma} to hold for all $\delta\in \Gamma-\gamma_0\equiv\{\gamma-\gamma_0\colon \gamma\in\Gamma\}$---see discussions of Assumption \ref{Ass: FSIF}(ii) in Section \ref{Sec: AN}. To this end, we assume throughout that $\Gamma^\dag$ is linear or at least convex so that $\gamma_0+t\delta\in\Gamma^\dag$ for all $t\in[0,1]$ and $\delta\in \Gamma-\gamma_0$, though no such requirements are imposed on $\Gamma$. Since $\alpha_0$ is defined for each moment equation, below we focus on a single moment equation by pretending that $g$ is real-valued (i.e., $d_g=1$). To formalize our discussions, we introduce the first assumption.

\begin{ass}\label{Ass: Riesz}
(i) $E[g(X,\theta_0,\cdot)]$ is Gateaux differentiable at $\gamma_0\in\Gamma\subset\mathbf H$ tangentially to $\Gamma-\gamma_0$ with $\Gamma\subset\mathbf H$ for a vector space $\mathbf H$; (ii) there exist a linear $\Pi_0\colon \mathbf H_0\to\mathbf A$ and a continuous linear $\Psi_0\colon \mathbf A_0\to\mathbf R$ such that $\nabla_\gamma E[g(X,\theta_0,\gamma_0)][\delta]=\Psi_0(\Pi_0(\delta))$ for all $\delta\in \Gamma-\gamma_0$ where $\mathbf H_0\equiv \mathrm{lin}(\Gamma-\gamma_0)$ and $\mathbf A_0$ is the closure of $\Pi_0(\mathbf H_0)$ in a Hilbert space $\mathbf A$ with norm $\|\cdot\|_{\mathbf A}$; (iii) $E[\phi(X,\theta_0,\cdot,\alpha)]$ is Gateaux differentiable at $\gamma_0$ tangentially to $\Gamma-\gamma_0$ for all $\alpha\in\mathbf A_0$; (iv) there exists a continuous bilinear $\Upsilon_0\colon \mathbf A_0\times\mathbf A_0\to\mathbf R$ such that  $\nabla_\gamma E[\phi(X,\theta_0,\gamma_0,\alpha)][\delta]=\Upsilon_0(\alpha,\Pi_0(\delta))$ for all $\delta\in \Gamma-\gamma_0$ and all $\alpha\in\mathbf A_0$.
\end{ass}

Assumptions \ref{Ass: Riesz}(i)(iii) simply impose minimal smoothness on the original moment function and the first step influence function so that \eqref{Eqn: Local Orthogonal wtr gamma} is well-defined, though we note that Fr\'{e}chet differentiability is required when establishing the asymptotic distribution of the debiased GMM estimator $\hat\theta_n$.  Assumptions \ref{Ass: Riesz}(ii)(iv) demand that the resulting derivatives possess certain composition structures. In particular, $\Pi_0$ serves as the link when $\gamma_0$ and $\alpha_0$ are functions of different variables. This is motivated by a necessary condition for $\sqrt n$-estimability of $\theta_0$---see \citet{SeveriniTripathi2012Efficiency} and also \citet{IchimuraNewey2022IF}. In settings such as Examples \ref{Ex: average linear effects} and \ref{Ex: welfare} where $\gamma_0$ and $\alpha_0$ are functions of the same variable, one may simply set $\Pi_0$ to be the identity map. The linearity of $\Upsilon_0$ with respect to $\alpha$ in Assumption \ref{Ass: Riesz}(iv) is inherited from the same property of first step influence functions in Section \ref{Sec: Examples}.

Our first main result below formally establishes under Assumptions \ref{Ass: Riesz} that \eqref{Eqn: Local Orthogonal wtr gamma} admits a unique solution $\alpha_0$ under an additional condition on $\Upsilon_0$. Following \citet{ChernozhukovNeweySingh2022Automatic}, we refer to this unique solution as the Riesz representer.

\begin{thm}\label{Thm: Riesz ID}
If Assumption \ref{Ass: Riesz} holds, then there is a unique $\alpha_0\in\mathbf A_0$ satisfying \eqref{Eqn: Local Orthogonal wtr gamma} for all $\delta\in \Gamma-\gamma_0$ provided $\Upsilon_0$ is coercive, i.e., for some $\varkappa>0$ and all $\alpha\in\mathbf A_0$,
\begin{align}\label{Eqn: Coercivity}
|\Upsilon_0(\alpha,\alpha)|\ge \varkappa\|\alpha\|_{\mathbf A}^2~.
\end{align}
\end{thm} % Remark 8 in \citet[p.140]{Brezis2011Sovolev} and \citet[pp.13-4]{ErnGuermond2021FiniteII} provide good intuitions for the Lax-Milgram theorem.

\iffalse
\begin{tikzpicture}[
  node distance=2.5cm and 3cm,
  every node/.style={anchor=center},
  every path/.style={>=Stealth},
  roundnode/.style={draw, minimum height=1.2em, inner sep=3pt}
]

% Top row
\node[roundnode] (n1) {$\gamma_0\in\Gamma\subset\mathbf H$};
\node[roundnode, right=of n1] (n2) {$\Gamma-\gamma_0$};
\node[roundnode, right=of n2] (n3) {$\mathrm{lin}(\Gamma-\gamma_0)$};

% Second row
\node[roundnode, below=of n1] (n6) {$\phi(\cdot,\theta_0,\gamma_0,\alpha_0)$};
\node[roundnode, below=of n2] (n5) {$\alpha_0\in\mathbf A_0$};
\node[roundnode, below=of n3] (n4) {$\mathbf A_0\subset\mathbf A$};

% Bottom node below n5
\node[roundnode, below=of n5] (n7) {$J_0$};

% Arrows top row
\draw[->] (n1) -- (n2);
\draw[->] (n2) -- (n3);

% Arrows downward
\draw[->] (n3) -- node[right=2pt] {$\Pi_0$} (n4);
\draw[->] (n6) -- node[above] {$\Psi_0$} node[below] {$\Upsilon_0$} (n5);
\draw[->] (n5) -- (n4);
\draw[->] (n5) -- (n7);

\end{tikzpicture}
\fi

% Theorem \ref{Thm: Riesz ID} complements the automatic DML literature by laying the identification foundation for the Riesz representer at a general level. 

Theorem \ref{Thm: Riesz ID} holds regardless of whether $\Gamma$ is linear or not. We emphasize that the constraint $\gamma_0\in\Gamma$ induces a corresponding constraint on $\alpha_0$ as reflected by $\alpha_0\in\overline{\Pi_0(\mathbf H_0)}$. However, these two constraints are different in general. For example, if $\gamma_0\in\Gamma$ means that $\gamma_0$ is monotone and $\Pi_0$ is the identity map, then the constraint $\alpha_0\in\overline{\Pi_0(\mathbf H_0)}$ means that $\alpha_0$ belongs to $\Gamma-\Gamma$, which differs from $\Gamma$ unless the latter is linear---see Section \ref{Sec: Riesz regression} for more detailed discussions. To appreciate the coercivity condition, suppose that $\mathbf A_0$ is finite dimensional so that we may identify $\Upsilon_0$ with some matrix $\Phi_0$. Then condition \eqref{Eqn: Coercivity} implies but is not implied by the invertibility of $\Phi_0$. In a certain sense, coercivity is not only sufficient but also necessary---see Remark \ref{Rem: coercivity}. Theorem \ref{Thm: Riesz ID} is a consequence of the Lax-Milgram theorem \citep{LaxMilgram1954Parabolic}, which is a foundational tool in partial differential equations. On a technical note, while the Lax-Milgram theorem may be viewed as a generalization of the Riesz representation theorem, the latter is not directly applicable by defining an inner product through $\Upsilon_0$, unless $\Upsilon_0$ is symmetric and has a definite sign (i.e., positive or negative).\footnote{Recall that $\Upsilon_0\colon \mathbf A_0\times\mathbf A_0\to\mathbf R$ is symmetric if $\Upsilon_0(a,b)=\Upsilon_0(b,a)$ for all $a,b\in\mathbf A_0$, positive if $\Upsilon_0(\alpha,\alpha)\ge 0$ for all $\alpha\in\mathbf A_0$, and negative if $\Upsilon_0(\alpha,\alpha)\le 0$ for all $\alpha\in\mathbf A_0$.} These additional properties, however, allow us to obtain an extremal characterization of $\alpha_0$, which we turn to next.

% We note that it is possible to establish the identification of $\alpha_0$ without the linearity of $a\mapsto \Upsilon_0(a,b)$ by employing a nonlinear version of the Lax-Milgram theorem. However, this linearity is essential to obtain our next result. 

Theorem \ref{Thm: Riesz ID} suggests we may estimate $\alpha_0$ directly based on  \eqref{Eqn: Local Orthogonal wtr gamma}. However, it is crucial to obtain an extremal characterization that facilitates estimation. Our next theorem relates $\alpha_0$ to a quadratic optimization problem. 

\begin{thm}\label{Thm: Characterization} 
Let Assumption \ref{Ass: Riesz} hold. If $\Upsilon_0$ is coercive, symmetric, and positive (resp.\ negative), then $\alpha_0$ is the unique solution in $\mathbf A_0$ satisfying \eqref{Eqn: Local Orthogonal wtr gamma} for all $\delta\in \Gamma-\gamma_0$ if and only if it uniquely minimizes (resp.\ maximizes) over $\mathbf A_0$ the functional
\begin{align}\label{Eqn: characterization}
\mathfrak C_0(\alpha)\equiv\frac{1}{2}\Upsilon_0(\alpha,\alpha)+\Psi_0(\alpha)~.
\end{align}
\end{thm}

Theorem \ref{Thm: Characterization} states that the Riesz representer is identified by \eqref{Eqn: Local Orthogonal wtr gamma} as $\delta$ ranges over $\Gamma-\gamma_0$ precisely when it uniquely optimizes over $\mathbf A_0$ the quadratic functional $\mathfrak C_0$. Hence, automatic estimation of $\alpha_0$ based on optimizing $\mathfrak C_0$ over $\mathbf A_0$ is possible because this problem only depends on the moment function (through $\Upsilon_0$, $\Psi_0$ and $\Pi_0$). \citet{ChernozhukovNeweyQuintasSyrgkanis2024RieszReg} and \citet{LaanBibautKallusLuedtke2026AutoDML} obtained essentially the same characterization for $\theta_0$ of the form $E[m(X,\gamma_0)]$ with exogenous $\gamma_0$ under conditions that rule out common shape constraints. Finally, symmetry and positiveness/negativeness of $\Upsilon_0$ play different roles in Theorem \ref{Thm: Characterization}---see Remark \ref{Rem: symmetry}.\footnote{We thank Whitney K.\ Newey for suggesting that we clarify the role played by symmetry.} 

% In related fields such as mathematical physics, $\mathfrak C_0$ is also known as an energy functional. 

\begin{rem}\label{Rem: coercivity}
Assumption \ref{Ass: Riesz} implies that the system of equations \eqref{Eqn: Local Orthogonal wtr gamma} holding for all $\delta\in\Gamma-\gamma_0$ is equivalent to: for all $b\in\mathbf A_0$,
\begin{align}\label{Eqn: Local Orthogonal wtr gamma2}
\Psi_0 (b)+\Upsilon_0(\alpha_0,b)=0~.
\end{align}
Since $\Psi_0$ is unknown, it may be desirable to strengthen the problem of finding $\alpha_0$ that satisfies \eqref{Eqn: Local Orthogonal wtr gamma2} for a single $\Psi_0$ to one of finding $\alpha_0$ that satisfies $\Psi (b)+\Upsilon_0(\alpha_0,b)=0$ for each continuous linear $\Psi$ under suitable restrictions on $\Upsilon_0$. Theorem \ref{Thm: Riesz ID2} in Appendix \ref{Sec: Supporting} shows that if $\Upsilon_0$ is symmetric, positive (or negative), continuous, and bilinear, then such an $\alpha_0$ exists for each continuous linear $\Psi$ if and only if $\Upsilon_0$ is coercive. \qed
\end{rem}

\begin{rem}\label{Rem: symmetry}
Given Assumption \ref{Ass: Riesz}, symmetry of $\Upsilon_0$ is not only sufficient but also necessary for \eqref{Eqn: Local Orthogonal wtr gamma2} to be the first order condition (with respect to $\alpha_0$) of a scalar functional defined on $\mathbf A_0$ known as a potential of $\alpha\mapsto \Psi_0(\cdot)+\Upsilon_0(\alpha,\cdot)$. Indeed, by Example 41.7 in \citet{Zeidler1990NonLinearIII}, such a potential exists if and only if $\Upsilon_0$ is symmetric, in which case the potential is precisely $\mathfrak C_0$ in \eqref{Eqn: characterization} (up to a constant). If $\mathbf A_0$ is finite dimensional so that $(a,b)\mapsto\Upsilon_0(a,b)=a^\intercal H_0b$ for some square matrix $H_0$, then symmetry of $\Upsilon_0$ reduces to symmetry of $H_0$ as the Hessian matrix of $\mathfrak C_0$ in \eqref{Eqn: characterization}. Given coercivity and symmetry, the sign of $\Upsilon_0$ determines whether $\mathfrak C_0$ is strictly convex or strictly concave and, accordingly, whether the optimization problem is minimization or maximization. \qed
\end{rem} % Proposition 25.10 in \citet{Zeidler1990IIB}. 

\subsubsection{Examples Revisited}\label{Sec: Examples revisited}

In this section, we revisit the examples in Section \ref{Sec: Examples} by verifying the assumptions in Theorem \ref{Thm: Riesz ID}. We omit Example \ref{Ex: welfare} as it is a special case of Example \ref{Ex: average linear effects}.

\begin{exctd}[\ref{Ex: average linear effects}]
Consider first the case where $E[m(X,\gamma)]$ is continuous and linear in $\gamma$ and $\gamma_0$ is defined by \eqref{Eqn: average linear effects gamma} with $\rho(X,\gamma_0)=Y-\gamma_0(Z)$. Then, $\Pi_0$ is the identity map so that $\mathbf A_0=\overline{\mathrm{lin}}(\Gamma-\gamma_0)$, $\Psi_0(\delta)=E[m(X,\delta)]$, and $\Upsilon_0(\alpha,\delta)=-E[\alpha(Z)\delta(Z)]$. Clearly, Assumption \ref{Ass: Riesz} is satisfied and hence, there is a unique $\alpha_0\in\mathbf A_0$ that satisfies \eqref{Eqn: Local Orthogonal wtr gamma} for all $\delta\in\Gamma-\gamma_0$ and also uniquely maximizes the functional over $\mathbf A_0$
\begin{align}
\mathfrak C_0(\alpha)=-\frac{1}{2} E[\alpha(Z)^2] + E[m(X,\alpha)]~.
\end{align}
This is effectively the same formulation as given in \citet{ChernozhukovNeweyQuintasSyrgkanis2024RieszReg}. For a nonlinear  $\gamma\mapsto E[m(X,\gamma)]$ and a general (possibly nonsmooth) residual $\rho(X,\gamma_0)$, suppose that $\gamma\mapsto E[m(X,\gamma)]$ and $\gamma\mapsto E[\rho(X,\gamma)|Z]\in L^2(Z)$ are both Gateaux differentiable tangentially to $\Gamma-\gamma_0$ such that $\nabla_\gamma E[m(X,\gamma_0)](\delta)=E[\nu_m(Z)\delta(Z)]$ and $\nabla_\gamma E[\rho(X,\gamma_0)|Z](\delta)=-\nu_\rho(Z)\delta(Z)$ for all $\delta\in \Gamma-\gamma_0$, for some $\nu_m\in\mathbf A_0$, and some $\nu_\rho\in L^\infty(Z)$. If $\nu_\rho\ge 0$, then by Theorem \ref{Thm: Riesz ID} there is a unique $\alpha_0\in\mathbf A_0$ that both satisfies \eqref{Eqn: Local Orthogonal wtr gamma} for all $\delta\in\Gamma-\gamma_0$ and maximizes the functional over $\mathbf A_0$
\begin{align}\label{Eqn: quadratic exgogenous gamma}
\mathfrak C_0(\alpha)= -\frac{1}{2} E[\nu_\rho(Z)\alpha(Z)^2] + E[\nu_m(Z)\alpha(Z)]~.
\end{align}
This is the population criterion implicitly underlying the Riesz regression in \citet{ChernozhukovNeweyQuintasSyrgkanis2024RieszReg} in this more general setup.\qed 
\end{exctd}

\begin{exctd}[\ref{Ex: NPIV}]
Consider first when $E[m(X,\gamma)]$ is continuous and linear in $\gamma\in L^2(Z)$, $\rho(X,\gamma_0)=Y-\gamma_0(Z)$, and $\mathbf A=L^2(W)$. Then $\nabla_\gamma E[m(X,\gamma_0)](\delta)=E[m(X,\delta)]$ and, by the Riesz representation theorem, $E[m(X,\delta)]=E[\nu_m(Z)\delta(Z)]$ for all $\delta\in \Gamma-\gamma_0$ and some $\nu_m\in \overline{\mathrm{lin}}(\Gamma-\gamma_0)$. Suppose there exists some $\bar\nu_m\in L^2(W)$ such that $E[\bar\nu_m(W)|Z]=\nu_m(Z)$, which is necessary for $E[m(X,\gamma_0)]$ to be $\sqrt n$-estimable \citep{SeveriniTripathi2012Efficiency}. Then we have by the law of iterated expectations that
\begin{align}\label{Eqn: NPIV, Psi}
E[m(X,\delta)] = E[E[\bar\nu_m(W)|Z]\delta(Z)]= E[\bar\nu_m(W)E[\delta(Z)|W]]~,
\end{align}
suggesting defining $\Pi_0(\delta)=E[\delta(Z)|W]$ for all $\delta\in\mathbf H_0$ and $\Psi_0(b)=E[\bar\nu_m(W)b(W)]$ for all $b\in\mathbf A_0$. Similarly, we also have that: for any $\delta\in \Gamma-\gamma_0$, 
\begin{align}\label{Eqn: NPIV, Upsilon}
\nabla_\gamma E[\phi(X,\theta_0,\gamma_0,\alpha)][\delta] = -E[\alpha(W)\delta(Z)] = -E[\alpha(W)E[\delta(Z)|W]]~,
\end{align}
which suggests defining $\Upsilon_0(a,b)=-E[a(W)b(W)]$ for any $a,b\in \mathbf A_0$. By Theorem \ref{Thm: Riesz ID}, there is a unique $\alpha_0\in\mathbf A_0$ satisfying \eqref{Eqn: Local Orthogonal wtr gamma} for all $\delta\in\Gamma-\gamma_0$, and this $\alpha_0$ must also uniquely maximize the functional over $\mathbf A_0$ 
\begin{align}\label{Eqn: NPIV, energy}
\mathfrak C_0(\alpha)= -\frac{1}{2} E[\alpha(W)^2] +E[\bar\nu_m(W)\alpha(W)]~.
\end{align}

Turning to the general case, suppose that there is some $\nu_\rho\in L^\infty(X)$ such that 
\begin{align}
\nabla_\gamma E[\alpha(W)\rho(X,\gamma_0)][\delta] = -E[\nu_\rho(X) \alpha(W) \delta(Z)]
\end{align}
for all $\delta\in \Gamma-\gamma_0$, and some $\nu_m\in\overline{\mathrm{lin}}(\Gamma-\gamma_0)$ such that, for all $\delta\in \Gamma-\gamma_0$,
\begin{align}\label{Eqn: NPIV, nu0}
\nabla_\gamma E[m(X,\gamma_0)][\delta]=E[\nu_m(Z)\delta(Z)]~.
\end{align}
Following \citet{IchimuraNewey2022IF}, redefine $\Pi_0(\delta)$ as the projection of $\nu_\rho(X)\delta(Z)$ onto $\mathbf A$ and assume that there exists some $\bar\nu_m(W)\in\mathbf A$ such that $\nu_m(Z)$ equals the projection of $\nu_\rho(X)\bar\nu_m(W)$ onto $ \overline{\mathrm{lin}}(\Gamma-\gamma_0)$. Then exploiting the projection theorem and the definition of $\Pi_0$ we may obtain that, for all $\delta\in \Gamma-\gamma_0$,
\begin{align}\label{Eqn: NPIV Psi general}
\nabla_\gamma E[m(X,\gamma_0)][\delta] = E[\bar\nu_m(W)\Pi_0(\delta)]~.
\end{align}
With the new definition of $\Pi_0$ but the same definitions of $\Psi_0$ and $\Upsilon_0$, we see that the assumptions in Theorem \ref{Thm: Riesz ID} hold and hence
\begin{align}\label{Eqn: NPIV C0 general}
\mathfrak C_0(\alpha)= -\frac{1}{2} E[\alpha(W)^2] + E[\bar\nu_m(W)\alpha(W)]
\end{align}
admits a unique maximizer $\alpha_0\in\mathbf A_0$ that also uniquely satisfies \eqref{Eqn: Local Orthogonal wtr gamma} for all $\delta\in \Gamma-\gamma_0$, where as usual $\mathbf A_0$ is the closure (in $\mathbf A$) of $\{\Pi_0(\delta)\colon \delta\in\mathbf H_0\}$. \qed
\end{exctd}

\subsection{Riesz Regression}\label{Sec: Riesz regression}

In this section, we develop a general procedure for estimating the Riesz representer $\alpha_0$ based on Theorem \ref{Thm: Characterization}. As in the literature, the estimation is automatic in the sense that only the knowledge of the functional forms of the moment function and the first step influence function are required but not the analytic expression of $\alpha_0$ itself. We call the procedure a Riesz regression following \citet{ChernozhukovNeweyQuintasSyrgkanis2024RieszReg}, which is a constrained quadratic optimization problem.

As a first step, we approximate the problem by a ``parametrized'' version. Despite the notation, the linear span $\mathbf H_0\equiv\mathrm{lin}(\Gamma-\gamma_0)$ does not involve $\gamma_0$ because $\mathrm{lin}(\Gamma-\gamma_0)=\mathrm{lin}(\Gamma-\Gamma)$, which further simplifies to $\mathrm{con}(\Gamma)-\mathrm{con}(\Gamma)$ if $0\in\Gamma$, to $\Gamma-\Gamma$ if $\Gamma$ is a convex cone, and to $\Gamma$ if $\Gamma$ is linear; see Lemma \ref{Lem: linear span}. For $\Delta\Gamma_n\subset \mathrm{lin}(\Gamma-\Gamma)$ a sieve space of $\mathrm{lin}(\Gamma-\Gamma)$, we then approximate $\mathbf A_0$ by $\mathbf A_{0,n}\equiv\Pi_0(\Delta\Gamma_n)$. In Section \ref{Sec: Implementation}, we shall provide details of concrete constructions.

In order to obtain fruitful results, we next formalize common structures of the quadratic functional $\mathfrak C_0$ by imposing the following assumption. As in Section \ref{Sec: Riesz ID}, we keep pretending $d_g=1$ in what follows. 

\begin{ass}\label{Ass: Convergence rate, population}
(i) $\Upsilon_0(a,b)\equiv E[l(X,a,b,\nu_0)]$ for all $a,b\in\mathbf A$, some possibly unknown $\nu_0$ in a vector space $\mathbf B$ with seminorm $\|\cdot\|_{\mathbf B}$, and some $l\colon\mathcal X\times\mathbf A\times\mathbf A\times\mathbf B\to\mathbf R$ such that $E[l(X,a,b,\nu)]$ is continuous and trilinear in $(a,b,\nu)\in\mathbf A\times\mathbf A\times\mathbf B$ and symmetric in $(a,b)
\in \mathbf A\times\mathbf A$; (ii) $\Psi_0(a)\equiv E[h(X,a,\nu_0)]$ for all $a\in \mathbf A$  and some $h\colon\mathcal X\times\mathbf A\times\mathbf B\to\mathbf R$ such that $E[h(X,a,\nu)]$ is continuous and bilinear in $(a,\nu)\in\mathbf A\times\mathbf B$; (iii) $E[l(X,a,a,\nu_0)]\le -\varkappa\|a\|_{\mathbf A}^2$ for all $a\in\mathbf A$ and some $\varkappa>0$. %; (iv) $E[l(X,a,a,\nu_0)]\ge -\varpi\|a\|_{\mathbf A}^2$ for all $a\in\mathbf A$ and some $\varpi>0$. 
\end{ass} % If $\Upsilon_0$ and $\Psi_0$ are Every Hilbert–Schmidt operators, then this assumption follows; see, e.g., Proposition 7.10.23 in Bogachev and Smolyanov (2020).

Assumption \ref{Ass: Convergence rate, population} essentially reformulates conditions in Theorem \ref{Thm: Characterization} when the maps $\Upsilon_0$ and $\Psi_0$ take the forms of expectations, with two caveats. First, Assumption \ref{Ass: Convergence rate, population}(iii) requires $\Upsilon_0$ to be negative, but this is without loss of generality and one may instead consider positive $\Upsilon_0$. Second, the conditions on $\Upsilon_0$ and $\Psi_0$ are strengthened to hold over $\mathbf A$ instead of $\mathbf A_0$.  This is not necessary when $\Pi_0$ is known but is needed when $\Pi_0$ and hence $\mathbf A_0$ are unknown and therefore must be estimated. 

To estimate $\alpha_0$, we first consider the case when $\Pi_0$ is known (e.g., the identity map). Given Assumptions \ref{Ass: Convergence rate, population} and a sample $\{X_i\}_{i=1}^n$ of $X$, we define: 
\begin{align}
\hat{\mathfrak C}_n(\alpha) \equiv \frac{1}{n}\sum_{i=1}^{n}\{\frac{1}{2}l(X_i,\alpha,\alpha,\hat\nu_n)+h(X_i,\alpha,\hat\nu_n)\}~,
\end{align}
where $\hat\nu_n\in\mathbf B_n\subset\mathbf B$, and let $\hat\alpha_n\in\mathbf A_{0,n}$ satisfy
\begin{align}\label{Eqn: Riesz estimator}
\hat{\mathfrak C}_n(\hat\alpha_n)-\lambda_n J_n(\hat\alpha_n)= \sup_{\alpha\in\mathbf A_{0,n}}\{\hat{\mathfrak C}_n(\alpha)-\lambda_n J_n(\alpha)\}-O_p(\delta_n^2)~,  
\end{align}
where $\delta_n$ is the order of optimization error, and $J_n\colon\mathbf A_{0,n}\to\mathbf R_+$ is a penalty function with the amount of penalization dictated by the tuning parameter $\lambda_n\ge 0$. One may set $\lambda_n=0$, resulting in an estimator without penalization. We also note that, in certain settings such as NPIV models, the linear term $\Psi_0$ in $\mathfrak C_0$ may be automatically estimated without needing to estimate $\nu_0$---see Remark \ref{Rem: linear term}. 

Establishing statistical guarantees of $\hat\alpha_n$ requires restrictions on the sieve, the penalization, and the sample. For ease of notation, let $\mathbb G_nf\equiv \sum_{i=1}^{n}\{f(X_i)-E[f(X)]\}/\sqrt n$ for a generic function $f$ and $f(X,\alpha,\nu)\equiv l(X,\alpha,\alpha,\nu)/2+h(X,\alpha,\nu)$ so that $\mathbb G_nf(\alpha,\nu)=\sum_{i=1}^{n}\{f(X_i,\alpha,\nu)-E[f(X,\alpha,\nu)]\}/\sqrt n$. For each $\delta>0$ let $U_{0,n}(\delta)$ be the set
\begin{align}
\{(\alpha,\nu)\in \mathbf A_{0,n}\times \mathbf B_n\colon \frac{\delta}{2}\le \|\alpha-\alpha_{0,n}\|_{\mathbf A} + \{\lambda_n J_n(\alpha)\}^{1/2}\le\delta, \varrho_n\|\nu-\nu_0\|_{\mathbf B}\le\delta\}~,
\end{align}
where $\varrho_n\ge 1$ and $\alpha_{0,n}\in\mathbf A_{0,n}$ (to be specified). Thus, $U_{0,n}(\delta)$ may be loosely viewed as a (partially punctured) neighborhood of $(\alpha_{0,n},\nu_0)$. We now impose:

\begin{ass}\label{Ass: Convergence rate, approximation}
(i) $\sup_{\eta\in\Delta\Gamma_n}\|\eta\|_{\mathbf H}\le\varrho_n$ for some $\varrho_n\ge 1$ and a seminorm $\|\cdot\|_{\mathbf H}$ on $\mathbf H$ such that $\Pi_0: \mathbf H_0\to\mathbf A$ is continuous and linear; (ii) $\|\alpha_{0,n}-\alpha_0\|_{\mathbf A}=O(\delta_n)$ for some $\alpha_{0,n}\in\mathbf A_{0,n}\subset\mathbf A_0$ and $\delta_n\to 0$; (iii) $\lambda_n J_n(\alpha_{0,n})=O(\delta_n^2)$ for some $\lambda_n\ge 0$ and $J_n\colon\mathbf A_{0,n}\to\mathbf R_+$; (iv) $E[\sup_{(\alpha,\nu)\in U_{0,n}(\delta)}|\mathbb G_n(f(\alpha,\nu)-f(\alpha_{0,n},\nu))|]\lesssim \omega_n(\delta)$ for each $\delta>0$ and some function $\omega_n\colon (0,\infty)\to\mathbf R_+$ such that $\omega_n(\delta)/\delta^\varsigma$ is nonincreasing in $\delta\in(0,\infty)$ for some $\varsigma<2$; (v) $\omega_n(\delta_n)/\sqrt n\le\delta_n^2$.
\end{ass}

\begin{ass}\label{Ass: Convergence rate, sample}
(i) $\{X_i\}_{i=1}^n$ is an i.i.d.\ sample of $X$; (ii) $\hat\nu_n\in\mathbf B_n\subset\mathbf B$ is an estimator of $\nu_0$ such that $\|\hat\nu_n-\nu_0\|_{\mathbf B}=o_p(1)$. 
\end{ass}

Assumption \ref{Ass: Convergence rate, approximation}(i) implies that $\Delta\Gamma_n$ and hence $\mathbf A_{0,n}$ are bounded by possibly diverging sequences. This is consistent with common choices of sieves (including neural networks). Assumption \ref{Ass: Convergence rate, approximation}(ii) formalizes $\alpha_{0,n}$ as an approximation of $\alpha_0$, which may be taken as the maximizer of $\mathfrak C_0$ over $\mathbf A_{0,n}$ in certain settings; see Lemma \ref{Lem: error bound}. Specific approximation errors are known for a variety of sieve spaces; see, e.g., \citet{Chen2007Handbook} for classical sieves and \citet{Yarotsky2017Error}, \citet{Yang2025ReLU}, \citet{NaglerLanger2026Optimal}, and references therein for deep neural networks.

% ANN: \citet{YarotskyZhevnerchuk2020Phase}
% RKHS: \citet{RasmussenWilliams2006Gaussian}, \citet{BerlinetAgnan2004RKHS}, \citet{ChernozhukovNeweySinghSyrgkanis2024Adversarial}, \citet{ScholkopfSmola2002Kernels}, \citet{SteinwartChristmann2008SVM}, \citet{Aronszajn1950RKHS}

Assumption \ref{Ass: Convergence rate, approximation}(iii) regulates the penalization function $J_n$ and the parameter $\lambda_n$. Assumption \ref{Ass: Convergence rate, approximation}(iv) is known as the modulus of continuity condition \citep{VaartWellner1996Book}, which, intuitively speaking, restricts the ``local size'' of the sieve $\mathbf A_{0,n}$ at $\alpha_{0,n}$. In turn, the modulus $\omega_n$ determines the order $\delta_n$ via Assumption \ref{Ass: Convergence rate, approximation}(v). We note that $\delta_n$ may be viewed as the critical radius, a concept that plays important roles in modern machine learning---see Appendix \ref{Sec: Entropies}. Finally, Assumption \ref{Ass: Convergence rate, sample} introduces the sample and a consistent estimator $\hat\nu_n$ of $\nu_0$. % Our results depend on Assumption \ref{Ass: Convergence rate, sample}

\begin{thm}\label{Thm: Convergence rate}
Let Assumptions \ref{Ass: Convergence rate, population}, \ref{Ass: Convergence rate, approximation}, and \ref{Ass: Convergence rate, sample} hold. If $\hat\alpha_n\colon\{X_i\}_{i=1}^n\to \mathbf A_{0,n}$ satisfies equation \eqref{Eqn: Riesz estimator}, then it follows that
\begin{align}
\|\hat\alpha_n-\alpha_0\|_{\mathbf A}= O_p(\delta_n+\varrho_n\|\hat\nu_n-\nu_0\|_{\mathbf B})~.
\end{align}
\end{thm}

Theorem \ref{Thm: Convergence rate} implies that one may employ a sieve $\mathbf A_{0,n}$ with a diverging bound while still having $\|\hat\alpha_n-\alpha_0\|_{\mathbf A}= O_p(\delta_n)$, as long as $\hat\nu_n$ converges to $\nu_0$ fast enough so that $\|\hat\nu_n-\nu_0\|_{\mathbf B}=O_p(\delta_n/\varrho_n)$. The sieve space $\mathbf A_{0,n}$ determines the convergence rate through its critical radius as well as its approximation capability. The convergence rate of $\hat\alpha_n$ is also impacted by the tuning parameter $\lambda_n$. Developing a data driven choice of $\lambda_n$ is an important issue that is beyond the scope of this paper.

Next, for $\Pi_0$ unknown, let $\mathbf D_n$ be a set of linear maps from $\mathbf H_0$ to $\mathbf A$ and define $\|\Pi\|_{op,n}\equiv \sup_{\eta\in\Delta\Gamma_n\colon\|\eta\|_{\mathbf H}\le 1}\|\Pi(\eta)\|_{\mathbf A}$ for any linear $\Pi\colon \mathbf H_0\to\mathbf A$. Accordingly, we redefine for each $\delta>0$ the set $U_{0,n}(\delta)$ as: for $\eta_{0,n}\in\Delta\Gamma_n$ such that $\alpha_{0,n}=\Pi_0(\eta_{0,n})$,
\begin{multline}\label{Eqn: sieve local size}
\{(\eta,\nu,\Pi)\in \Delta\Gamma_n\times \mathbf B_n\times \mathbf D_n\colon  \frac{\delta}{2}\le \|\Pi\eta-\Pi\eta_{0,n}\|_{\mathbf A} + \{\lambda_n \bar J_n(\eta)\}^{1/2}\le\delta,\\
  \varrho_n\{\|\nu-\nu_0\|_{\mathbf B}\vee \|\Pi-\Pi_0\|_{op,n}\}\le\delta\}~.
\end{multline}
We now introduce our final assumption in this section as follows.

\begin{ass}\label{Ass: Convergence rate, Pi} % $d_{\mathbf A}(\alpha,\mathbf A_{0,n})\to 0$ as $n\to\infty$ for any $\alpha\in\mathbf A_0$
(i) $t\eta\in\Delta\Gamma_n$ for all $\eta\in\Delta\Gamma_n$ and all $t\in(0,1]$; (ii) $\hat\Pi_n\in\mathbf D_n$ satisfies $\varrho_n\|\hat\Pi_n-\Pi_0\|_{op,n}=o_p(1)$; (iii) there is some $\omega_n\colon (0,\infty)\to\mathbf R_+$ satisfying $E[\sup_{(\eta,\nu,\Pi)\in U_{0,n}(\delta)}|\mathbb G_n(f(\Pi(\eta),\nu)-f(\Pi(\eta_{0,n}),\nu))|]\lesssim \omega_n(\delta)$ for each $\delta>0$ such that $\omega_n(\delta)/\delta^\varsigma$ is nonincreasing in $\delta$ for some $\varsigma<2$ and $\omega_n(\delta_n)/\sqrt n\le\delta_n^2$; (iv) either $\hat\Pi_n(\Delta\Gamma_n)\subset\mathbf A_0$ with probability approaching one or $\mathfrak C_0(\alpha)-\mathfrak C_0(\alpha_0) \lesssim \|\alpha-\alpha_0\|_{\mathbf A}^\tau$ for some $\tau>1$ whenever $d_{\mathbf A}(\alpha,\mathbf A_0)\le\epsilon$ for some $\epsilon>0$. 
\end{ass}

Assumption \ref{Ass: Convergence rate, Pi}(i) is a mild simplifying condition so that $\|\cdot\|_{op,n}$ may be defined over the unit ball in $\Delta\Gamma_n$. Assumption \ref{Ass: Convergence rate, Pi}(ii) introduces a $\varrho_n$-consistent estimator $\hat\Pi_n$ of $\Pi_0$. Assumption \ref{Ass: Convergence rate, Pi}(iii) is an analog of Assumption \ref{Ass: Convergence rate, approximation}(iv) that accommodates unknown $\Pi_0$. Assumption \ref{Ass: Convergence rate, Pi}(iv) is needed because our estimator $\hat\alpha_n$ below may lie outside the parameter space $\mathbf A_0$ of $\alpha_0$ (due to the estimator error in $\hat\Pi_n$). The condition $\hat\Pi_n(\Delta\Gamma_n)\subset\mathbf A_0$ with probability approaching one means that we may analyze $\hat\alpha_n$ as if it belongs to $\mathbf A_0$. Alternatively, the condition $\mathfrak C_0(\alpha)-\mathfrak C_0(\alpha_0) \lesssim\|\alpha-\alpha_0\|_{\mathbf A}^\tau$ whenever $d_{\mathbf A}(\alpha,\mathbf A_0)\le\epsilon$ regulates the curvature of $\mathfrak C_0$ over a small enlargement of $\mathbf A_0$ so that the estimator error of $\hat\Pi_n$ is under control. In Example \ref{Ex: NPIV}, this curvature condition automatically holds with $\tau=2$---see Appendix \ref{Sec: curvature} for details.
 
Given the estimator $\hat\Pi_n$, we estimate the sieve $\mathbf A_{0,n}$ by $\hat{\mathbf A}_n\equiv\{\hat\Pi_n(\eta)\colon \eta\in\Delta\Gamma_n\}$. In turn, we define the estimator $\hat\alpha_n\equiv \hat\Pi_n(\hat\eta_n)$ where $\hat\eta_n\in\Delta\Gamma_n$ is such that
\begin{align}\label{Eqn: Riesz estimator2}
\hat{\mathfrak C}_n(\hat\Pi_n(\hat\eta_n))-\lambda_n \bar J_n(\hat\eta_n)= \sup_{\eta\in\Delta\Gamma_n}\{\hat{\mathfrak C}_n(\hat\Pi_n(\eta))-\lambda_n \bar J_n(\eta)\}-O_p(\delta_n^2)~,
\end{align}
where $\bar J_n\colon\Delta\Gamma_n\to\mathbf R_+$ relates to the previous penalty $J_n$ via $J_n(\Pi(\eta))=\bar J_n(\eta)$ for any $\Pi\in\mathbf D_n$ and $\eta\in\Delta\Gamma_n$, i.e., $J_n(\Pi(\eta))$ depends on $\Pi(\eta)$ only through $\eta$ (as is common in machine learning). As previously, the linear term in $\mathfrak C_0$ may be estimated automatically in some settings without requiring an estimator of $\nu_0$---see Remark \ref{Rem: linear term}.

The next theorem obtains the convergence rate of $\hat\alpha_n$ when $\Pi_0$ is unknown.

\begin{thm}\label{Thm: Convergence rate2}
Let Assumptions \ref{Ass: Convergence rate, population}, \ref{Ass: Convergence rate, approximation}(i)(ii)(iii), \ref{Ass: Convergence rate, sample}, and \ref{Ass: Convergence rate, Pi} hold. If $\hat\alpha_n\equiv \hat\Pi_n(\hat\eta_n)$ for $\hat\eta_n\colon\{X_i\}_{i=1}^n\to \Delta\Gamma_n$ that satisfies equation \eqref{Eqn: Riesz estimator2}, then it follows that
\begin{align}\label{Eqn: Convergence rate2 aux0}
\|\hat\alpha_n-\alpha_0\|_{\mathbf A}= O_p(\delta_n+\varrho_n\|\hat\nu_n-\nu_0\|_{\mathbf B}+\varrho_n\|\hat\Pi_n-\Pi_0\|_{op,n})~.
\end{align}
\end{thm}

% Theorem 7.10.23 in \citet{BogachevSmolyanov2020RealFunctional} 

Theorem \ref{Thm: Convergence rate2} shows that the convergence rate of $\hat\alpha_n$ is in addition controlled by the convergence rate of $\hat\Pi_n$. Note that Theorem \ref{Thm: Convergence rate2} includes Theorem \ref{Thm: Convergence rate} as a special case by setting $\mathbf D_n=\{\Pi_0\}$. In Example \ref{Ex: NPIV}, $\Pi_0$ is the conditional expectation operator and so $\hat\Pi_n$ may be constructed by nonparametric regression methods. While the convergence rates of various classical nonparametric regression estimators are well understood, those of modern machine learners are less established. Recent advances on convergence rates of nonparametric regression estimators obtained by deep learning include \citet{Schmidt2020ReLU}, \citet{FarrellLiangMisra2021DNN}, and \citet{KohlerLanger2021RateDNN}. Finally, in certain settings such as NPIV models, the $\hat\nu_n$ term does not appear in the order because $\Upsilon_0$ does not involve $\nu_0$ and $\Psi_0$ may be estimated automatically as explained in Remark \ref{Rem: linear term}. 

%  and \citet{MourtadaGaiffasScornet2020Mondrian}, \citet{Klusowski2021RF}, \citet{ChiVosslerFanLv2022RF}, \citet{KlusowskiTian2024Tree}, \citet{CattaneoChandakKlusowski2024ObliqueTrees}, and \citet{CattaneoKlusowskiUnderwood2025Mondrian} for random forests. 

\begin{rem}\label{Rem: linear term}
Suppose that $\nabla_\gamma E[g(X,\theta_0,\gamma_0)][\delta]=E[\nabla_\gamma g(X,\delta)]$ for all $\delta\in\mathbf H_0$ and some \textit{known} map $\nabla_\gamma g\colon\mathcal X\times\mathbf H_0\to\mathbf R$. Then $\Psi_0(\Pi_0(\delta))=E[\nabla_\gamma g(X,\delta)]$ by definition so that the population analog of \eqref{Eqn: Riesz estimator} and \eqref{Eqn: Riesz estimator2} is
\begin{align}\label{Eqn: Psi rewritten}
  \sup_{\eta\in\Delta\Gamma_n} \frac{1}{2}E[l(X,\Pi_0(\eta),\Pi_0(\eta),\nu_0)] + E[\nabla_\gamma g(X,\eta)]~.
\end{align}
Thus, it suffices to estimate the unknowns involved in the quadratic term only. In the special case when $\theta_0=E[m(X,\gamma_0)]$ is continuous and linear in $\gamma_0$, we have $\nabla_\gamma g=m$. In our simulation studies, we estimate $\alpha_0$ based on \eqref{Eqn: Psi rewritten} in an NPIV setting. \qed
\end{rem}

\subsection{Asymptotic Normality}\label{Sec: AN}

While the asymptotic normality of the debiased GMM estimator $\hat\theta_n$ is a relatively standard result, we present it for completeness. To this end, we impose assumptions similar to but more primitive than those in \citet{ChernozhukovEscancianoIchimuraNewey2022LocalRobust} as we separate analytic conditions from probabilistic ones. Following the literature \citep{ChernozhukovChetverikovDemirerDufloHansenNeweyRobins2018Double}, we combine the orthogonal moment in \eqref{Eqn: moment function augmented} with cross-fitting---the latter in particular helps further alleviate overfitting biases inherited from first step learners. 

Formally, let $\bigcup_{k=1}^K I_k$ be a $K$-fold partition of $\{1,\ldots,n\}$ for some $K>1$ and denote by $n_k$ the sample size of the $k$-th fold. Suppose that $\hat\gamma_{-k}$ and $\hat\alpha_{-k}$ are estimators of $\gamma_0$ and $\alpha_0$ respectively that are based on observations {\it not} from the $k$-th fold. Further assume that $\tilde\theta_{-k}$ is a preliminary estimator of $\theta_0$ based on observations not in the $k$-th fold.  For notational simplicity, define $\psi(X,\theta,\gamma,\alpha) \equiv g(X,\theta,\gamma) +\phi(X,\theta,\gamma,\alpha)$, $\psi_0(X)\equiv \psi(X,\theta_0,\gamma_0,\alpha_0)$, and, for a sample $\{X_i\}_{i=1}^n$,
\begin{align}
\hat\psi_n(\theta)\equiv\frac{1}{n}\sum_{k=1}^{K}\sum_{i\in I_k }\{ g(X_i,\theta,\hat\gamma_{-k}) + \phi(X_i,\tilde\theta_{-k},\hat\gamma_{-k}, \hat\alpha_{-k})\}~.
\end{align}
Since each moment function is associated with a (potentially) different Riesz representer, $\alpha_0$ and $\hat\alpha_{-k}$ in this section are therefore understood as $d_g\times 1$ vectors of functions. For a $d_g\times d_g$ weighting matrix $\hat\Omega_n$, the debiased GMM estimator is then
\begin{align}\label{Eqn: Debiased GMM}
\hat\theta_n\in\arg\min_{\theta\in\Theta} \hat\psi_n(\theta) ^\intercal \hat\Omega_n \hat\psi_n(\theta)~.
\end{align}

% \citet{Hausman1978Specification}, \citet{Rao1973Linear} (p.320, p.322), \citet{Fisher1925Theory}.

Having introduced the notation, we impose the assumptions in this section.

\begin{ass}\label{Ass: Model}
(i) $\theta_0\in\Theta^\circ\subset\mathbf R^{d_\theta}$ and $\gamma_0\in\Gamma\subset\mathbf H$ with a seminorm $\|\cdot\|_{\mathbf H}$; (ii) $g(X,\cdot,\gamma)$ is differentiable over an open neighborhood $U_0$ of $\theta_0$ for all $\gamma$ in a neighborhood $V_0\subset\Gamma$ of $\gamma_0$; (iii) $E[\nabla_\theta g(X,\cdot,\cdot)]$ is continuous at $(\theta_0,\gamma_0)$ with $E[\nabla_\theta g(X,\theta_0,\gamma_0)]$ of full column rank; (iv) $E[\sup_{(\theta,\gamma)\in U_0\times V_0}\|\nabla_\theta g(X,\theta,\gamma)\|^{1+\varsigma}]<\infty$ for some $\varsigma>0$.
\end{ass}

\begin{ass}\label{Ass: FSIF}
(i) $E[\psi_0(X)]=0$ and $E[\|\psi_0(X)\|^2]<\infty$; (ii) $\|E[\psi(X,\theta_0,\gamma,\alpha_0)]\|\lesssim \|\gamma-\gamma_0\|_{\mathbf H}^2$ for all $\gamma\in V_0$; (iii) $E[\phi(X,\theta,\gamma_0,\alpha)]=0$ for all $\theta\in\Theta$ and $\alpha\in\mathbf A_0\subset\mathbf A$ with $(\mathbf A,\|\cdot\|_{\mathbf A})$ a normed space;  (iv) $ \psi(X,\theta_0,\cdot,\alpha_0)$ is mean-square continuous at $\gamma_0$; (v) $\phi(X,\theta,\gamma_0,\alpha)$ is mean-square continuous at $(\theta,\alpha)=(\theta_0,\alpha_0)$; (vi)  $\Delta_0(X,\theta,\gamma,\alpha)\equiv \phi(X,\theta,\gamma,\alpha)-\phi(X,\theta,\gamma_0,\alpha)-\phi(X,\theta_0,\gamma,\alpha_0)+\phi_0(X)$ is mean-square continuous at $(\theta,\gamma,\alpha)=(\theta_0,\gamma_0,\alpha_0)$; (vii) $\|E[\Delta_0(X,\theta,\gamma,\alpha)]\|\lesssim \|\theta-\theta_0\|^2+\|\gamma-\gamma_0\|_{\mathbf H}^2+\|\alpha-\alpha_0\|_{\mathbf A}^2~$ for all $(\theta,\gamma,\alpha)\in U_0\times V_0\times A_0$ with $A_0\subset\mathbf A_0$ a neighborhood of $\alpha_0$.
\end{ass}

\begin{ass}\label{Ass: Data}
(i) $\{X_i\}_{i=1}^n$ is an i.i.d.\ sample of $X$; (ii) $\hat\theta_n\overset{p}{\to}\theta_0$ for $\hat\theta_n\colon\{X_i\}_{i=1}^n\to\Theta$ satisfying \eqref{Eqn: Debiased GMM}; (iii) $\|\tilde\theta_{-k}-\theta_0\|\vee\|\hat\gamma_{-k}-\gamma_0\|_{\mathbf H}\vee \|\hat\alpha_{-k}-\alpha_0\|_{\mathbf A}=o_p(n^{-1/4})$ as $n\to\infty$ with $\hat\gamma_{-k}\in\Gamma$ and $\hat\alpha_{-k}\in\mathbf A$ for all $k=1,\ldots,K$; (iv) $\min_{k=1}^K n_k\to\infty$ as $n\to\infty$; (v) $\hat\Omega_n\overset{p}{\to}\Omega_0$ for some positive definite  matrix $\Omega_0$ of size $d_g\times d_g$.
\end{ass}

\begin{ass}\label{Ass: variance estimation}
(i) $E[\psi(X,\theta,\gamma,\alpha)\psi(X,\theta,\gamma,\alpha)^\intercal]$ is continuous at $(\theta_0,\gamma_0,\alpha_0)$; (ii) $E[\sup_{(\theta,\gamma,\alpha)\in U_0\times V_0\times A_0}\|\psi(X,\theta,\gamma,\alpha)\|^{2+\varsigma}]<\infty$; (iii) $\|\nabla_\theta g(X,\theta,\gamma)-\nabla_\theta g(X,\theta_0,\gamma)\|\le  b(X,\gamma)\|\theta-\theta_0\|$ for all $(\theta,\gamma)\in U_0\times V_0$ with $E[\sup_{\gamma\in V_0} b(X,\gamma)]<\infty$.
\end{ass}

Assumption \ref{Ass: Model} places standard restrictions on the GMM model. Assumption \ref{Ass: Model}(iv) in particular is a sufficient condition to deal with the triangular array nature of each fold $I_k$. Assumption \ref{Ass: FSIF} imposes conditions on the augmented moment function. Assumption \ref{Ass: FSIF}(i) is a minimal requirement. Assumption \ref{Ass: FSIF}(ii) entails the orthogonality of $E[\psi(X,\theta_0,\cdot,\alpha_0)]$ at $\gamma_0$, i.e., $\nabla_\gamma E[\psi(X,\theta_0,\gamma_0,\alpha_0)][\gamma-\gamma_0]=0$ for all $\gamma\in V_0$. Given $E[\psi_0(X)]=0$ and the orthogonality, Assumption \ref{Ass: FSIF}(ii) holds if $E[\psi(X,\theta_0,\cdot,\alpha_0)]$ is twice continuously Fr\'{e}chet differentiable around $\gamma_0$. Assumption \ref{Ass: FSIF}(iii) is a global orthogonality property of the first step influence function \citep{ChernozhukovEscancianoIchimuraNewey2022LocalRobust}. The continuity in Assumptions \ref{Ass: FSIF}(iv)(v)(vi) follows if $g$ and $\phi$ are mean-square continuous at the truth. Assumption \ref{Ass: FSIF}(vii) is mild since $E[\Delta_0(X,\theta,\gamma,\alpha)]$ may be viewed as a discrete version of the second derivative of $E[\phi(X,\cdot,\cdot,\cdot)]$ at $(\theta_0,\gamma_0,\alpha_0)$, with the discrepancy of order $o(\|\theta-\theta_0\|^2+\|\gamma-\gamma_0\|_{\mathbf H}^2+\|\alpha-\alpha_0\|_{\mathbf A}^2)$. 

% ''the small bias property'' as in \citet{NeweyHsiehRobins1998Undersmoothing} and \citet{NeweyHsiehRobins2004Twicing}
% See Proposition 7.2.1 in \citet{Cartan1971Calculus} and Theorem 9.40 in \citet{Rudin1976Principle}.

Assumption \ref{Ass: Data} formalizes restrictions on the sample and relevant estimators. Assumption \ref{Ass: Data}(i) is standard, while Assumption \ref{Ass: Data}(ii) is a high level condition imposed for the sake of transparency of the theory---see Lemma \ref{Lem: DGMM consistency} for lower level conditions. Assumption \ref{Ass: Data}(iii) demands the usual faster-than-$n^{-1/4}$ rate of convergence on the first step estimators $\hat\gamma_{-k}$ and $\hat\alpha_{-k}$ as well as the preliminary estimator $\tilde\theta_{-k}$. Note that when the augmented moment function is doubly robust (so that $E[\psi(X,\theta_0,\gamma,\alpha_0)]=E[\psi(X,\theta_0,\gamma_0,\alpha)]=0$) and $\|E[\Delta_0(X,\theta,\gamma,\alpha)]\|=O(\|\gamma-\gamma_0\|_{\mathbf H}\|\alpha-\alpha_0\|_{\mathbf A})$, it suffices to have $\|\hat\gamma_{-k}-\gamma_0\|_{\mathbf H} \|\hat\alpha_{-k}-\alpha_0\|_{\mathbf A}=o_p(n^{-1/2})$. Assumption \ref{Ass: Data}(iv) is also minimal by requiring the sample size in each fold to be large. Assumption \ref{Ass: Data}(v) is the common consistency requirement on the weighting matrix. 

Assumption \ref{Ass: variance estimation} imposes additional conditions for variance estimation, which requires sample analogs of $\Sigma_0\equiv \mathrm{Var}[\psi_0(X)]$ and $G_0\equiv E[\nabla_\theta g(X,\theta_0,\gamma_0)]$ defined as:
\begin{align}
\hat\Sigma_n\equiv \frac{1}{n}\sum_{k=1}^{K}\sum_{i\in I_k} \hat\psi_{-k}(X_i)\hat\psi_{-k}^\intercal(X_i)~,\, \hat G_n\equiv \frac{1}{n}\sum_{k=1}^{K}\sum_{i\in I_k}\nabla_\theta g(X_i,\hat\theta_n,\hat\gamma_{-k})~,
\end{align}
for $\hat\psi_{-k}(X_i)\equiv \psi(X_i,\tilde\theta_{-k},\hat\gamma_{-k},\hat\alpha_{-k})$. In particular, Assumptions \ref{Ass: variance estimation}(i)(ii) and (iii) ensure the consistency of $\hat\Sigma_n$ and $\hat G_n$ respectively.

Proposition \ref{Pro: AN} formally establishes the asymptotic normality of the debiased GMM estimator and provides a consistent estimator for its asymptotic variance. 

\begin{pro}\label{Pro: AN}
Let Assumptions \ref{Ass: Model}, \ref{Ass: FSIF}, and \ref{Ass: Data} hold. Then it follows that, for $V_{\theta,0}\equiv [G_0^\intercal\Omega_0 G_0]^{-1}G_0^\intercal\Omega_0\Sigma_0\Omega_0 G_0[G_0^\intercal\Omega_0 G_0]^{-1}$,
\begin{align}\label{Eqn: linear expansion}
\sqrt n\{\hat\theta_n-\theta_0\} &= - [G_0^\intercal\Omega_0 G_0]^{-1}G_0^\intercal\Omega_0 \frac{1}{\sqrt n}\sum_{i=1}^{n}\psi_0(X_i) +o_p(1)\overset{L}{\to}N(0,V_{\theta,0})~.
\end{align}
If in addition Assumption \ref{Ass: variance estimation} holds, then
\begin{align}
\hat V_{\theta,n}\equiv [\hat G_n^\intercal\hat\Omega_n \hat G_n]^{-1}\hat G_n^\intercal\hat\Omega_n\hat\Sigma_n\hat\Omega_n \hat G_n[\hat G_n^\intercal\hat\Omega_n \hat G_n]^{-1}\overset{p}{\to} V_{\theta,0}~.
\end{align}
\end{pro}

\iffalse
Proposition \ref{Pro: AN} establishes the asymptotic linear expansion of $\hat\theta_n$, which in turn delivers the usual asymptotic normality. The ``magnitude'' of the asymptotic variance $V_{\theta,0}$ depends on the population weighting matrix $\Omega_0$. The choice $\Omega_0=\Sigma_0^{-1}$ leads to the efficient GMM estimator with the ``smallest'' variance $[G_0^\intercal\Sigma_0^{-1} G_0]^{-1}$. 
\fi

% Some references: \citet{ChenSantos2018OverID}, \citet{Vaart1991Differentibility}, \citet{Vaart1989InfoBound}, \citet{Newey1994AsymptoticVar}
% \citet{ChetverikovWlihelm2017MonNPIV}: a specific model with a specific shape; only about first step; finite sample results. 
% \citet{Hahn1998Role}: shows that ATE is adaptive wtr the first stage (i.e., the propensity) but not ATT. Moreover, not necessary (even harmful) to use propensity score for efficient estimation of ATT
% \citet{HiranoImbensRidder2003EfficientATE}: using estimated propensity score (instead of the true) leads to efficient estimators.

\section{Implementation and Simulations}\label{Sec: Implementation and Sims}

In this section, we outline some implementation details and conduct simulation studies. In particular, we discuss how monotonicity and convexity may be enforced in neural networks, while noting that much work remains to develop architectures for other important constraints in economics, e.g., the Slutsky restriction, supermodularity, quasi-concavity, and joint constraints. We focus on the finite-sample gains from exploiting shape constraints in both exogenous and endogenous settings.

% https://arxiv.org/pdf/2402.03970

\subsection{Implementation}\label{Sec: Implementation}

Algorithm \ref{Algo: DML} is a general recipe for computing the debiased-GMM estimator $\hat\theta_n$ with optimal weighting and its variance estimator $\hat V_{\theta,n}$. It starts with splitting the sample into $K$ folds. As in \citet{ChernozhukovChetverikovDemirerDufloHansenNeweyRobins2018Double,ChernozhukovEscancianoIchimuraNewey2022LocalRobust,ChernozhukovNeweySingh2022Automatic}, the general theory dictates that $K$ be small so that the sample size in each fold remains relatively large. In practice, one may choose $K=5$ or $K=10$ as recommended by \citet{ChernozhukovEscancianoIchimuraNewey2022LocalRobust,ChernozhukovNeweySingh2022Automatic}. Given the split sample, we then construct preliminary estimators for the parameters $\gamma_0$, $\alpha_0$ and possibly also for $\theta_0$ which may enter both the weighting matrix $\Omega_0=\Sigma_0^{-1}$ and the first step influence function $\phi$. In just-identified models (i.e., $d_\theta=d_g$), there is no need to obtain the preliminary estimator $\tilde\theta_{-k}$ if $\phi$ does not involve $\theta_0$.  
 
\begin{algorithm}% [t]
\caption{Automatic DML via Riesz Regression}\label{Algo: DML}
\SetAlgoLined

\KwIn{$\{X_i\}_{i=1}^n$, $\Gamma$, $K$ \Comment*[f]{Data,  constraint set of $\gamma_0$, and $\#$ of folds}\newline 
      $g,\phi$  \Comment*[f]{Moment function, first step influence function}}
      
\KwOut{Debiased GMM estimator $\hat\theta_n$ and its variance estimator $\hat V_{\theta,n}$}      

\BlankLine

Partition $I\equiv\{1,\ldots,n\}=\bigcup_{k=1}^K I_k$ \Comment*[r]{Sample splitting}

\BlankLine

\For{$k\leftarrow 1$ \KwTo $K$} 
{\Comment*[l]{Obtain preliminary estimators based on $\{X_i\colon i\in I\backslash I_k\}$}
$\hat\gamma_{-k}\leftarrow$ Construct an estimator $\hat\gamma_{-k}\in\Gamma$ of $\gamma_0$ \Comment*[r]{See discussions below}
$\hat\alpha_{-k}\leftarrow$ Run Riesz regression with induced constraints \Comment*[r]{See Section \ref{Sec: Riesz regression}} 
$\tilde\theta_{-k}\leftarrow$ Construct a preliminary estimator $\tilde \theta_{-k}$ of $\theta_0$  \Comment*[r]{Only if necessary} 

\BlankLine

\Comment*[l]{Obtain some intermediate objects for GMM optimization}
$\hat\psi(X_i,\theta)\leftarrow  g(X_i,\theta,\hat\gamma_{-k})+\phi(X_i,\tilde\theta_{-k},\hat\gamma_{-k},\hat\alpha_{-k})$ for $i\in I_k$ \;
$\hat\psi_n(\theta) \mathrel{+}= \frac{1}{n}\sum_{i\in I_k}\hat\psi(X_i,\theta)$  \Comment*[r]{Orthogonalized sample moment} 
$\hat\Sigma_n \mathrel{+}= \frac{1}{n}\sum_{i\in I_k}\hat\psi(X_i,\tilde\theta_{-k})\hat\psi(X_i,\tilde\theta_{-k})^\intercal$  \Comment*[r]{Set to $I_{d_g}$ if $d_g=d_\theta$} 
}

\BlankLine

\Comment*[l]{Compute $\hat\theta_n$ and its variance estimator}
$\hat\theta_n\leftarrow \min_{\theta\in\Theta} \hat\psi_n(\theta)^\intercal\hat\Sigma_n^{-1}\hat\psi_n(\theta)$ \Comment*[r]{Debiased GMM estimator} 
$\hat G_n  \leftarrow \frac{1}{n}\sum_{k=1}^K\sum_{i\in I_k}\nabla_\theta g(X_i,\hat\theta_n,\hat\gamma_{-k})$  \Comment*[r]{Sample Jacobian matrix} 
$\hat V_{\theta,n}\leftarrow [\hat G_n^\intercal \hat\Sigma_n^{-1}\hat G_n]^{-1}$ \Comment*[r]{Variance estimator}
\end{algorithm}
% The construction of $\tilde\theta_{-k}$: not clear in \citet{ChernozhukovEscancianoIchimuraNewey2022LocalRobust}.

A key step in estimating $\gamma_0$ and $\alpha_0$ is incorporating the relevant shape constraint. A generic strategy is to post-process an unconstrained estimator \citep{BonakdarpourChatterjeeBarberLafferty2018Reshape,ChenChernozhukovFernandezKostyshakLuo2021Shape}. While this approach works well in low dimensional settings, it becomes impractical as dimensionality increases. A strand of the neural network literature instead develops architectures that enforce shape constraints by construction, thereby yielding global guarantees and improved computational tractability. Since these architectures do not appear well known in economics, below we review those employed in our numerical exercises, assuming that the reader has a basic familiarity with neural networks---see, e.g., \citet{GoodfellowBengioCourville2016DL} for an introduction.

%%%%%%%%%%%%%%%%%%%%%% Draw neural networks %%%%%%%%%%%%%%%%%%%%%%%%%%%
% ---------- Common TikZ styles ----------
\tikzset{layer sep/.store in=\LayerSep,layer sep=0.70,
         neuron/.style={circle,draw,minimum size=5.9mm,inner sep=0pt},
         networkblock/.style={draw,minimum width=3.2cm,minimum height=1.1cm,
                              align=center,font=\footnotesize,inner sep=3pt},
         networkheadblock/.style={draw,minimum width=4.1cm,minimum height=3.15cm,
                                  align=center,font=\footnotesize,inner sep=3pt},
         nodeunit/.style={minimum size=5.9mm,inner sep=0pt,
                          text height=1.25ex,text depth=1.25ex,font=\small},
         labin/.style={font=\footnotesize,inner sep=0pt},
         paneltitle/.style={font=\footnotesize,inner sep=0pt},
         conn/.style={-{Stealth[length=1.3mm,width=0.9mm]},ultra thin},
         network dots/.pic={
           \foreach \dy in {-0.08,0,0.08}{\fill (0,\dy) circle[radius=0.45pt];}
         }}
% Fully connect two comma-separated node lists
\newcommand{\FullyConnect}[2]{%
  \foreach \fromnode in {#1}{%
    \foreach \tonode in {#2}{%
      \edef\DrawConnection{\noexpand\draw[conn] (\fromnode) -- (\tonode); }%
      \DrawConnection
    }%
  }%
}
\makeatletter
\pgfplotsset{
  colormap/BuGn-9,
  auto title/.style={
    title={(\@alph{\pgfplots@group@current@plot}) #1},
  },
  NetworkGroup/.style={
    group style={horizontal sep=0pt},
    width=0.46\textwidth,
    height=0.40\textwidth,
    xmin=-2.0,xmax=2.0,
    ymin=-1.45,ymax=2.85,
    hide axis,
    clip=false,
    scale only axis,
    title style={paneltitle,at={(axis description cs:0.5,-0.11)},anchor=north},
  },
}
\makeatother

% ============================================================
% Figure 1: fully monotone/convex network - BuGn palette
% ============================================================
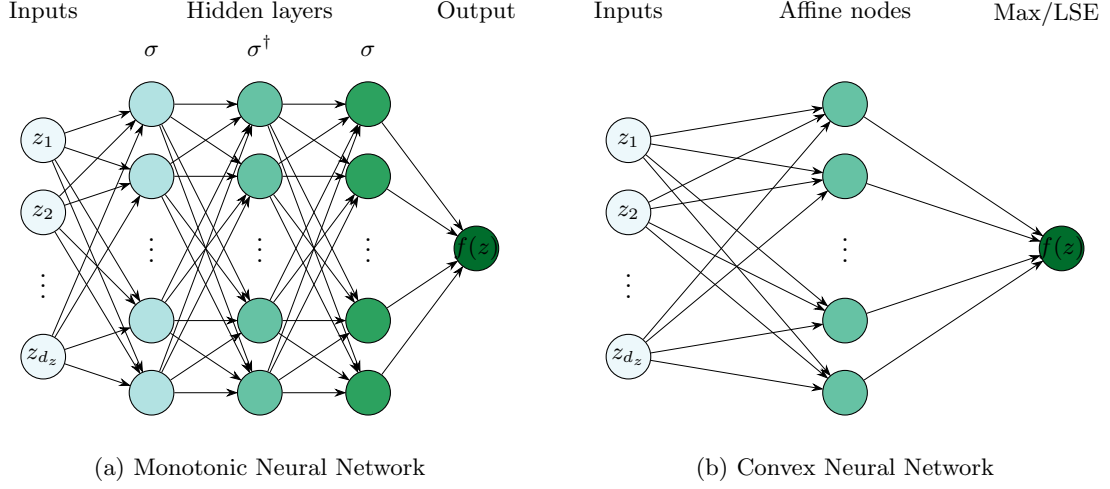
\begin{figure}[H]
\centering
\begin{tikzpicture}
\begin{groupplot}[NetworkGroup,group style={group name=net,group size=2 by 1}]
% Monotone neural network
\nextgroupplot[auto title={Monotonic Neural Network}]
\pgfplotsset{colormap/BuGn-9}

% Input layer
\node[neuron,fill=BuGn-B] (mnnx1) at (axis cs:-1.7,1.30) {};
\node[neuron,fill=BuGn-B] (mnnx2) at (axis cs:-1.7,0.60) {};
\node[nodeunit]                 (mnnxd) at (axis cs:-1.7,-0.10) {$\vcenter{\hbox{$\vdots$}}$};
\node[neuron,fill=BuGn-B] (mnnxp) at (axis cs:-1.7,-0.80) {};
\node[labin] at (mnnx1.center) {$z_1$};
\node[labin] at (mnnx2.center) {$z_2$};
\node[labin] at (mnnxp.center) {$z_{d_z}$};

% Hidden layer 1
\node[neuron,fill=BuGn-E] (mnnm1a) at (axis cs:-0.85,1.65) {};
\node[neuron,fill=BuGn-E] (mnnm1b) at (axis cs:-0.85,0.95) {};
\node[nodeunit]                 (mnnm1d) at (axis cs:-0.85,0.25) {$\vcenter{\hbox{$\vdots$}}$};
\node[neuron,fill=BuGn-E] (mnnm1c) at (axis cs:-0.85,-0.45) {};
\node[neuron,fill=BuGn-E] (mnnm1e) at (axis cs:-0.85,-1.15) {};

% Hidden layer 2
\node[neuron,fill=BuGn-G] (mnnm2a) at (axis cs:0,1.65) {};
\node[neuron,fill=BuGn-G] (mnnm2b) at (axis cs:0,0.95) {};
\node[nodeunit]                 (mnnm2d) at (axis cs:0,0.25) {$\vcenter{\hbox{$\vdots$}}$};
\node[neuron,fill=BuGn-G] (mnnm2c) at (axis cs:0,-0.45) {};
\node[neuron,fill=BuGn-G] (mnnm2e) at (axis cs:0,-1.15) {};

% Hidden layer 3
\node[neuron,fill=BuGn-I] (mnnm3a) at (axis cs:0.85,1.65) {};
\node[neuron,fill=BuGn-I] (mnnm3b) at (axis cs:0.85,0.95) {};
\node[nodeunit]                 (mnnm3d) at (axis cs:0.85,0.25) {$\vcenter{\hbox{$\vdots$}}$};
\node[neuron,fill=BuGn-I] (mnnm3c) at (axis cs:0.85,-0.45) {};
\node[neuron,fill=BuGn-I] (mnnm3e) at (axis cs:0.85,-1.15) {};

% Output layer
\node[neuron,fill=BuGn-K] (mnnyout) at (axis cs:1.7,0.25) {};
\node[labin] at (mnnyout.center) {$f(z)$};

% Connections
\FullyConnect{mnnx1,mnnx2,mnnxp}{mnnm1a,mnnm1b,mnnm1c,mnnm1e}
\FullyConnect{mnnm1a,mnnm1b,mnnm1c,mnnm1e}{mnnm2a,mnnm2b,mnnm2c,mnnm2e}
\FullyConnect{mnnm2a,mnnm2b,mnnm2c,mnnm2e}{mnnm3a,mnnm3b,mnnm3c,mnnm3e}
\FullyConnect{mnnm3a,mnnm3b,mnnm3c,mnnm3e}{mnnyout}

% Layer labels
\node[labin,anchor=north] at (axis cs:-1.7,2.65) {Inputs};
\node[labin,anchor=north] at (axis cs:0,2.65) {Hidden layers};
\node[labin,anchor=base] at (axis cs:-0.85,2.12) {$\sigma$};
\node[labin,anchor=base] at (axis cs:0,2.12) {$\sigma^\dag$};
\node[labin,anchor=base] at (axis cs:0.85,2.12) {$\sigma$};
\node[labin,anchor=north] at (axis cs:1.7,2.65) {Output};

% Convex branch
\nextgroupplot[auto title={Convex Neural Network}]
\pgfplotsset{colormap/BuGn-9}

% Inputs
\node[neuron,fill=BuGn-B] (cnnx1) at (axis cs:-1.7,1.30) {};
\node[neuron,fill=BuGn-B] (cnnx2) at (axis cs:-1.7,0.60) {};
\node[nodeunit]                 (cnnxd) at (axis cs:-1.7,-0.10) {$\vcenter{\hbox{$\vdots$}}$};
\node[neuron,fill=BuGn-B] (cnnxp) at (axis cs:-1.7,-0.80) {};
\node[labin] at (cnnx1.center) {$z_{1}$};
\node[labin] at (cnnx2.center) {$z_{2}$};
\node[labin] at (cnnxp.center) {$z_{d_z}$};

% Affine nodes
\node[neuron,fill=BuGn-G] (cnna1) at (axis cs:0,1.65) {};
\node[neuron,fill=BuGn-G] (cnna2) at (axis cs:0,0.95) {};
\node[nodeunit]                 (cnnad) at (axis cs:0,0.25) {$\vcenter{\hbox{$\vdots$}}$};
\node[neuron,fill=BuGn-G] (cnnaKm) at (axis cs:0,-0.45) {};
\node[neuron,fill=BuGn-G] (cnnaK) at (axis cs:0,-1.15) {};

% Max or log-sum-exp outputs
\node[neuron,fill=BuGn-K] (cnnhout) at (axis cs:1.7,0.25) {};
\node[labin] at (cnnhout.center) {$f(z)$};

% Connections
\FullyConnect{cnnx1,cnnx2,cnnxp}{cnna1,cnna2,cnnaKm,cnnaK}
\FullyConnect{cnna1,cnna2,cnnaKm,cnnaK}{cnnhout}

% Layer labels
\node[labin,anchor=north] at (axis cs:-1.7,2.65) {Inputs};
\node[labin,anchor=north] at (axis cs:0,2.65) {Affine nodes};
\node[labin,anchor=north east,align=center] at (axis cs:2.0,2.65) {Max/LSE};
\end{groupplot}
\end{tikzpicture}
\caption{Left figure: The monotone neural network of \citet{SartorSinigagliaSusto2025MNN} where the weights are nonnegative. Right figure: The max-affine architecture and the log-sum-exp (LSE) architecture for convexity \citep{CalafioreGaubertPossieri2019Convex}.}
\label{Fig: MNN CNN full}
\end{figure}

\begin{figure}[H]
\centering
\pgfmathsetlengthmacro{\ShapeBoxHeight}{2*\LayerSep*1.35cm}
\def\BranchLabelGap{1.5mm}
\begin{tikzpicture}[x=1.55cm,y=1.35cm,scale=0.85,transform shape]
% ----- Inputs -----
\node[neuron,fill=BuGn-A] (xC1) at (0,{1.5*\LayerSep+\LayerSep}) {};
\node[neuron,fill=BuGn-A] (xCp) at (0,{1.5*\LayerSep-\LayerSep}) {};
\pic at ($(xC1)!0.5!(xCp)$) {network dots};
\node[labin] at (xC1.center) {$z^{\mathrm{c}}_{1}$};
\node[labin] at (xCp.center) {$z^{\mathrm{c}}_{p_{\mathrm{c}}}$};

\pgfplotsset{colormap/BuPu-9}
\node[neuron,fill=BuPu-B] (xF1) at (0,{-2.5*\LayerSep+\LayerSep}) {};
\node[neuron,fill=BuPu-B] (xFp) at (0,{-2.5*\LayerSep-\LayerSep}) {};
\pic at ($(xF1)!0.5!(xFp)$) {network dots};
\node[labin] at (xF1.center) {$z^{\mathrm{f}}_{1}$};
\node[labin] at (xFp.center) {$z^{\mathrm{f}}_{p_{\mathrm{f}}}$};
\pgfplotsset{colormap/BuGn-9}

\node[labin,rotate=90,anchor=center] at (-0.4,{1.5*\LayerSep}) {Constrained inputs};
\node[labin,rotate=90,anchor=center] at (-0.4,{-2.5*\LayerSep}) {Free inputs};

% ----- Constrained branch hidden layers and outputs -----
\node[networkblock,minimum height=\ShapeBoxHeight,fill=BuGn-F] (shapeHidden) at (2,{1.5*\LayerSep})
  {Hidden layers};
\node[neuron,fill=BuGn-G] (m4a) at (4,{3*\LayerSep}) {};
\node[neuron,fill=BuGn-G] (m4b) at (4,{2*\LayerSep}) {};
\node[neuron,fill=BuGn-G] (m4c) at (4,0) {};
\pic at ($(m4b)!0.5!(m4c)$) {network dots};

% ----- Free branch -----
\pgfplotsset{colormap/BuPu-9}
\foreach \y/\n in {-1/a,-2/b,-4/c}{\node[neuron,fill=BuPu-E] (n1\n) at (1,{\y*\LayerSep}) {};}
\foreach \y/\n in {-1/a,-2/b,-4/c}{\node[neuron,fill=BuPu-G] (n2\n) at (2,{\y*\LayerSep}) {};}
\foreach \y/\n in {-1/a,-2/b,-4/c}{\node[neuron,fill=BuPu-I] (n3\n) at (3,{\y*\LayerSep}) {};}
\foreach \y/\n in {-1/a,-2/b,-4/c}{\node[neuron,fill=BuPu-K] (n4\n) at (4,{\y*\LayerSep}) {};}
\foreach \above/\below in {n1b/n1c,n2b/n2c,n3b/n3c,n4b/n4c}{
  \pic at ($(\above)!0.5!(\below)$) {network dots};
}

% ----- Connections -----
\draw[conn] (xC1) -- (shapeHidden.west);
\draw[conn] (xCp) -- (shapeHidden.west);
\foreach \tonode in {m4a,m4b,m4c}{\draw[conn] (shapeHidden.east) -- (\tonode);}
\FullyConnect{xF1,xFp}{n1a,n1b,n1c}
\FullyConnect{n1a,n1b,n1c}{n2a,n2b,n2c}
\FullyConnect{n2a,n2b,n2c}{n3a,n3b,n3c}
\FullyConnect{n3a,n3b,n3c}{n4a,n4b,n4c}

% ----- Constrained head hidden layers and output -----
\pgfplotsset{colormap/BuGn-9}
\node[networkheadblock,fill=BuGn-I] (headHidden) at (6.5,{-0.5*\LayerSep})
  {Hidden layers};
\node[neuron,fill=BuGn-K] (y2) at (9,{-0.5*\LayerSep}) {};
\node[labin] at (y2.center) {$f(z)$};
\foreach \fromnode in {m4a,m4b,m4c,n4a,n4b,n4c}{\draw[conn] (\fromnode) -- (headHidden.west);}
\draw[conn] (headHidden.east) -- (y2);

% Branch labels
\node[labin,anchor=south] at ($(2,0 |- m4a.north)+(0,\BranchLabelGap)$) {Constrained branch};
\node[labin,anchor=north] at ($(2,0 |- n4c.south)+(0,-\BranchLabelGap)$) {Free branch};
\node[labin,anchor=south] at ($(6.5,0 |- m4a.north)+(0,\BranchLabelGap)$) {Constrained head};
\end{tikzpicture}
\caption{A two-branch neural network: The upper branch is a constrained subnetwork, the lower branch is an unconstrained subnetwork, and the head combines the two subnetworks in a way that ensures the constraint with respect to $(z_1^{\mathrm{c}},\ldots,z_{p_{\mathrm{c}}}^{\mathrm{c}})$.}
\label{Fig: MNN CNN partial}
\end{figure}
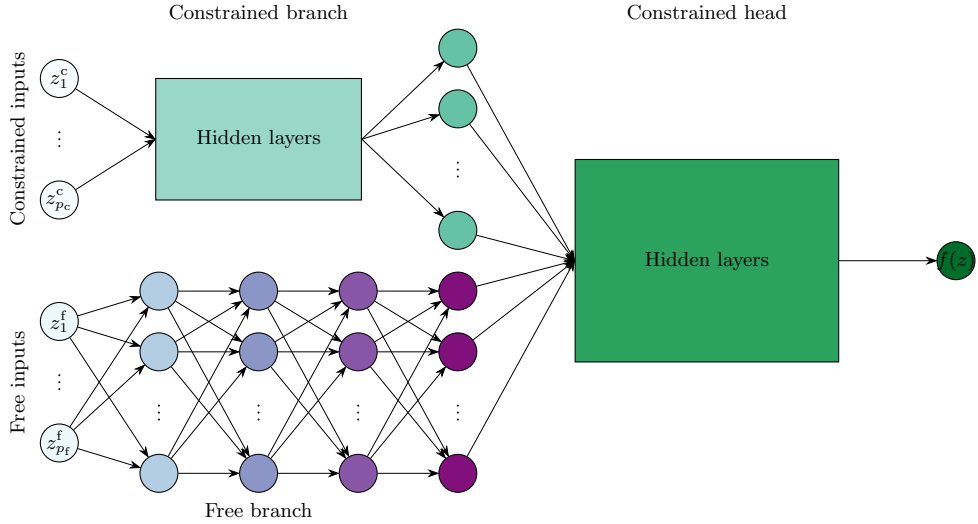

The literature primarily focuses on monotonicity and convexity. Figure \ref{Fig: MNN CNN full}(a) presents the monotonic neural network of \citet{SartorSinigagliaSusto2025MNN} in which (i) the weights are constrained to be nonnegative; (ii) the activation functions sequentially alternate between $\sigma\colon\mathbf R\to\mathbf R$ and its point reflection $\sigma^\dag(\cdot)\equiv-\sigma(-\cdot)$; (iii) $\sigma$ is weakly increasing and admits a finite limit at either $\infty$ or $-\infty$ (or both); and (iv) there are at least three hidden layers. Crucially, compared to prior monotonic networks, this architecture permits the use of unbounded activation functions (e.g., ReLU) while also retaining the universal approximation property. In many applications, however, one is interested in partially monotonic functions. Figure \ref{Fig: MNN CNN partial} illustrates a two-branch neural network that is monotonic in a subvector $(z_1^{\mathrm{c}},\ldots,z_{p_{\mathrm{c}}}^{\mathrm{c}})$ of the input $z=(z_1,\ldots,z_{d_z})$.

Convexity requires a more careful treatment. Although a number of partially convex architectures have been developed (see, e.g., \citet{AmosXuKolter2017ICNN} and \citet{KimKim2022Convex}), they do not appear to perform well for inference, as documented in \citet{ChenChoiFang2026SCDL}. We therefore adopt the two-branch convex architecture proposed by \citet{ChenChoiFang2026SCDL}, where the convex branch is a pointwise max-affine subnet with multiple outputs, and the head is similar to the convex branch except that the inputs from the convex branch are constrained to be nonnegative so that the overall architecture is partially convex in $(z_1^{\mathrm{c}},\ldots,z_{p_{\mathrm{c}}}^{\mathrm{c}})$---see Figures \ref{Fig: MNN CNN full}(b) and \ref{Fig: MNN CNN partial}. The same architecture can also accommodate concavity: we estimate the negative of the target function and then reverse the sign of the resulting estimate.

\iffalse
Finally, Algorithm \ref{Algo: DML} depends on how the data is split. In practice, different splitting schemes may lead to different estimates, which may be undesirable. \citet{ChernozhukovChetverikovDemirerDufloHansenNeweyRobins2018Double} recommend median aggregation over multiple splits to alleviate the concern. Specifically, let $\hat\theta_{n,s}$ and $\hat V_{n,s}$ be the output of Algorithm \ref{Algo: DML} relative to a particular sample split indexed by $s=1,\ldots,\mathsf S$. Then median aggregation picks the entry-wise median $\hat\theta_n^*$ of $\{\hat\theta_{n,s}\}_{s=1}^{\mathsf S}$ and the matrix with median spectral norm from the $\mathsf S$ matrices:
\begin{align}
\hat V_{n,1}+(\hat\theta_{n,1}-\hat\theta_n^*)(\hat\theta_{n,1}-\hat\theta_n^*)^\intercal,\ldots,\hat V_{n,\mathsf S}+(\hat\theta_{n,\mathsf S}-\hat\theta_n^*)(\hat\theta_{n,\mathsf S}-\hat\theta_n^*)^\intercal~.
\end{align}
Following \citet{AhrensChernozhukovHansenKozburSchafferWiemann2026DML}, we choose $\mathsf S=5$ in our empirical applications. 
% Beyond its robustness to outliers, median aggregation also reduces estimation risks \citep{ChernozhukovDemirerDufloFernandez2025GML}.  
\fi

\subsection{Simulation Studies}\label{Sec: Sims}

In this section, we conduct Monte Carlo simulations to illustrate our estimation procedure and demonstrate the potential gains of imposing shape constraints. Below, we let $\Phi$ be the standard normal cdf and denote by $\mathrm{Uni}[0,1]$, $\mathrm{Ber}(\varrho)$, and $N(0,1)$ the uniform distribution on $[0,1]$, the Bernoulli distribution with success probability $\varrho$, and the standard normal distribution respectively.   

% STATE THE GOAL OF EACH SIMULATION DESIGN

\subsubsection{Average Treatment Effects}\label{Sec: Sims, ATE}

Let $Z\in\mathbf R^{d_z}$ be a vector of i.i.d.\ entries drawn from $\mathrm{Uni}[0,1]$, $D\sim \mathrm{Ber}(\varrho_0(Z))$ conditional on $Z$, $U\sim N(0,1)$ drawn independently of $D$ and $Z$, and
\begin{align}
Y= \mu_0(Z)+D\tau_0(Z)+U~.
\end{align}
Let $Z_{\mathrm{c}}\in\mathbf R^{d_{\mathrm{c}}}$ and $Z_{\mathrm{f}}\in\mathbf R^{d_{\mathrm{f}}}$ denote the constrained and free inputs respectively, and, for $Z=(Z_{\mathrm{c}}^\intercal,Z_{\mathrm{f}}^\intercal)^\intercal$, set $\varrho_0(Z)=1/[1+\exp\{-Z^\intercal\mathsf a_\varrho\}]$ and
\begin{align}
\mu_0(Z)&={\mathsf a_{\mu,\mathrm{c}}^\intercal} Z_{\mathrm{c}}+\Phi({\mathsf b_{\mu,\mathrm{c}}^\intercal} Z_{\mathrm{c}})+\Phi({\mathsf a_{\mu,\mathrm{f}}^\intercal} Z_{\mathrm{f}}+{\mathsf b_{\mu,\mathrm{f}}^\intercal} q(Z_{\mathrm{f}}))~,\\
\tau_0(Z)&=[{\mathsf a_{\tau,\mathrm{c}}^\intercal} Z_{\mathrm{c}}+\Phi({\mathsf b_{\tau,\mathrm{c}}^\intercal} Z_{\mathrm{c}})]\Phi(\mathsf a_{\tau,\mathrm{f}}^\intercal Z_{\mathrm{f}})~,   
\end{align}
where $q(Z_{\mathrm{f}})$ collects all second-degree polynomial terms of $Z_{\mathrm{f}}$ (squares and pairwise interactions), and each of the vectors $\mathsf a_\varrho$, $\mathsf a_{\mu,\mathrm{c}}$, $\mathsf b_{\mu,\mathrm{c}}$, $\mathsf a_{\mu,\mathrm{f}}$, $\mathsf b_{\mu,\mathrm{f}}$, $\mathsf a_{\tau,\mathrm{c}}$, $\mathsf b_{\tau,\mathrm{c}}$, and $\mathsf a_{\tau,\mathrm{f}}$ consists of i.i.d.\ entries drawn from uniform distributions whose supports are given in Table \ref{Tab: ATE designs}. By design, both $\mu_0(Z)$ and $\tau_0(Z)$ are monotone and concave in $Z_{\mathrm{c}}$. Consequently, $(D,Z)\mapsto\gamma_0(D,Z)=\mu_0(Z)+D\tau_0(Z)$ satisfies the same shape restrictions. We investigate the separate effects of leveraging monotonicity and concavity on the finite-sample performance of the resulting debiased estimators.

\begin{table}[!htbp]
\centering
\caption{Supports of Coefficients for the Baseline ATE Designs}
\label{Tab: ATE designs}
%\sisetup{round-mode=places, round-precision=1}  
\begin{tblr}{
  colsep=1pt,
  rowsep=2pt,
  column{1}={leftsep=2pt},
  column{Z}={rightsep=2pt},
%  column{1-6}={cmd=\num},
  colspec={
    Q[c,m,wd=0.15\linewidth]
    Q[c,m,wd=0.15\linewidth]
    Q[c,m,wd=0.15\linewidth]    
    Q[c,m,wd=0.15\linewidth]
    Q[c,m,wd=0.15\linewidth]
    Q[c,m,wd=0.15\linewidth]
  },
  hline{1,2,Z} = {-}{},
%  row{1}={cmd=},
}
$\mathsf a_\varrho$ & $\mathsf a_{\mu,\mathrm{c}},\mathsf b_{\mu,\mathrm{c}}$ & ${\mathsf a_{\mu,\mathrm{f}}}$ & ${\mathsf b_{\mu,\mathrm{f}}}$ & $\mathsf a_{\tau,\mathrm{c}},\mathsf b_{\tau,\mathrm{c}}$ & ${\mathsf a_{\tau,\mathrm{f}}}$ \\
{}[-0.5,0.5] & [0.2,0.3] & [-0.1,0.1] & [-0.05,0.05] & [0.2,0.3] & [-0.2,0.2] \\
\end{tblr}
\end{table}
 
We also consider deviations from the above baseline design. First, we vary the signal-to-noise ratio by drawing $U \sim N(0,\sigma_U^2)$ with $\sigma_U^2 \in \{0.5,1,5\}$. The case $\sigma_U^2=1$ serves as the baseline, while $\sigma_U^2=0.5$ and $\sigma_U^2=5$ correspond to higher and lower signal-to-noise ratios, respectively. Second, we vary the ``strength'' of monotonicity/concavity by considering three supports for the entries of $\mathsf a_{\mu,\mathrm{c}}$, $\mathsf b_{\mu,\mathrm{c}}$, $\mathsf a_{\tau,\mathrm{c}}$, and $\mathsf b_{\tau,\mathrm{c}}$: $[0.1,0.2]$ (weak), $[0.2,0.3]$ (baseline), and $[0.3,0.4]$ (strong). For brevity, we report only the results for the baseline specification with $\sigma_U^2=1$. The results under $\sigma_U^2\in\{0.5,5\}$ are qualitatively similar, suggesting that our findings are robust to alternative signal-to-noise ratios. These additional results are available upon request.

For each design, we draw $5000$ independent  i.i.d.\ samples of size $1000$, each with $d_z=40$ covariates that comprise $d_{\mathrm{c}}=20$ constrained ones and $d_{\mathrm{f}}=20$ free ones. We then estimate $\gamma_0$ and $\alpha_0$ using partially monotonic/convex neural networks, varying the number of imposed constrained inputs over $p_{\mathrm{c}}=0,1,\ldots,20$. In the partially monotonic architecture, both the monotone and free branches consist of two hidden layers with width $32$ for each and $32$ outputs, while the head network has two hidden layers with the same width for each. In the partially convex architecture, the free branch consists of three hidden layers with width $128$ for each and $8$ outputs, the convex branch consists of $32$ affine nodes with $4$ outputs, while the head has $128$ affine nodes. 

Training of the networks is carried out using the AdamW optimizer, a state-of-the-art variant of stochastic gradient descent. We further incorporate dropout, early stopping, and $\ell^2$-penalization for the sake of regularization. For each sample, we construct $K=5$ folds such that treated and untreated individuals are approximately evenly distributed across the folds. Given existing theoretical and simulation studies (see, e.g., \citet{ChernozhukovNeweyQuintasSyrgkanis2024RieszReg} and references therein), we benchmark our results against DML via Riesz regression but without shape constraints, which corresponds to $p_{\mathrm{c}}=0$ in the plots below.

% Selected from the second row of Draft7 Figures 3--6: $\sigma_U^2=1$.
% The grid uses $t=4$ and includes the endpoint, so $p\in\{0,1,5,9,13,17,20\}$.
\pgfplotstableread{
num_mon bias std mean se coverage length
0	0.009039872	0.072466711	1.381145572	0.071396392	0.952	0.279873857
1	0.000266501	0.070234476	1.372372201	0.070270935	0.947	0.275462063
5	0.002036933	0.069527558	1.374142633	0.069300152	0.9462	0.271656596
9	0.022116964	0.069382357	1.394222664	0.069177566	0.9374	0.27117606
13	0.012232518	0.06857148	1.384338218	0.068721717	0.9472	0.269389129
17	0.016108279	0.068317802	1.388213979	0.068128666	0.9414	0.267064371
20	0.011119099	0.067298628	1.383224799	0.067299259	0.9474	0.263813094
}\MonLM

\pgfplotstableread{
num_mon bias std mean se coverage length
0	0.007204591	0.073929027	2.013739087	0.073076758	0.9482	0.286460891
1	-0.004257891	0.072093505	2.002276606	0.071664658	0.9474	0.280925458
5	-0.002056687	0.070354959	2.004477809	0.069958324	0.9458	0.27423663
9	0.01161494	0.069722074	2.018149436	0.069426529	0.9428	0.272151993
13	0.006463038	0.069340462	2.012997534	0.068891871	0.9464	0.270056133
17	0.010298526	0.068708109	2.016833023	0.068281159	0.9438	0.267662145
20	0.00921005	0.067774701	2.015744546	0.067630856	0.9472	0.265112955
}\MonMM

\pgfplotstableread{
num_mon bias std mean se coverage length
0	0.004767242	0.07633053	2.607255209	0.075294589	0.9508	0.29515479
1	-0.00881513	0.072886081	2.593672837	0.073390909	0.9476	0.287692365
5	-0.006978	0.071773499	2.595509967	0.071036219	0.9408	0.27846198
9	0.003671135	0.07078449	2.606159101	0.070159375	0.9474	0.27502475
13	0.000102286	0.069905226	2.602590253	0.069331072	0.9492	0.2717778
17	0.002865768	0.069523704	2.605353734	0.0687721	0.9446	0.269586633
20	0.006670849	0.068737326	2.609158815	0.068239404	0.9466	0.267498463
}\MonHM

\pgfplotstableread{
num_mon bias std mean se coverage length
0	0.009039872	0.072466711	1.381145572	0.071396392	0.952	0.279873857
1	0.027490055	0.070429752	1.399595755	0.070936667	0.9348	0.278071734
5	0.013296048	0.069357202	1.385401748	0.069994264	0.9502	0.274377515
9	0.006107393	0.069132457	1.378213093	0.069608902	0.9466	0.272866895
13	0.003929105	0.068376285	1.376034806	0.069223481	0.9496	0.271356047
17	0.003258974	0.067923819	1.375364674	0.068949095	0.9526	0.270280452
20	0.005902683	0.067105004	1.378008383	0.068780165	0.951	0.269618247
}\ConLM

\pgfplotstableread{
num_mon bias std mean se coverage length
0	0.007204591	0.073929027	2.013739087	0.073076758	0.9482	0.286460891
1	0.026496917	0.071428294	2.033031414	0.07158757	0.9344	0.280623275
5	0.010887851	0.070784223	2.017422348	0.070813949	0.9508	0.277590678
9	0.003194483	0.070317319	2.009728979	0.070465918	0.9464	0.276226399
13	0.004310014	0.069594818	2.01084451	0.070066682	0.9494	0.274661395
17	0.005456295	0.068750468	2.011990792	0.0697093	0.9518	0.273260455
20	0.007150389	0.068319208	2.013684886	0.069475436	0.952	0.272343709
}\ConMM

\pgfplotstableread{
num_mon bias std mean se coverage length
0	0.004767242	0.07633053	2.607255209	0.075294589	0.9508	0.29515479
1	0.02398886	0.072251341	2.626476827	0.072069217	0.9382	0.282511329
5	0.005773408	0.072026708	2.608261374	0.071776503	0.948	0.281363892
9	-0.002663477	0.071276426	2.599824489	0.071464746	0.9458	0.280141804
13	0.003718299	0.070855048	2.606206266	0.071112049	0.9472	0.278759232
17	0.008086438	0.069885816	2.610574405	0.071041648	0.9532	0.278483261
20	0.008900983	0.069606871	2.611388949	0.070940508	0.9522	0.27808679
}\ConHM

% Define some tikz styles
\pgfplotsset{
  GroupCommon/.style={
    group style={horizontal sep=0.8cm,vertical sep=1cm},
    scaled y ticks=false,
    xmin=0, xmax=20,
    xtick={0,1,5,9,13,17,20},
    xlabel={$p_{\mathrm{c}}$},
    x label style={at={(axis description cs:0.95,0)},anchor=south,font=\footnotesize},
    tick label style={/pgf/number format/fixed,font=\footnotesize},
    yticklabel style={rotate=90,/pgf/number format/fixed,/pgf/number format/precision=3,/pgf/number format/fixed zerofill},
    enlarge x limits={abs=0.5},
    enlarge y limits,
    legend style={cells={align=center},column sep=3pt,row sep=3pt,draw=none},
    cycle list={
      {smooth,tension=0.5,color=Paired-B, mark=10-pointed star,mark size=1.75pt,line width=0.5pt},
      {smooth,tension=0.5,color=Dark2-B, mark=halfsquare*,mark size=1.75pt,line width=0.5pt},
      {smooth,tension=0.5,color=Paired-D, mark=halfcircle*,mark size=1.75pt,line width=0.5pt},
      {smooth,tension=0.5,color=Paired-J, mark=triangle*,mark size=1.75pt,line width=0.5pt},
    }
    },
}

%%% Results: selected second row from the 3x3 monotonicity groupplot
\begin{figure}[!ht]
\centering
\begin{tikzpicture}[scale=0.62, transform shape]
% Left axes: Bias and Coverage
\begin{groupplot}[GroupCommon,
    group style={group name=mon,group size=3 by 2,vertical sep=1.45cm},
    axis y line*=left,
]
\nextgroupplot[
    ymin=-0.069096467,ymax=0.069096467,ytick={-0.069096467,0,0.069096467},
    y axis line style={color=Paired-B},
    ytick style={color=Paired-B},
    every y tick label/.append style={color=Paired-B},
    legend style={legend to name=LegBias}
]
\node[anchor=north,align=center] at (axis description cs:0.5,0.99)
  {Shape Strength: $[0.1,0.2]$};
\addplot[color=Paired-B,densely dashed,forget plot] coordinates {(0,0) (20,0)};
\addplot table[x=num_mon,y=bias] from \MonLM;
\addlegendentry{Bias (left axis)}
\nextgroupplot[
    ymin=-0.069990040,ymax=0.069990040,ytick={-0.069990040,0,0.069990040},
    y axis line style={color=Paired-B},
    ytick style={color=Paired-B},
    every y tick label/.append style={color=Paired-B},
]
\node[anchor=north,align=center] at (axis description cs:0.5,0.99)
  {Shape Strength: $[0.2,0.3]$};
\addplot[color=Paired-B,densely dashed,forget plot] coordinates {(0,0) (20,0)};
\addplot table[x=num_mon,y=bias] from \MonMM;
\nextgroupplot[
    ymin=-0.071340609,ymax=0.071340609,ytick={-0.071340609,0,0.071340609},
    y axis line style={color=Paired-B},
    ytick style={color=Paired-B},
    every y tick label/.append style={color=Paired-B},
]
\node[anchor=north,align=center] at (axis description cs:0.5,0.99)
  {Shape Strength: $[0.3,0.4]$};
\addplot[color=Paired-B,densely dashed,forget plot] coordinates {(0,0) (20,0)};
\addplot table[x=num_mon,y=bias] from \MonHM;
\nextgroupplot[
    ymin=0.9,ymax=1.0,ytick={0.9,0.95,1.0},yticklabels={0.9,0.95,1},
    y axis line style={color=Paired-D},
    ytick style={color=Paired-D},
    every y tick label/.append style={color=Paired-D},
    cycle list shift=2,
    legend style={legend to name=LegCov}
]
\node[anchor=north,align=center] at (axis description cs:0.5,0.99)
  {Shape Strength: $[0.1,0.2]$};
\addplot[color=Paired-D,densely dashed,forget plot] coordinates {(0,0.95) (20,0.95)};
\addplot table[x=num_mon,y=coverage] from \MonLM;
\addlegendentry{Coverage of CI (left axis)}
\nextgroupplot[
    ymin=0.9,ymax=1.0,ytick={0.9,0.95,1.0},yticklabels={0.9,0.95,1},
    y axis line style={color=Paired-D},
    ytick style={color=Paired-D},
    every y tick label/.append style={color=Paired-D},
    cycle list shift=2
]
\node[anchor=north,align=center] at (axis description cs:0.5,0.99)
  {Shape Strength: $[0.2,0.3]$};
\addplot[color=Paired-D,densely dashed,forget plot] coordinates {(0,0.95) (20,0.95)};
\addplot table[x=num_mon,y=coverage] from \MonMM;
\nextgroupplot[
    ymin=0.9,ymax=1.0,ytick={0.9,0.95,1.0},yticklabels={0.9,0.95,1},
    y axis line style={color=Paired-D},
    ytick style={color=Paired-D},
    every y tick label/.append style={color=Paired-D},
    cycle list shift=2
]
\node[anchor=north,align=center] at (axis description cs:0.5,0.99)
  {Shape Strength: $[0.3,0.4]$};
\addplot[color=Paired-D,densely dashed,forget plot] coordinates {(0,0.95) (20,0.95)};
\addplot table[x=num_mon,y=coverage] from \MonHM;
\end{groupplot}
% Right axes: Standard Deviation and CI Length
\begin{groupplot}[GroupCommon,
    group style={group size=3 by 2,vertical sep=1.45cm},
    axis y line*=right,
    axis x line=none,
    hide x axis,
]
\nextgroupplot[
    ymin=0.067298628,ymax=0.072466711,ytick={0.067298628,0.069882670,0.072466711},
    y axis line style={color=Dark2-B},
    ytick style={color=Dark2-B},
    every y tick label/.append style={color=Dark2-B,/pgf/number format/precision=3,/pgf/number format/fixed zerofill},
    cycle list shift=1,
    legend style={legend to name=LegStd}
]
\addplot table[x=num_mon,y=std] from \MonLM;
\addlegendentry{Std (right axis)}
\nextgroupplot[
    ymin=0.067774701,ymax=0.073929027,ytick={0.067774701,0.070851864,0.073929027},
    y axis line style={color=Dark2-B},
    ytick style={color=Dark2-B},
    every y tick label/.append style={color=Dark2-B,/pgf/number format/precision=3,/pgf/number format/fixed zerofill},
    cycle list shift=1
]
\addplot table[x=num_mon,y=std] from \MonMM;
\nextgroupplot[
    ymin=0.068737326,ymax=0.076330530,ytick={0.068737326,0.072533928,0.076330530},
    y axis line style={color=Dark2-B},
    ytick style={color=Dark2-B},
    every y tick label/.append style={color=Dark2-B,/pgf/number format/precision=3,/pgf/number format/fixed zerofill},
    cycle list shift=1
]
\addplot table[x=num_mon,y=std] from \MonHM;
\nextgroupplot[
    ymin=0.263813094,ymax=0.279873857,ytick={0.263813094,0.271843476,0.279873857},
    y axis line style={color=Paired-J},
    ytick style={color=Paired-J},
    every y tick label/.append style={color=Paired-J,/pgf/number format/precision=3,/pgf/number format/fixed zerofill},
    cycle list shift=3,
    legend style={legend to name=LegLen}
]
\addplot table[x=num_mon,y=length] from \MonLM;
\addlegendentry{Length of CI (right axis)}
\nextgroupplot[
    ymin=0.265112955,ymax=0.286460891,ytick={0.265112955,0.275786923,0.286460891},
    y axis line style={color=Paired-J},
    ytick style={color=Paired-J},
    every y tick label/.append style={color=Paired-J,/pgf/number format/precision=3,/pgf/number format/fixed zerofill},
    cycle list shift=3
]
\addplot table[x=num_mon,y=length] from \MonMM;
\nextgroupplot[
    ymin=0.267498463,ymax=0.295154790,ytick={0.267498463,0.281326627,0.295154790},
    y axis line style={color=Paired-J},
    ytick style={color=Paired-J},
    every y tick label/.append style={color=Paired-J,/pgf/number format/precision=3,/pgf/number format/fixed zerofill},
    cycle list shift=3
]
\addplot table[x=num_mon,y=length] from \MonHM;
\end{groupplot}
\node[anchor=north] at ($(mon c2r1.south)+(0,-0.45cm)$) {%
  \pgfplotslegendfromname{LegBias}\hspace{1.2em}\pgfplotslegendfromname{LegStd}%
};
\node[anchor=north] at ($(mon c2r2.south)+(0,-0.5cm)$) {%
  \pgfplotslegendfromname{LegCov}\hspace{1.2em}\pgfplotslegendfromname{LegLen}%
};
\end{tikzpicture}
\caption{Debiased ATE under Monotonicity.} \label{Fig: NP Mon}
\end{figure}
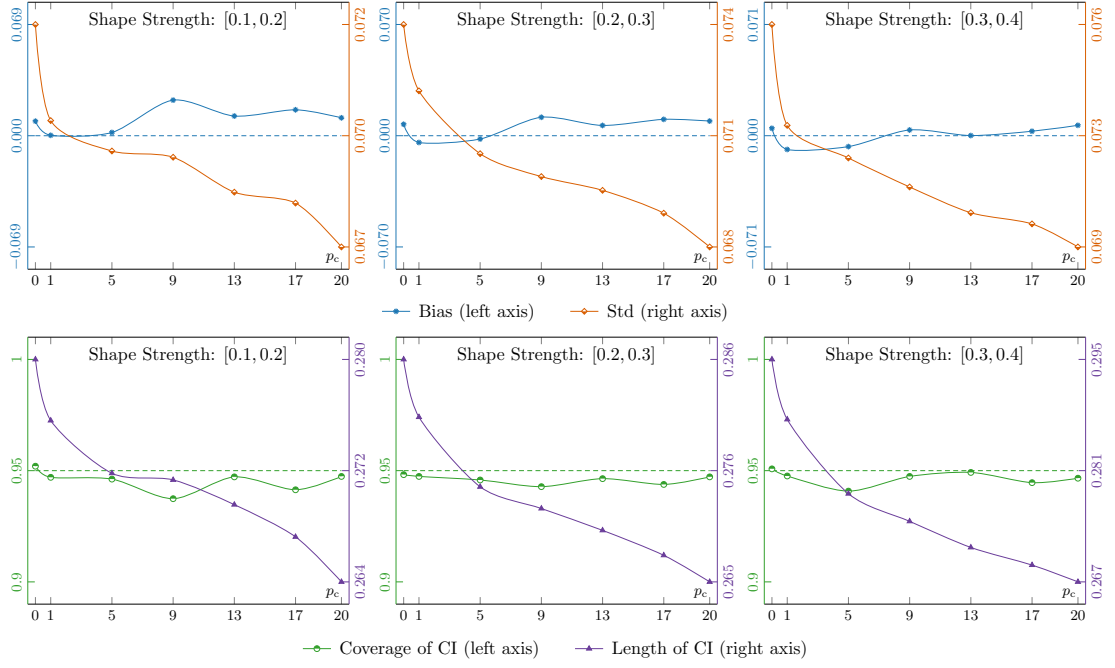

%%% Convexity results: selected second row from the 3x3 convexity groupplot
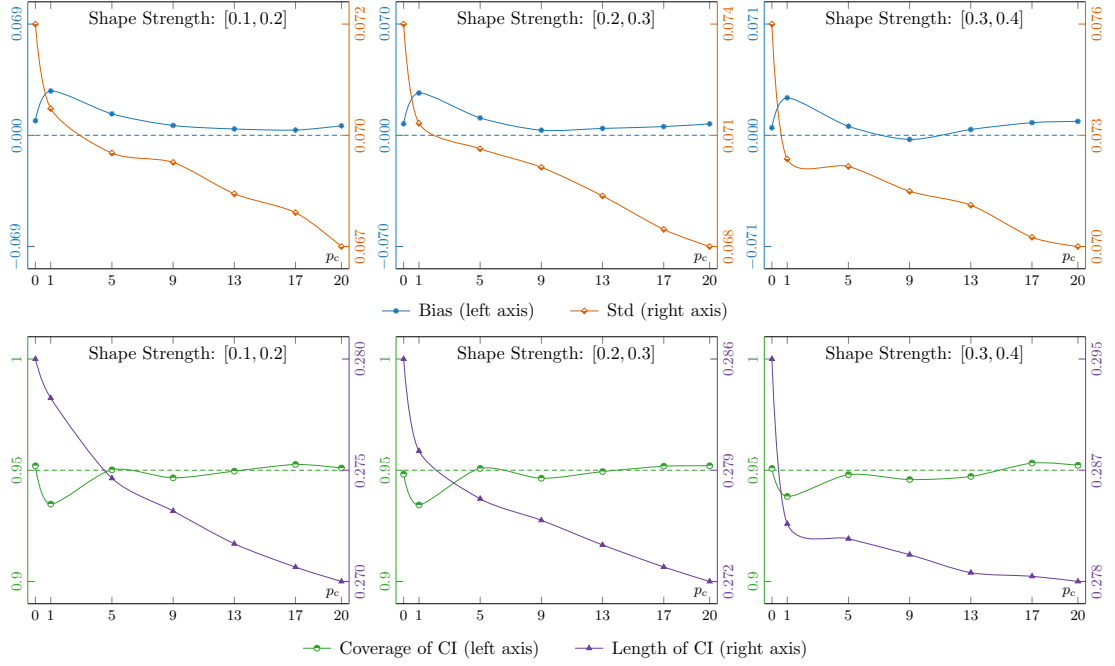
\begin{figure}[!ht]
\centering
\begin{tikzpicture}[scale=0.62, transform shape]
% Left axes: Bias and Coverage
\begin{groupplot}[GroupCommon,
    group style={group name=con,group size=3 by 2,vertical sep=1.45cm},
    xlabel={$p_{\mathrm{c}}$},
    axis y line*=left,
]
\nextgroupplot[
    ymin=-0.068962521,ymax=0.068962521,ytick={-0.068962521,0,0.068962521},
    y axis line style={color=Paired-B},
    ytick style={color=Paired-B},
    every y tick label/.append style={color=Paired-B},
    legend style={legend to name=LegConBias}
]
\node[anchor=north,align=center] at (axis description cs:0.5,0.99)
  {Shape Strength: $[0.1,0.2]$};
\addplot[color=Paired-B,densely dashed,forget plot] coordinates {(0,0) (20,0)};
\addplot table[x=num_mon,y=bias] from \ConLM;
\addlegendentry{Bias (left axis)}
\nextgroupplot[
    ymin=-0.069744239,ymax=0.069744239,ytick={-0.069744239,0,0.069744239},
    y axis line style={color=Paired-B},
    ytick style={color=Paired-B},
    every y tick label/.append style={color=Paired-B},
]
\node[anchor=north,align=center] at (axis description cs:0.5,0.99)
  {Shape Strength: $[0.2,0.3]$};
\addplot[color=Paired-B,densely dashed,forget plot] coordinates {(0,0) (20,0)};
\addplot table[x=num_mon,y=bias] from \ConMM;
\nextgroupplot[
    ymin=-0.071204162,ymax=0.071204162,ytick={-0.071204162,0,0.071204162},
    y axis line style={color=Paired-B},
    ytick style={color=Paired-B},
    every y tick label/.append style={color=Paired-B},
]
\node[anchor=north,align=center] at (axis description cs:0.5,0.99)
  {Shape Strength: $[0.3,0.4]$};
\addplot[color=Paired-B,densely dashed,forget plot] coordinates {(0,0) (20,0)};
\addplot table[x=num_mon,y=bias] from \ConHM;
\nextgroupplot[
    ymin=0.9,ymax=1.0,ytick={0.9,0.95,1.0},yticklabels={0.9,0.95,1},
    y axis line style={color=Paired-D},
    ytick style={color=Paired-D},
    every y tick label/.append style={color=Paired-D},
    cycle list shift=2,
    legend style={legend to name=LegConCov}
]
\node[anchor=north,align=center] at (axis description cs:0.5,0.99)
  {Shape Strength: $[0.1,0.2]$};
\addplot[color=Paired-D,densely dashed,forget plot] coordinates {(0,0.95) (20,0.95)};
\addplot table[x=num_mon,y=coverage] from \ConLM;
\addlegendentry{Coverage of CI (left axis)}
\nextgroupplot[
    ymin=0.9,ymax=1.0,ytick={0.9,0.95,1.0},yticklabels={0.9,0.95,1},
    y axis line style={color=Paired-D},
    ytick style={color=Paired-D},
    every y tick label/.append style={color=Paired-D},
    cycle list shift=2
]
\node[anchor=north,align=center] at (axis description cs:0.5,0.99)
  {Shape Strength: $[0.2,0.3]$};
\addplot[color=Paired-D,densely dashed,forget plot] coordinates {(0,0.95) (20,0.95)};
\addplot table[x=num_mon,y=coverage] from \ConMM;
\nextgroupplot[
    ymin=0.9,ymax=1.0,ytick={0.9,0.95,1.0},yticklabels={0.9,0.95,1},
    y axis line style={color=Paired-D},
    ytick style={color=Paired-D},
    every y tick label/.append style={color=Paired-D},
    cycle list shift=2
]
\node[anchor=north,align=center] at (axis description cs:0.5,0.99)
  {Shape Strength: $[0.3,0.4]$};
\addplot[color=Paired-D,densely dashed,forget plot] coordinates {(0,0.95) (20,0.95)};
\addplot table[x=num_mon,y=coverage] from \ConHM;
\end{groupplot}
% Right axes: Standard Deviation and CI Length
\begin{groupplot}[GroupCommon,
    group style={group size=3 by 2,vertical sep=1.45cm},
    xlabel={$p_{\mathrm{c}}$},
    axis y line*=right,
    axis x line=none,
    hide x axis,
]
\nextgroupplot[
    ymin=0.067105004,ymax=0.072466711,ytick={0.067105004,0.069785858,0.072466711},
    y axis line style={color=Dark2-B},
    ytick style={color=Dark2-B},
    every y tick label/.append style={color=Dark2-B,/pgf/number format/precision=3,/pgf/number format/fixed zerofill},
    cycle list shift=1,
    legend style={legend to name=LegConStd}
]
\addplot table[x=num_mon,y=std] from \ConLM;
\addlegendentry{Std (right axis)}
\nextgroupplot[
    ymin=0.068319208,ymax=0.073929027,ytick={0.068319208,0.071124118,0.073929027},
    y axis line style={color=Dark2-B},
    ytick style={color=Dark2-B},
    every y tick label/.append style={color=Dark2-B,/pgf/number format/precision=3,/pgf/number format/fixed zerofill},
    cycle list shift=1
]
\addplot table[x=num_mon,y=std] from \ConMM;
\nextgroupplot[
    ymin=0.069606871,ymax=0.076330530,ytick={0.069606871,0.072968701,0.076330530},
    y axis line style={color=Dark2-B},
    ytick style={color=Dark2-B},
    every y tick label/.append style={color=Dark2-B,/pgf/number format/precision=3,/pgf/number format/fixed zerofill},
    cycle list shift=1
]
\addplot table[x=num_mon,y=std] from \ConHM;
\nextgroupplot[
    ymin=0.269618247,ymax=0.279873857,ytick={0.269618247,0.274746052,0.279873857},
    y axis line style={color=Paired-J},
    ytick style={color=Paired-J},
    every y tick label/.append style={color=Paired-J,/pgf/number format/precision=3,/pgf/number format/fixed zerofill},
    cycle list shift=3,
    legend style={legend to name=LegConLen}
]
\addplot table[x=num_mon,y=length] from \ConLM;
\addlegendentry{Length of CI (right axis)}
\nextgroupplot[
    ymin=0.272343709,ymax=0.286460891,ytick={0.272343709,0.279402300,0.286460891},
    y axis line style={color=Paired-J},
    ytick style={color=Paired-J},
    every y tick label/.append style={color=Paired-J,/pgf/number format/precision=3,/pgf/number format/fixed zerofill},
    cycle list shift=3
]
\addplot table[x=num_mon,y=length] from \ConMM;
\nextgroupplot[
    ymin=0.278086790,ymax=0.295154790,ytick={0.278086790,0.286620790,0.295154790},
    y axis line style={color=Paired-J},
    ytick style={color=Paired-J},
    every y tick label/.append style={color=Paired-J,/pgf/number format/precision=3,/pgf/number format/fixed zerofill},
    cycle list shift=3
]
\addplot table[x=num_mon,y=length] from \ConHM;
\end{groupplot}
\node[anchor=north] at ($(con c2r1.south)+(0,-0.45cm)$) {%
  \pgfplotslegendfromname{LegConBias}\hspace{1.2em}\pgfplotslegendfromname{LegConStd}%
};
\node[anchor=north] at ($(con c2r2.south)+(0,-0.5cm)$) {%
  \pgfplotslegendfromname{LegConCov}\hspace{1.2em}\pgfplotslegendfromname{LegConLen}%
};
\end{tikzpicture}
\caption{Debiased ATE under Concavity.} \label{Fig: NP Con}
\end{figure}

Figures \ref{Fig: NP Mon} and \ref{Fig: NP Con} summarize the ATE simulations for the baseline signal-to-noise ratio based on different supports for the coefficients $\mathsf a_{\mu,\mathrm{c}}$, $\mathsf b_{\mu,\mathrm{c}}$, $\mathsf a_{\tau,\mathrm{c}}$, and $\mathsf b_{\tau,\mathrm{c}}$. In both monotonicity and concavity designs, the bias remains small relative to sampling variability, coverage is close to the nominal level, and the standard deviation and confidence-interval length generally decrease as more shape restrictions are imposed. These patterns suggest that shape restrictions can serve as substantive regularization, improving estimation precision and inference accuracy while maintaining small bias.

We emphasize that a uniform tuning scheme is used across the number $p_{\mathrm{c}}$ of constrained inputs. In particular, the choices of network architectures, dropout rates, learning rates, early stopping criteria, and penalty levels are neither data-driven nor necessarily optimal. The resulting performance is therefore particularly encouraging. It would be of interest to develop principled guidance for selecting these tuning parameters that both admits theoretical guarantees and performs well empirically.

\subsubsection{Functionals of NPIV Regression}

Let $(D,Z_{\mathrm{c}}^{\intercal},Z_{\mathrm{f}}^{\intercal})^{\intercal}\in\mathbf{R}^{1+d_{\mathrm{c}}+d_{\mathrm{f}}}$ be the vector of endogenous inputs and $(S,W_{\mathrm{c}}^{\intercal},W_{\mathrm{f}}^{\intercal})^{\intercal}\in\mathbf{R}^{1+d_{\mathrm{c}}+d_{\mathrm{f}}}$ the vector of instruments such that
\begin{align}
Y=\gamma_0(D,Z_{\mathrm{c}},Z_{\mathrm{f}})+U,\qquad E[U\mid S, W_{\mathrm{c}}, W_{\mathrm{f}}]=0~.
\end{align}
To build the data generating process, we let $U_0,V_0,U_{\mathrm{c}},U_{\mathrm{f}}$, and $\{V_{\mathrm{c},j},V_{\mathrm{f},j}\colon j\ge 1\}$ be independent standard normal random variables and set  $U=(U_0+U_{\mathrm{c}}+U_{\mathrm{f}})/\sqrt{3}$, 
\begin{align}
D=\Phi\!\left(\frac{\eta U_0+V_0}{\sqrt{1+\eta^2}}\right)~,\,Z_{\mathrm{c},1}=\Phi\!\left(\frac{\eta U_{\mathrm{c}}+V_{\mathrm{c},1}}{\sqrt{1+\eta^2}}\right)~,\,Z_{\mathrm{f},1}=\Phi\!\left(\frac{\eta U_{\mathrm{f}}+V_{\mathrm{f},1}}{\sqrt{1+\eta^2}}\right)~,
\end{align}
$S=\Phi(V_0)$, $W_{\mathrm{c},1}=\Phi(V_{\mathrm{c},1})$, $W_{\mathrm{f},1}=\Phi(V_{\mathrm{f},1})$, and, for $j\ge2$,
\begin{align}
Z_{\mathrm{c},j}=W_{\mathrm{c},j}
    =\Phi\!\left(\frac{\rho(V_0+V_{\mathrm{c},1}+V_{\mathrm{f},1})}{\sqrt{3}}
      +\sqrt{1-\rho^2}\,V_{\mathrm{c},j}\right)~,\notag\\
Z_{\mathrm{f},j}=W_{\mathrm{f},j}
    =\Phi\!\left(\frac{\rho(V_0+V_{\mathrm{c},1}+V_{\mathrm{f},1})}{\sqrt{3}}
      +\sqrt{1-\rho^2}\,V_{\mathrm{f},j}\right)~.      
\end{align}
Intuitively, $\eta$ controls the strength of endogeneity while $\rho$ controls the correlation between $\{D,Z_{\mathrm{c},1},Z_{\mathrm{f},1}\}$ and $\{Z_{\mathrm{c},j},Z_{\mathrm{f},j}\colon j\geq 2\}$. In turn, we set
\begin{align}
    \gamma_0(D,Z_{\mathrm{c}},Z_{\mathrm{f}})
    =D+[\mathsf a_{\mathrm{c}}^\intercal Z_{\mathrm{c}}+\Phi(\mathsf b_{\mathrm{c}}^\intercal Z_{\mathrm{c}})]
       \Phi(\mathsf a_{\mathrm{f}}^\intercal Z_{\mathrm{f}}+\mathsf b_{\mathrm{f}}^\intercal q(Z_{\mathrm{f}}))~,
\end{align}
where $\mathsf a_{\mathrm{c}}$ and $\mathsf b_{\mathrm{c}}$ have i.i.d.\ entries drawn from $\mathrm{Uni}[0.2,0.3]$, $\mathsf a_{\mathrm{f}}$ has i.i.d.\ entries drawn from $\mathrm{Uni}[-0.1,0.1]$, and
$\mathsf b_{\mathrm{f}}$ has i.i.d.\ entries drawn from $\mathrm{Uni}[-0.05,0.05]$. By construction, $\gamma_0$ is monotonic in $Z_{\mathrm{c}}$, and in this NPIV design we focus on monotonicity. 

The parameter of interest is the weighted average derivative: for $Z\equiv(Z_{\mathrm{c}}^{\intercal},Z_{\mathrm{f}}^{\intercal})^\intercal$,
\begin{align}
  \theta_0\equiv E[w(D)\nabla_d\gamma_0(D,Z)]~,
\end{align}
where $w$ is a weight function and $\nabla_d$ indicates the partial derivative with respect to $D$. Intuitively, $\theta_0$ measures the average marginal/treatment effect of a small change in $D$ on the outcome $Y$. Suppose that $D$ is supported on $[a,b]$ with $-\infty<a<b<\infty$ and that $d\mapsto w(d)f_0(d,z)$ vanishes on $a$ and $b$ for (almost) every $z$ with $f_0$ the density of $(D,Z)$. Then we obtain under regularity conditions that 
\begin{align}
  \theta_0 = -E[w(D)\nabla_d\ln(w(D)f_0(D,Z))\gamma_0(D,Z)]~,
\end{align}
which is continuous and linear in $\gamma_0$ if $w(D)\nabla_d\ln(w(D)f_0(D,Z))\in L^2(D,Z)$. 

The main comparative statics vary the endogeneity level $\eta$ from $0.1$ to $1.0$ based on $\rho=0.1$. We set $w$ to be the normalized indicator of the $[0.1,0.9]$ quantile region of $D$ so that $\theta_0=1$. Relative to the ATE designs, the additional unknown object is the projection operator $\Pi_0$ given by $\Pi_0(\delta)=E[\delta(D,Z_{\mathrm{c}},Z_{\mathrm{f}})|S,W_{\mathrm{c}},W_{\mathrm{f}}]$ for any $\delta\in L^2(D,Z_{\mathrm{c}},Z_{\mathrm{f}})$. We obtain $\hat\Pi_n(\delta)$ by a ridge regression of $\delta$ on a vector of basis functions of $(S,W_{\mathrm{c}},W_{\mathrm{f}})$. We consider two constructions of these basis functions. The first one is an adaptively learned basis inspired by \citet{Bruns2025TSML}. Specifically, we fit a deep neural network to predict $Y$ from $(S,W_{\mathrm{c}},W_{\mathrm{f}})$ and employ the activations in its final hidden layer as basis functions. The second one is a specific choice of spline basis functions used in 
\citet[p.~1862]{ChenChenTamer2023NPIV}. 

To ease computation, we consider two specifications, one with $d_{\mathrm{c}}=d_{\mathrm{f}}=1$ and one with $d_{\mathrm{c}}=d_{\mathrm{f}}=5$, and then draw 500 independent i.i.d.\ samples of size 1000 for each design. We estimate the nuisances $\gamma_0$ and $\alpha_0$ using the same partially monotonic neural network
architecture as in Section \ref{Sec: Sims, ATE}, varying the number of imposed constrained inputs over $p_{\mathrm{c}}=0,1,\ldots,d_{\mathrm{c}}$. Training of the networks is carried out using the same AdamW optimizer, dropout, early stopping, and $\ell_2$-penalization as in the ATE simulations. For each sample, we construct $K=5$ folds for cross-fitting. As before, we benchmark our results against DML via Riesz regression but without shape constraints, which corresponds to $p_{\mathrm{c}}=0$ in the plots below.

%Figure \ref{Fig:NPIVMonotonicity} reports the simulation results. Overall, imposing monotonicity substantially improves coverage probability while reducing bias and mean squared error across endogeneity levels relative to the unconstrained benchmark. When $d_{\mathrm{c}}=d_{\mathrm{f}}=1$, adding the single monotonicity restriction raises coverage from about $50\%$ to close to the nominal $95\%$ level. When $d_{\mathrm{c}}=d_{\mathrm{f}}=5$, coverage probabilities increase substantially as more valid monotonicity restrictions are imposed, rising from below $10\%$ for the unconstrained benchmark to roughly 50--80\% when all five monotonicity restrictions are imposed. The bias and mean squared error also decline with higher values of $p_{\mathrm{c}}$, indicating clear gains in estimation accuracy in the high dimensional case. We note that achieving accurate coverage probabilities in high dimensional NPIV models is notoriously difficult---see, e.g., \citet{BennettKallusMaoNeweySyrgkanisUehara2025StrongID} and \citet{Bruns2025TSML} for related simulation evidence. Our results suggest that imposing shape constraints can help mitigate this challenge in these high dimensional inverse problems.

\pgfplotstableread{
eta pm0 pm1
0.1 0.504000000 0.954000000
0.2 0.506000000 0.948000000
0.3 0.498000000 0.948000000
0.4 0.496000000 0.942000000
0.5 0.496000000 0.952000000
0.6 0.510000000 0.940000000
0.7 0.498000000 0.928000000
0.8 0.490000000 0.942000000
0.9 0.504000000 0.942000000
1.0 0.524000000 0.950000000
}\NPIVOneCoverage

\pgfplotstableread{
eta pm0 pm1
0.1 -0.062901740 -0.029790198
0.2 -0.067800841 -0.028454479
0.3 -0.072203688 -0.025911979
0.4 -0.077364539 -0.029073808
0.5 -0.082567237 -0.027291200
0.6 -0.092458112 -0.025538096
0.7 -0.098472011 -0.023261110
0.8 -0.108009344 -0.023605807
0.9 -0.114783656 -0.022755386
1.0 -0.122451760 -0.025255426
}\NPIVOneBias

\pgfplotstableread{
eta pm0 pm1
0.1 0.029913114 0.013890082
0.2 0.030633043 0.014574314
0.3 0.033374373 0.015511749
0.4 0.036130226 0.016587384
0.5 0.040420215 0.017691700
0.6 0.045231303 0.020028298
0.7 0.050188412 0.022012582
0.8 0.055315317 0.023618394
0.9 0.059976672 0.026304121
1.0 0.064264659 0.029091482
}\NPIVOneMSE

\pgfplotstableread{
eta pm0 pm1 pm3 pm5
0.1 0.092000000 0.698000000 0.748000000 0.794000000
0.2 0.074000000 0.688000000 0.742000000 0.810000000
0.3 0.070000000 0.666000000 0.726000000 0.774000000
0.4 0.052000000 0.628000000 0.692000000 0.752000000
0.5 0.042000000 0.622000000 0.694000000 0.708000000
0.6 0.032000000 0.594000000 0.658000000 0.704000000
0.7 0.034000000 0.568000000 0.622000000 0.678000000
0.8 0.032000000 0.542000000 0.588000000 0.654000000
0.9 0.024000000 0.526000000 0.570000000 0.624000000
1.0 0.016000000 0.488000000 0.554000000 0.572000000
}\NPIVFiveCoverage

\pgfplotstableread{
eta pm0 pm1 pm3 pm5
0.1 -0.261582946 -0.057512979 -0.046756125 -0.040874387
0.2 -0.277904062 -0.064951272 -0.049404491 -0.041663732
0.3 -0.302376437 -0.070214754 -0.056650298 -0.049620320
0.4 -0.327251424 -0.075996287 -0.064755355 -0.054990715
0.5 -0.361374658 -0.085708130 -0.075923315 -0.066349879
0.6 -0.399536190 -0.099981222 -0.086775860 -0.078157370
0.7 -0.435843352 -0.115162929 -0.099078348 -0.091002018
0.8 -0.470035469 -0.132825014 -0.115458741 -0.103403743
0.9 -0.501143920 -0.154898285 -0.131658743 -0.117845826
1.0 -0.533352633 -0.178865136 -0.151951597 -0.136498229
}\NPIVFiveBias

\pgfplotstableread{
eta pm0 pm1 pm3 pm5
0.1 0.109574737 0.021298294 0.019640862 0.018855358
0.2 0.119968799 0.023154430 0.020243381 0.018749997
0.3 0.137380778 0.024828199 0.022833140 0.021257430
0.4 0.155227120 0.027946675 0.026323367 0.023673684
0.5 0.179837904 0.031125245 0.029017552 0.028058314
0.6 0.210131340 0.036028697 0.033435727 0.032539123
0.7 0.241718314 0.042074721 0.040455180 0.038274088
0.8 0.274076424 0.049555720 0.047203496 0.044008821
0.9 0.302487071 0.059079170 0.055331926 0.051948342
1.0 0.334651062 0.070926121 0.064312499 0.061423557
}\NPIVFiveMSE

\pgfplotsset{
  NPIVGroupCommon/.style={
    group style={horizontal sep=0.8cm,vertical sep=1.25cm},
    scaled y ticks=false,
    xmin=0.1, xmax=1.0,
    xtick={0.1,0.2,0.3,0.4,0.5,0.6,0.7,0.8,0.9,1.0},
    xlabel={$\eta$},
    x label style={at={(axis description cs:0.95,0)},anchor=south,font=\footnotesize},
    tick label style={/pgf/number format/fixed,font=\footnotesize},
    yticklabel style={rotate=90,/pgf/number format/fixed,/pgf/number format/precision=3,/pgf/number format/fixed zerofill},
    enlarge x limits={abs=0.02},
    enlarge y limits,
    legend columns=-1,
    legend style={cells={align=center},column sep=6pt,row sep=0pt,draw=none},
    cycle list={
      {smooth,tension=0.5,color=Paired-B, mark=10-pointed star,mark size=1.75pt,line width=0.5pt},
      {smooth,tension=0.5,color=Dark2-B, mark=halfsquare*,mark size=1.75pt,line width=0.5pt},
      {smooth,tension=0.5,color=Paired-D, mark=halfcircle*,mark size=1.75pt,line width=0.5pt},
      {smooth,tension=0.5,color=Paired-J, mark=triangle*,mark size=1.75pt,line width=0.5pt},
    }
  },
}

\begin{figure}[!ht]
\centering
\begin{tikzpicture}[scale=0.61, transform shape]
\begin{groupplot}[NPIVGroupCommon,
    group style={group name=npiv,group size=3 by 2,vertical sep=1.45cm},
]
\nextgroupplot[
    title={Coverage},
    ymin=0,ymax=1.05,ytick={0,0.5,0.95},yticklabels={0,0.5,0.95},
    yticklabel style={rotate=0,font=\footnotesize},
    legend style={legend to name=NPIVLegOne}
]
\node[anchor=north,align=center] at (axis description cs:0.5,0.99)
  {$d_{\mathrm{c}}=d_{\mathrm{f}}=1$};
\addplot table[x=eta,y=pm0] from \NPIVOneCoverage;
\addlegendentry{$p_{\mathrm{c}}=0$}
\addplot table[x=eta,y=pm1] from \NPIVOneCoverage;
\addlegendentry{$p_{\mathrm{c}}=1$}
\addplot[densely dashed,black,forget plot] coordinates {(0.1,0.95) (1.0,0.95)};

\nextgroupplot[
    title={Bias},
    ymin=-0.13,ymax=0.01,ytick={-0.12,-0.06,0}
]
\node[anchor=north,align=center] at (axis description cs:0.5,0.99)
  {$d_{\mathrm{c}}=d_{\mathrm{f}}=1$};
\addplot table[x=eta,y=pm0] from \NPIVOneBias;
\addplot table[x=eta,y=pm1] from \NPIVOneBias;
\addplot[densely dashed,black,forget plot] coordinates {(0.1,0) (1.0,0)};

\nextgroupplot[
    title={Mean Squared Error},
    ymin=0,ymax=0.07,ytick={0,0.035,0.07}
]
\node[anchor=north,align=center] at (axis description cs:0.5,0.99)
  {$d_{\mathrm{c}}=d_{\mathrm{f}}=1$};
\addplot table[x=eta,y=pm0] from \NPIVOneMSE;
\addplot table[x=eta,y=pm1] from \NPIVOneMSE;
\addplot[densely dashed,black,forget plot] coordinates {(0.1,0) (1.0,0)};

\nextgroupplot[
    ymin=0,ymax=1.05,ytick={0,0.5,0.95},yticklabels={0,0.5,0.95},
    yticklabel style={rotate=0,font=\footnotesize},
    legend style={legend to name=NPIVLegFive}
]
\node[anchor=north,align=center] at (axis description cs:0.5,0.99)
  {$d_{\mathrm{c}}=d_{\mathrm{f}}=5$};
\addplot table[x=eta,y=pm0] from \NPIVFiveCoverage;
\addlegendentry{$p_{\mathrm{c}}=0$}
\addplot table[x=eta,y=pm1] from \NPIVFiveCoverage;
\addlegendentry{$p_{\mathrm{c}}=1$}
\addplot table[x=eta,y=pm3] from \NPIVFiveCoverage;
\addlegendentry{$p_{\mathrm{c}}=3$}
\addplot table[x=eta,y=pm5] from \NPIVFiveCoverage;
\addlegendentry{$p_{\mathrm{c}}=5$}
\addplot[densely dashed,black,forget plot] coordinates {(0.1,0.95) (1.0,0.95)};

\nextgroupplot[
    ymin=-0.55,ymax=0.02,ytick={-0.5,-0.25,0}
]
\node[anchor=north,align=center] at (axis description cs:0.5,0.99)
  {$d_{\mathrm{c}}=d_{\mathrm{f}}=5$};
\addplot table[x=eta,y=pm0] from \NPIVFiveBias;
\addplot table[x=eta,y=pm1] from \NPIVFiveBias;
\addplot table[x=eta,y=pm3] from \NPIVFiveBias;
\addplot table[x=eta,y=pm5] from \NPIVFiveBias;
\addplot[densely dashed,black,forget plot] coordinates {(0.1,0) (1.0,0)};

\nextgroupplot[
    ymin=0,ymax=0.35,ytick={0,0.175,0.35}
]
\node[anchor=north,align=center] at (axis description cs:0.5,0.99)
  {$d_{\mathrm{c}}=d_{\mathrm{f}}=5$};
\addplot table[x=eta,y=pm0] from \NPIVFiveMSE;
\addplot table[x=eta,y=pm1] from \NPIVFiveMSE;
\addplot table[x=eta,y=pm3] from \NPIVFiveMSE;
\addplot table[x=eta,y=pm5] from \NPIVFiveMSE;
\addplot[densely dashed,black,forget plot] coordinates {(0.1,0) (1.0,0)};
\end{groupplot}
\node[anchor=north] at ($(npiv c2r1.south)+(0,-0.45cm)$) {%
  \pgfplotslegendfromname{NPIVLegOne}%
};
\node[anchor=north] at ($(npiv c2r2.south)+(0,-0.5cm)$) {%
  \pgfplotslegendfromname{NPIVLegFive}%
};
\end{tikzpicture}
\caption{Debiased Functionals of NPIV under Monotonicity: Adaptive Basis.}
\label{Fig:NPIVMonotonicity}
\end{figure}
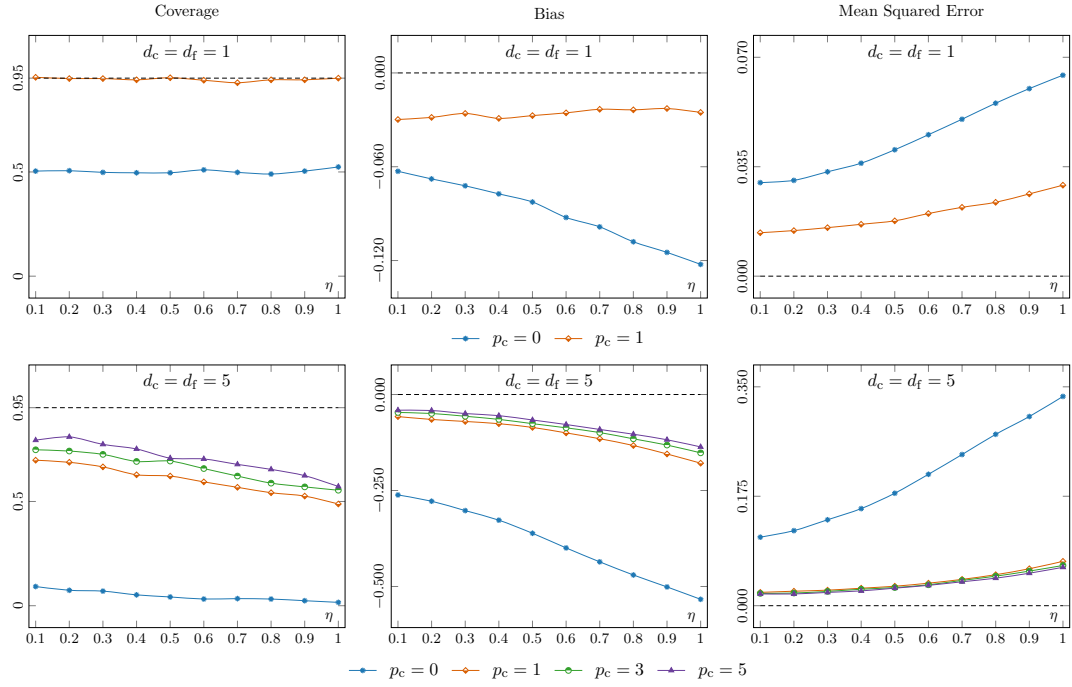

Figure \ref{Fig:NPIVMonotonicity} reports the simulation results based on the adaptive basis. Overall, imposing monotonicity substantially improves coverage probability while reducing bias and mean squared error across endogeneity levels relative to the unconstrained benchmark. When $d_{\mathrm{c}}=d_{\mathrm{f}}=1$, adding the single monotonicity restriction raises coverage from about $50\%$ to close to the nominal $95\%$ level. When $d_{\mathrm{c}}=d_{\mathrm{f}}=5$, coverage probabilities increase substantially as more monotonicity restrictions are imposed, rising from below $10\%$ for the unconstrained benchmark to roughly 50--80\% when all five monotonicity restrictions are imposed. The bias and mean squared error also decline with higher values of $p_{\mathrm{c}}$, indicating clear gains in estimation accuracy in the high dimensional case. Figure \ref{Fig:NPIVMonotonicitySpline} shows that the spline basis implementation delivers similar findings. We note that achieving accurate coverage probabilities in high dimensional NPIV models is notoriously difficult---see, e.g., \citet{BennettKallusMaoNeweySyrgkanisUehara2025StrongID} and \citet{Bruns2025TSML} for related simulation evidence. Our results suggest that imposing shape constraints can help mitigate this challenge in these high dimensional inverse problems. 

\pgfplotstableread{
eta pm0 pm1
0.1 0.392000000 0.904000000
0.2 0.392000000 0.910000000
0.3 0.376000000 0.906000000
0.4 0.366000000 0.900000000
0.5 0.372000000 0.908000000
0.6 0.348000000 0.890000000
0.7 0.360000000 0.896000000
0.8 0.362000000 0.880000000
0.9 0.362000000 0.894000000
1.0 0.356000000 0.876000000
}\NPIVSplineOneCoverage

\pgfplotstableread{
eta pm0 pm1
0.1 -0.105117650 -0.004186915
0.2 -0.110102383 -0.004988590
0.3 -0.119382818 -0.008409777
0.4 -0.127541303 -0.007055514
0.5 -0.135488984 -0.010703606
0.6 -0.146909398 -0.016768821
0.7 -0.156998746 -0.020608792
0.8 -0.169641960 -0.024061935
0.9 -0.178635922 -0.028447377
1.0 -0.188881501 -0.037457928
}\NPIVSplineOneBias

\pgfplotstableread{
eta pm0 pm1
0.1 0.039067700 0.015671467
0.2 0.041118862 0.014893703
0.3 0.046010553 0.015872737
0.4 0.051120141 0.016665281
0.5 0.054559042 0.017371498
0.6 0.061881182 0.020172121
0.7 0.067125027 0.022230449
0.8 0.073843113 0.025237452
0.9 0.079824932 0.027235375
1.0 0.085710711 0.031353640
}\NPIVSplineOneMSE

\pgfplotstableread{
eta pm0 pm1 pm3 pm5
0.1 0.092000000 0.612000000 0.642000000 0.700000000
0.2 0.080000000 0.606000000 0.650000000 0.668000000
0.3 0.062000000 0.558000000 0.638000000 0.664000000
0.4 0.044000000 0.520000000 0.586000000 0.646000000
0.5 0.042000000 0.516000000 0.558000000 0.574000000
0.6 0.028000000 0.472000000 0.516000000 0.570000000
0.7 0.020000000 0.396000000 0.480000000 0.492000000
0.8 0.016000000 0.378000000 0.446000000 0.448000000
0.9 0.016000000 0.324000000 0.398000000 0.428000000
1.0 0.010000000 0.284000000 0.350000000 0.372000000
}\NPIVSplineFiveCoverage

\pgfplotstableread{
eta pm0 pm1 pm3 pm5
0.1 -0.272919339 -0.055662661 -0.050733471 -0.048711623
0.2 -0.289564225 -0.062605444 -0.053068998 -0.058637260
0.3 -0.318010973 -0.076724572 -0.062639810 -0.062741916
0.4 -0.352754749 -0.089463980 -0.077486646 -0.072945493
0.5 -0.383779067 -0.101679530 -0.090679725 -0.088697680
0.6 -0.427628349 -0.126041521 -0.111128489 -0.108037349
0.7 -0.472427572 -0.150115667 -0.130843867 -0.130986624
0.8 -0.511157828 -0.174339298 -0.156925132 -0.156850090
0.9 -0.548284021 -0.207918637 -0.181296992 -0.180934853
1.0 -0.580650960 -0.239008996 -0.215794017 -0.217189526
}\NPIVSplineFiveBias

\pgfplotstableread{
eta pm0 pm1 pm3 pm5
0.1 0.113631430 0.020202204 0.020279492 0.020462485
0.2 0.125668490 0.021904933 0.021499798 0.022526410
0.3 0.145563844 0.025318539 0.024270672 0.024128664
0.4 0.171730738 0.029544026 0.029059087 0.027937183
0.5 0.196141021 0.033825982 0.033054951 0.033350312
0.6 0.231426659 0.041128897 0.039113649 0.039851514
0.7 0.269379741 0.052006776 0.046892182 0.048758793
0.8 0.307598293 0.063650828 0.060019823 0.059927726
0.9 0.341110076 0.079129346 0.071728581 0.068506644
1.0 0.375602957 0.094874464 0.086709819 0.089780207
}\NPIVSplineFiveMSE

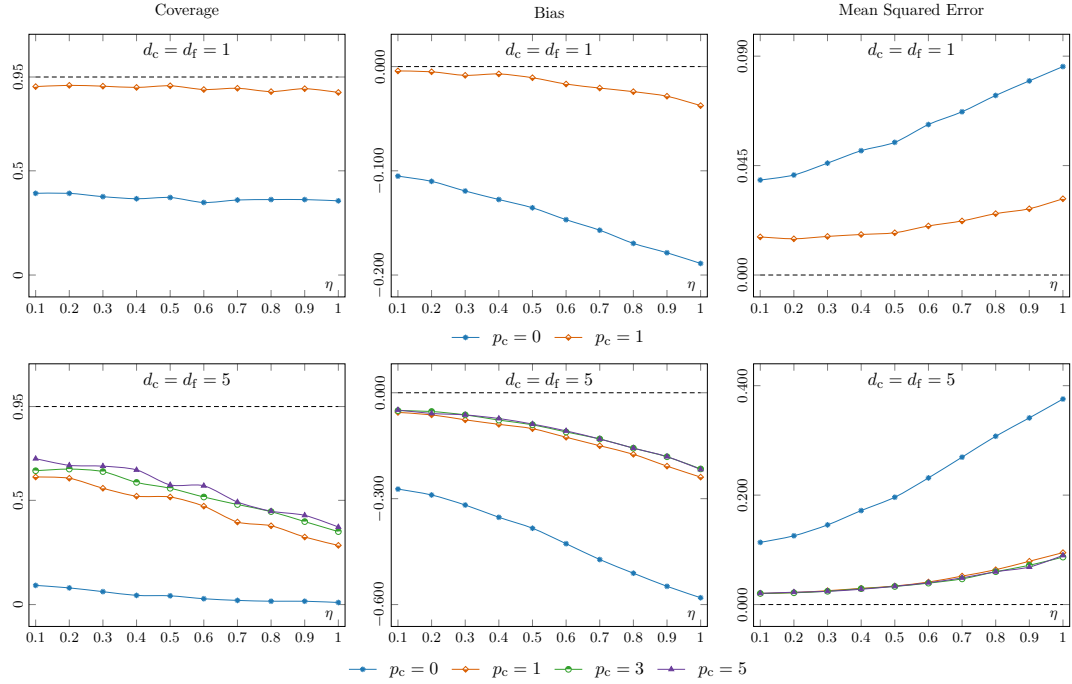
\begin{figure}[!ht]
\centering
\begin{tikzpicture}[scale=0.61, transform shape]
\begin{groupplot}[NPIVGroupCommon,
    group style={group name=npivspline,group size=3 by 2,vertical sep=1.45cm},
]
\nextgroupplot[
    title={Coverage},
    ymin=0,ymax=1.05,ytick={0,0.5,0.95},yticklabels={0,0.5,0.95},
    yticklabel style={rotate=0,font=\footnotesize},
    legend style={legend to name=NPIVSplineLegOne}
]
\node[anchor=north,align=center] at (axis description cs:0.5,0.99)
  {$d_{\mathrm{c}}=d_{\mathrm{f}}=1$};
\addplot table[x=eta,y=pm0] from \NPIVSplineOneCoverage;
\addlegendentry{$p_{\mathrm{c}}=0$}
\addplot table[x=eta,y=pm1] from \NPIVSplineOneCoverage;
\addlegendentry{$p_{\mathrm{c}}=1$}
\addplot[densely dashed,black,forget plot] coordinates {(0.1,0.95) (1.0,0.95)};

\nextgroupplot[
    title={Bias},
    ymin=-0.20,ymax=0.01,ytick={-0.20,-0.10,0}
]
\node[anchor=north,align=center] at (axis description cs:0.5,0.99)
  {$d_{\mathrm{c}}=d_{\mathrm{f}}=1$};
\addplot table[x=eta,y=pm0] from \NPIVSplineOneBias;
\addplot table[x=eta,y=pm1] from \NPIVSplineOneBias;
\addplot[densely dashed,black,forget plot] coordinates {(0.1,0) (1.0,0)};

\nextgroupplot[
    title={Mean Squared Error},
    ymin=0,ymax=0.09,ytick={0,0.045,0.09}
]
\node[anchor=north,align=center] at (axis description cs:0.5,0.99)
  {$d_{\mathrm{c}}=d_{\mathrm{f}}=1$};
\addplot table[x=eta,y=pm0] from \NPIVSplineOneMSE;
\addplot table[x=eta,y=pm1] from \NPIVSplineOneMSE;
\addplot[densely dashed,black,forget plot] coordinates {(0.1,0) (1.0,0)};

\nextgroupplot[
    ymin=0,ymax=1.05,ytick={0,0.5,0.95},yticklabels={0,0.5,0.95},
    yticklabel style={rotate=0,font=\footnotesize},
    legend style={legend to name=NPIVSplineLegFive}
]
\node[anchor=north,align=center] at (axis description cs:0.5,0.99)
  {$d_{\mathrm{c}}=d_{\mathrm{f}}=5$};
\addplot table[x=eta,y=pm0] from \NPIVSplineFiveCoverage;
\addlegendentry{$p_{\mathrm{c}}=0$}
\addplot table[x=eta,y=pm1] from \NPIVSplineFiveCoverage;
\addlegendentry{$p_{\mathrm{c}}=1$}
\addplot table[x=eta,y=pm3] from \NPIVSplineFiveCoverage;
\addlegendentry{$p_{\mathrm{c}}=3$}
\addplot table[x=eta,y=pm5] from \NPIVSplineFiveCoverage;
\addlegendentry{$p_{\mathrm{c}}=5$}
\addplot[densely dashed,black,forget plot] coordinates {(0.1,0.95) (1.0,0.95)};

\nextgroupplot[
    ymin=-0.60,ymax=0.02,ytick={-0.6,-0.3,0}
]
\node[anchor=north,align=center] at (axis description cs:0.5,0.99)
  {$d_{\mathrm{c}}=d_{\mathrm{f}}=5$};
\addplot table[x=eta,y=pm0] from \NPIVSplineFiveBias;
\addplot table[x=eta,y=pm1] from \NPIVSplineFiveBias;
\addplot table[x=eta,y=pm3] from \NPIVSplineFiveBias;
\addplot table[x=eta,y=pm5] from \NPIVSplineFiveBias;
\addplot[densely dashed,black,forget plot] coordinates {(0.1,0) (1.0,0)};

\nextgroupplot[
    ymin=0,ymax=0.40,ytick={0,0.2,0.4}
]
\node[anchor=north,align=center] at (axis description cs:0.5,0.99)
  {$d_{\mathrm{c}}=d_{\mathrm{f}}=5$};
\addplot table[x=eta,y=pm0] from \NPIVSplineFiveMSE;
\addplot table[x=eta,y=pm1] from \NPIVSplineFiveMSE;
\addplot table[x=eta,y=pm3] from \NPIVSplineFiveMSE;
\addplot table[x=eta,y=pm5] from \NPIVSplineFiveMSE;
\addplot[densely dashed,black,forget plot] coordinates {(0.1,0) (1.0,0)};
\end{groupplot}
\node[anchor=north] at ($(npivspline c2r1.south)+(0,-0.45cm)$) {%
  \pgfplotslegendfromname{NPIVSplineLegOne}%
};
\node[anchor=north] at ($(npivspline c2r2.south)+(0,-0.5cm)$) {%
  \pgfplotslegendfromname{NPIVSplineLegFive}%
};
\end{tikzpicture}
\caption{Debiased Functionals of NPIV under Monotonicity: Spline Basis.}
\label{Fig:NPIVMonotonicitySpline}
\end{figure}

\section{Empirical Applications}\label{Sec: Applications}

% Monopsony in Online Labor Markets
% Relinquishing riches: Auctions versus informal negotiations in texas oil and gas leasing
% \citet{DiNnardoFortinLemieux1996Wage}

In this section, we present two empirical applications. The first is a DiD analysis of Medicaid expansions and mortality, while the second studies the intertemporal relationship between working hours and wage growth. These applications illustrate the role of shape constraints as a middle ground between flexible nonparametric estimation and restrictive parametric modeling. Fully nonparametric methods can accommodate rich heterogeneity but may lack statistical guarantees in high dimensional settings. Parametric specifications are more tractable but rely on functional form assumptions that may be too restrictive. Shape constraints impose substantive restrictions that can improve precision while retaining substantial flexibility. 

% The two applications complement one another by illustrating how shape constraints can be motivated either by institutional knowledge or by economic theory.

\subsection{DiD under Shape Constraints}

We proceed by introducing the DiD setup and formalizing the estimation framework as a specialization of our general theory. We then bring the model to the data by revisiting the Medicaid and mortality example in \citet{BakerCallawayCunninghamGoodmanSantAnna2026Guide}. 

\subsubsection{The DiD Setup}

We consider a staggered treatment regime with a binary treatment, i.e., once a unit is treated, it remains treated in all subsequent periods. Thus, an individual's treatment path is fully characterized by the period in which treatment first occurs. Specifically, suppose that there are $T$ periods indexed by $t=1,\ldots,T$ and let $G$ be the time when an individual is first treated. Then individuals may be classified into different groups based on the support $\mathcal G$ of $G$. We assume that no individuals are treated at $t=1$ and set $G=\infty$ if an individual never receives the treatment. Without loss of generality, we also assume the existence of a never-treated group by, if necessary, discarding observations during the periods when the last cohort is treated; otherwise, treatment effects in those periods would not be identified due to the absence of comparison groups. Given the absorbing nature of the treatment, the path of potential outcomes for an individual from group $G=g$ may be denoted as $\{Y_t(g)\}_{t=1}^T$. The observed outcome $Y_t$ at time $t$ relates to the potential outcomes via $Y_t=\sum_{g\in\mathcal G} Y_t(g)1\{G=g\}$. In addition, we observe a vector $Z$ of time invariant (pre-treatment) covariates. 

Following \citet{CallawaySantAnna2021MultipleDID}, we define the average treatment effect for group $g$ at time $t$ (GT-ATT) by
\begin{align}
\theta_0(g,t)\equiv E[Y_t(g)-Y_t(\infty)|G=g]~. 
\end{align}
These GT-ATTs are of interest in their own right while also serving as building blocks for inference on more aggregated causal effects. For example, letting $e\equiv t-g$ be the event time, we may consider the event study parameter: 
\begin{align}\label{Eqn: DID event study}
\theta_{\mathrm{es}}(e)\equiv \sum_{g\in\mathcal G}P(G=g|G+e\le T)\theta_0(g,g+e)~.
\end{align}
Intuitively, $\theta_{\mathrm{es}}(e)$ measures the average treatment effect among groups that have been exposed to treatment for $e$ time periods. 

In order to identify $\theta_0(g,t)$, we must select a collection of comparison groups which we denote by $\mathcal C_{gt}$. A common choice is $\mathcal C_{gt}=\{s\in\mathcal G\colon s>g\vee t\}$, i.e., groups not yet treated at time $t$, although it may be desirable to choose $\mathcal C_{gt}$ as a proper subset of $\{s\in\mathcal G\colon s>g\vee t\}$. Suppose that (i) $Y_t(g)=Y_t(\infty)$ for all $t\le g-1$ and $g\in\mathcal G$ (no anticipation), (ii) $P(G=g|Z)>0$ for all $g\in\mathcal G$ (overlap), and (iii) conditional parallel trends hold in the sense that, for any $g\in\mathcal G\backslash\{\infty\}$, any $t\ge g$, and any $c\in\mathcal C_{gt}\subset \{s\in\mathcal G\colon s>g\vee t\}$,
\begin{align}
E[Y_t(\infty)-Y_{t-1}(\infty)|G=g,Z] = E[Y_t(\infty)-Y_{t-1}(\infty)|G=c,Z]~.
\end{align} 
Then it follows that, for any $g\in\mathcal G\backslash\{\infty\}$ and $t\ge g$,
\begin{align}\label{Eqn: DID OR}
\theta_0(g,t)= E[Y_t-Y_{g-1}-\mu_{gt,0}(1,Z)|G=g]~,
\end{align}
where $\mu_{gt,0}(C_{gt},Z)\equiv E[Y_t-Y_{g-1}|C_{gt},Z]$ for $C_{gt}\equiv 1\{G\in\mathcal C_{gt}\}$. Equation \eqref{Eqn: DID OR} is essentially the identification of $\theta_0(g,t)$ in \citet{CallawaySantAnna2021MultipleDID} based on outcome regression using not yet treated groups as controls.  
 
Given the identification in \eqref{Eqn: DID OR}, we may embed the setup into our framework by identifying the data vector $X$ with $(Y_{g-1},Y_t, G, C_{gt},Z)$, the first step $\gamma_0$ with $(\eta_{g,0},\mu_{gt,0})$ for $\eta_{g,0}\equiv P(G=g)$, and the moment function $g$ with
\begin{align}\label{Eqn: DID OR aux}
g(X,\theta,\gamma_0)=\frac{1\{G=g\}}{\eta_{g,0}}\{Y_t-Y_{g-1}-\mu_{gt,0}(1,Z)-\theta\}~.
\end{align}
Note that \eqref{Eqn: DID OR aux} is already orthogonal to $\eta_{g,0}$. By \citet{ChernozhukovNeweySingh2022Automatic} (see also 
\citet{CallawaySantAnna2021MultipleDID} for the original result albeit in a slightly different form), we may then orthogonalize \eqref{Eqn: DID OR aux} with respect to $\gamma_0$ by setting
\begin{multline}\label{Eqn: DID DML1}
\psi(X;\theta,\gamma,\alpha)=\frac{1\{G=g\}}{\eta}\{Y_t-Y_{g-1}-\mu(1,Z)-\theta\}\\
+\alpha(C_{gt},Z)\{Y_t-Y_{g-1}-\mu(C_{gt},Z)\}~,
\end{multline}
for any $\gamma\equiv (\eta,\mu)$ with $\eta\in(0,1)$ and $\mu\colon \{0,1\}\times\mathcal Z\to\mathbf R$ and any $\alpha\colon\{0,1\}\times\mathcal Z\to\mathbf R$, where the truth of $\alpha$ (i.e., the Riesz representer) is
\begin{align}\label{Eqn: DID DML2}
\alpha_{gt,0}(C_{gt},Z)\equiv -\frac{P(G=g|Z)C_{gt}}{P(G=g)P(C_{gt}=1|Z)}~.
\end{align}

To incorporate shape constraints, we note that although, in principle, they can be imposed directly on $\mu_{gt,0}$ (the conditional expectation of the trend $Y_t-Y_{g-1}$), it may be easier to motivate them at the levels. For example, if $E[Y_t|C_{gt},Z]$ and $E[Y_{g-1}|C_{gt},Z]$ are monotonic with respect to (a subset of) $Z$, then we may let the parameter space $\Gamma$ of $\mu_{gt,0}$ be $\Gamma=\Gamma^{\uparrow}-\Gamma^{\uparrow}$ where $\Gamma^{\uparrow}\subset L^2(\{0,1\}\times\mathcal Z)$ is a set of monotonic functions. In this case, our theory implies that
\begin{align}\label{Eqn: DID DML3}
\alpha_{gt,0}=\arg\max_{\alpha\in\Gamma^{\uparrow}-\Gamma^{\uparrow}}-\frac{1}{2}E[\alpha(C_{gt},Z)^2]-E[\frac{1\{G=g\}}{\eta_{g,0}}\alpha(1,Z)]~.
\end{align}
Alternatively, we note that we may replace $\mu(C_{gt},Z)$ with $\mu(1,Z)$ in \eqref{Eqn: DID DML1} so that we may simplify the estimation of the first step by training $\mu_{gt,0}(1,Z)$ using observations from the comparison groups. Accordingly, we set $\alpha_{gt,0}(C_{gt},Z)=C_{gt}\alpha^\flat_{gt,0}(Z)$ where $\alpha^\flat_{gt,0}(Z)\equiv -P(G=g|Z)/[P(G=g)P(C_{gt}=1|Z)]$ solves
\begin{align}\label{Eqn: DID DML4}
\max_{\alpha^\flat\in\Gamma_\flat^{\uparrow}-\Gamma_\flat^{\uparrow}}-\frac{1}{2}E[C_{gt}\alpha^\flat(Z)^2]-E[\frac{1\{G=g\}}{\eta_{g,0}}\alpha^\flat(Z)]~,
\end{align}
where $\Gamma_\flat^\uparrow\subset L^2(\mathcal Z)$ is a set of monotonic functions $\alpha^\flat\colon\mathcal Z\to\mathbf R$. This is the specification we use to impose shape constraints in the empirical application below.

Given a sample $\{X_i\}_{i=1}^n$ of $X$, equation \eqref{Eqn: DID DML1} implies that we may estimate $\theta_0(g,t)$ based on observations at time $t$ consisting of units from group $g$ and $\mathcal C_{gt}$:
\begin{multline}
\hat\theta_n(g,t)=\frac{1}{n}\sum_{i=1}^{n}\frac{1\{G_i=g\}}{\hat\eta_{g,n}}\{Y_{it}-Y_{i,g-1}-\hat\mu_{gt,n}(1,Z_i)\}\\
+\frac{1}{n}\sum_{i=1}^{n}\hat\alpha_{gt,n}(C_{gt,i},Z_i)\{Y_{it}-Y_{i,g-1}-\hat\mu_{gt,n}(1,Z_i)\}~,
\end{multline}
where $\hat\eta_{g,n}\equiv\sum_{i=1}^{n}1\{G_i=g\}/n$ is the sample analog of $\eta_{g,0}$, $\hat\mu_{gt,n}(1,\cdot)$ is an estimator of $\mu_{gt,0}(1,\cdot)$, and $\hat\alpha_{gt,n}$ is obtained by running the Riesz regression.

\subsubsection{Medicaid and Mortality}\label{Sec: Medicaid}

Medicaid was enacted in 1965 as a public health insurance program providing coverage to low-income individuals. The Affordable Care Act (ACA), passed by Congress in 2010, expanded Medicaid eligibility to all adults with incomes up to $138\%$ of the federal poverty level. The ACA initially mandated expansion in all states, but a subsequent Supreme Court ruling rendered it optional, resulting in a staggered treatment design---see Table \ref{Tab: Medicaid_expansion} for the expansion timing. A central question is whether the Medicaid expansion causally reduces mortality. To address this, we build on the county level panel data constructed by \citet{BakerCallawayCunninghamGoodmanSantAnna2026Guide} in which the outcome is the crude mortality rate (per 100,000) among adults aged 20--64 and the covariates consist of the percentages of the population that are respectively female, white, and Hispanic, as well as the unemployment rate, the poverty rate, and median income. The panel includes 2,604 counties and spans the period 2009--2019. The group index $G_i\in\{2014,2015,2016,2019,\infty\}$---recall that $G_i=\infty$ means that county $i$ did not expand the coverage by 2019. The pre-treatment periods allow us to conduct pre-trend analysis.

\begin{table}[!htbp]
\centering
\sisetup{round-mode=places, round-precision=2}  
\begin{talltblr}[
  caption = {Medicaid Expansion under the Affordable Care Act},
  label = {Tab: Medicaid_expansion},
  remark{Notes} = {Source: \citet{BakerCallawayCunninghamGoodmanSantAnna2026Guide}. D.C.\ and pre-2014 adoption states are excluded from the analysis. The last column reports the shares of adults aged 20--64 in 2013.},
]{
  colsep=2pt,
  rowsep=2pt,
  column{1}={leftsep=2pt},
  column{Z}={rightsep=2pt},
  column{3-5}={cmd=\num},
  colspec={
    Q[c,m,wd=0.21\linewidth]
    Q[c,m,wd=0.3\linewidth]
    Q[c,m,wd=0.13\linewidth]
    Q[c,m,wd=0.13\linewidth]
    Q[c,m,wd=0.13\linewidth]
  },
  hline{1,Z} = {-}{},
  hline{2} = {3-5}{-}{},
  hline{3} = {-}{},
  row{1-2}={cmd=},
}
\SetCell[r=2]{c} {Expansion Year} & \SetCell[r=2]{c} {States} & \SetCell[c=3]{c} {Shares} & & \\
& & {By States} & {By Counties} & {By Adults} \\
Pre-2014 & DE, MA, NY, VT & 0.086 & 0.03 & 0.09 \\
2014 & AR, AZ, CA, CO, CT, HI, IA, IL, KY, MD, RI, WA, WV & 0.44 & 0.36 & 0.45 \\
2015 & AK, IN, PA & 0.06 & 0.06 & 0.06 \\
2016 & LA, MT & 0.04 & 0.04 & 0.02 \\
2019 & ME, VA & 0.04 & 0.05 & 0.03 \\
2020 & ID, NE, UT & 0.06 & 0.04 & 0.02 \\
2021 & MO, OK & 0.04 & 0.06 & 0.03 \\
2023 & NC, SD & 0.04 & 0.05 & 0.03 \\
Non-Expansion & AL, FL, GA, KS, MS, SC, TN, TX, WI, WY & 0.20 & 0.31 & 0.26 \\
\end{talltblr}
\end{table}

\iffalse
In exploring shape constraints in the present setting, we note that there is well documented empirical evidence supporting monotonic relationships between mortality and two key covariates: the poverty rate and income. For example, based on county level data for 1990, 2000, and 2010, \citet{CurrieSchwandt2016Mortality} document that mortality rates increase with the poverty rate and decrease with median income. Moreover, both relationships are robust across gender and age groups. These findings align with \citet{GordonSommers2016Recessions}, who report that higher poverty rates and lower median incomes are associated with higher mortality across multiple causes of death and adult subpopulations. Using administrative data at the individual level from 2001--2014, \citet{ChettyEtAl2016IncomeLife} show that life expectancy increases continuously with income, with no evidence of a threshold beyond which higher income ceases to be associated with longer life. More recently, \citet{BradyKohlerZheng2023Novel} find that mortality risk increases with poverty exposure, from a 42\% higher hazard under current poverty to a 71\% higher hazard for those in poverty over the past 10 years. In what follows, we incorporate this prior information by imposing that the conditional mean of the outcome is increasing in the poverty rate and decreasing in median income.
\fi
To motivate the shape constraints in the present setting, we draw on well-documented evidence that mortality increases with poverty and decreases with income. For example, using county-level data for 1990, 2000, and 2010, \citet{CurrieSchwandt2016Mortality} document both relationships across gender and age groups. Using individual-level administrative data from 2001--2014, \citet{ChettyEtAl2016IncomeLife} show that life expectancy increases continuously with income, with no apparent threshold. \citet{BradyKohlerZheng2023Novel} further find that mortality risk rises with poverty exposure, from a 42\% higher hazard under current poverty to a 71\% higher hazard after poverty exposure over the preceding decade. Guided by this evidence, we impose that the conditional mean of the outcome is increasing in the poverty rate and decreasing in median income.
 
% \citet{GordonSommers2016Recessions}: While previous research has found that higher unemployment may produce small beneficial effects on survival, our analysis using more recent and granular data suggests this relationship no longer holds.

% -------------------------
% Figure settings
% -------------------------

\pgfplotsset{
  gtatt panel/.style={
    width=0.5\textwidth,
    height=0.3\textwidth,
    xmin=2008.5, xmax=2019.5,
    xtick={2009,2011,2013,2015,2017,2019},
    extra y ticks={0},
    extra y tick labels={},
    extra y tick style={grid=major, grid style={solid, gray!40}},
    major tick length=1pt,
    tick label style={/pgf/number format/fixed,/pgf/number format/1000 sep={},font=\tiny},
    yticklabel style={rotate=90},
    title style={yshift=-0.65em},
    legend style={draw=none,legend columns=-1,column sep=6pt},
    mark size=0.68pt,
    line width=0.5pt,
    error bars/error mark options={solid,rotate=90,mark size=0.8pt},
    extra x tick labels={},
    extra x tick style={grid=major, grid style={densely dashed}},
  },
  right axis/.style={ytick pos=right, yticklabel pos=right},
  style cs/.style={color=Dark2-A,mark=10-pointed star},
  style dml/.style={color=Dark2-B,mark=halfsquare left*,error bars/error bar style={densely dotted}},
  style sdml/.style={color=Dark2-C,mark=halfcircle*,every mark/.append style={rotate=270},
                     error bars/error bar style={densely dashdotted}},
  event panel/.style={
    gtatt panel,
    xmin=-10.5, xmax=5.5,
    xtick={-10,-9,-8,-7,-6,-5,-4,-3,-2,-1,0,1,2,3,4,5},
    extra x ticks={-1},
  },
}

\newcommand{\off}{0.30}

% -------------------------
% Helpers
% -------------------------
\newcommand{\PlotGTATT}[4]{%
  % #1=group, #2=x-offset, #3=method style, #4=column name
  \addplot+[#3, forget plot,restrict expr to domain={\thisrow{group}}{#1:#1},error bars/.cd, y dir=both, y explicit] 
     table [x expr=\thisrow{time}+(#2),y=#4,
            y error minus expr={\thisrow{#4}-\thisrow{ci_low_#4}},
            y error plus expr={\thisrow{ci_high_#4}-\thisrow{#4}}] {\GTATT};
}

\newcommand{\PlotEventGTATT}[3]{%
  % #1=group, #2=method style, #3=column name
  \addplot+[#2, forget plot,restrict expr to domain={\thisrow{group}}{#1:#1}]
     table [x=event, y=#3] {\GTATT};
}

\newcommand{\PlotEventATT}[3]{%
  % #1=x-offset, #2=method style, #3=column name
  \addplot+[#2, forget plot,error bars/.cd, y dir=both, y explicit]
     table [x expr=\thisrow{event}+(#1),y=#3,
            y error minus expr={\thisrow{#3}-\thisrow{ci_low_#3}},
            y error plus expr={\thisrow{ci_high_#3}-\thisrow{#3}}] {\EventATT};
}

% -------------------------
% Data
% -------------------------
\pgfplotstableread{
group	time	event	cs	ci_low_cs	ci_high_cs	dml	ci_low_dml	ci_high_dml	sdml	ci_low_sdml	ci_high_sdml
2014	2009	-5	3.948646	-0.928772	8.826064	3.523899	-0.090910	7.138708	2.303480	-1.297603	5.904563
2014	2010	-4	1.747725	-4.280025	7.775476	-1.608158	-4.278177	1.061861	-2.208195	-5.237202	0.820812
2014	2011	-3	2.954795	-1.574432	7.484022	1.176764	-1.304205	3.657732	-0.166376	-2.787810	2.455058
2014	2012	-2	2.848708	-3.738555	9.435972	1.972770	-0.130555	4.076095	1.090128	-1.534811	3.715067
2014	2013	-1	0.000000	-0.000000	0.000000	0.000000	0.000000	0.000000	0.000000	0.000000	0.000000
2014	2014	0	-2.131197	-8.505874	4.243479	-1.582444	-3.627916	0.463027	-1.590870	-4.028102	0.846363
2014	2015	1	-3.647074	-11.953958	4.659811	-0.384803	-2.604925	1.835318	-0.176149	-3.117000	2.764703
2014	2016	2	-4.963733	-16.472488	6.545023	1.420331	-2.216208	5.056870	0.932510	-2.672616	4.537637
2014	2017	3	-6.064612	-17.166432	5.037207	1.722951	-1.629148	5.075050	1.534260	-2.279705	5.348226
2014	2018	4	-6.370834	-19.331034	6.589367	2.127473	-1.022951	5.277897	2.002689	-1.972682	5.978059
2014	2019	5	2.477249	-13.582551	18.537050	5.376410	1.826185	8.926635	7.665728	3.451323	11.880133
2015	2009	-6	4.920485	-3.160473	13.001443	3.885150	-4.864860	12.635160	2.348260	-6.376455	11.072974
2015	2010	-5	-0.183365	-7.437234	7.070505	-1.059362	-8.754017	6.635293	-0.799739	-9.054191	7.454714
2015	2011	-4	8.479854	-1.637637	18.597344	7.545236	-3.986218	19.076689	6.619476	-5.340010	18.578962
2015	2012	-3	3.459603	-3.716142	10.635348	2.316672	-5.106735	9.740079	0.644504	-8.021409	9.310416
2015	2013	-2	3.012421	-3.650499	9.675342	1.442521	-3.756346	6.641388	0.523516	-4.616835	5.663868
2015	2014	-1	0.000000	-0.000000	0.000000	0.000000	0.000000	0.000000	0.000000	0.000000	0.000000
2015	2015	0	2.703258	-4.332511	9.739028	3.261534	-3.338786	9.861855	5.271275	-1.560370	12.102921
2015	2016	1	12.148210	5.284887	19.011533	12.374765	5.041894	19.707637	14.384918	7.330039	21.439796
2015	2017	2	17.677326	9.950412	25.404241	18.348153	9.385488	27.310818	18.711310	9.490248	27.932371
2015	2018	3	4.595611	-2.448458	11.639680	4.871064	-1.883436	11.625564	7.865058	0.558312	15.171804
2015	2019	4	7.016006	-0.709274	14.741287	6.148231	-0.733801	13.030263	10.657663	3.173123	18.142202
2016	2009	-7	14.048370	-9.286057	37.382796	1.214199	-14.551690	16.980088	1.152305	-13.893301	16.197911
2016	2010	-6	40.983577	-48.429000	130.396153	-7.623372	-20.089579	4.842835	-7.141116	-18.891828	4.609596
2016	2011	-5	9.289281	-37.210715	55.789277	-14.680419	-31.665132	2.304293	-14.095416	-29.237747	1.046916
2016	2012	-4	2.551677	-12.092603	17.195958	-5.603523	-20.466134	9.259087	-4.673502	-17.510532	8.163528
2016	2013	-3	5.143678	-8.033429	18.320785	0.817462	-11.277045	12.911970	1.415289	-10.454732	13.285309
2016	2014	-2	3.683277	-10.974922	18.341476	0.866514	-14.093323	15.826352	0.943132	-14.278174	16.164439
2016	2015	-1	0.000000	-0.000000	0.000000	0.000000	0.000000	0.000000	0.000000	0.000000	0.000000
2016	2016	0	-8.697067	-23.608093	6.213958	-9.067751	-23.949157	5.813655	-9.088724	-24.113680	5.936232
2016	2017	1	-9.916357	-22.257083	2.424368	-8.862719	-20.500548	2.775111	-7.126508	-18.576569	4.323552
2016	2018	2	-18.922757	-31.700941	-6.144574	-17.601591	-29.699477	-5.503704	-14.166636	-26.137776	-2.195496
2016	2019	3	-13.986775	-27.535852	-0.437697	-10.306329	-22.918923	2.306265	-9.414144	-23.282549	4.454261
2019	2009	-10	-15.350323	-33.067353	2.366706	-4.483858	-17.087003	8.119286	-11.639466	-26.602817	3.323885
2019	2010	-9	-25.791807	-42.620676	-8.962938	-12.418756	-25.535209	0.697697	-15.168911	-29.609977	-0.727846
2019	2011	-8	-17.258425	-34.714323	0.197472	-8.090339	-21.380892	5.200214	-13.943306	-26.425235	-1.461376
2019	2012	-7	-12.467257	-28.174203	3.239689	-5.231264	-18.218487	7.755960	-15.405140	-29.171861	-1.638420
2019	2013	-6	-15.205226	-28.510416	-1.900036	-10.250854	-21.778406	1.276699	-16.603800	-27.936032	-5.271568
2019	2014	-5	-13.879145	-28.047122	0.288833	-9.493391	-20.655045	1.668263	-12.877358	-23.461946	-2.292770
2019	2015	-4	-5.506914	-16.912073	5.898245	-3.143221	-12.861272	6.574829	-10.407823	-20.032382	-0.783263
2019	2016	-3	-1.100543	-10.292344	8.091258	-0.373400	-8.412398	7.665598	-5.264003	-13.615630	3.087624
2019	2017	-2	3.201367	-5.774498	12.177232	4.199000	-3.632495	12.030495	0.769034	-6.930262	8.468329
2019	2018	-1	0.000000	-0.000000	0.000000	0.000000	0.000000	0.000000	0.000000	0.000000	0.000000
2019	2019	0	3.308994	-5.763475	12.381463	1.806603	-6.992260	10.605467	2.671930	-6.729099	12.072958
}\GTATT

We implement three estimation methods: the Callaway-Sant'Anna (parametric) method employed in \citet{BakerCallawayCunninghamGoodmanSantAnna2026Guide} as the benchmark, debiased machine learning (via Riesz regression) without shape constraints, and our shape constrained debiased machine learning, respectively labeled as CS, DML, and SDML. The estimands in all three cases are weighted as in \citet{BakerCallawayCunninghamGoodmanSantAnna2026Guide}. The implementation details of DML and SDML are aligned with those in Section \ref{Sec: Sims}. For simplicity, we only report pointwise (in time) confidence intervals at the 5\% significance level. 

\begin{figure}[H]
\centering\scriptsize
\caption{GT-ATTs over Calendar Time for Each Expansion Group}
\label{Fig: GTATT calendar}
\begin{tikzpicture}
\begin{groupplot}[gtatt panel,
                  group style={group name=GTATT, group size=2 by 2, horizontal sep=1.8em, vertical sep=2.4em}
                 ]
  % Group 2014
  \nextgroupplot[title={2014}, legend to name=LegendGTATT, ymin=-25, ymax=25, ytick={-20,-10,0,10,20}, extra x ticks={2014}]
  \PlotGTATT{2014}{0}{style cs}{cs}
  \PlotGTATT{2014}{-\off}{style dml}{dml}
  \PlotGTATT{2014}{\off}{style sdml}{sdml}
  \addlegendimage{style cs}\addlegendentry{CS}
  \addlegendimage{style dml}\addlegendentry{DML}
  \addlegendimage{style sdml}\addlegendentry{SDML}
  % Group 2015
  \nextgroupplot[title={2015}, right axis, ymin=-15, ymax=30, ytick={-15,0,15,30}, extra x ticks={2015}]
  \PlotGTATT{2015}{0}{style cs}{cs}
  \PlotGTATT{2015}{-\off}{style dml}{dml}
  \PlotGTATT{2015}{\off}{style sdml}{sdml}
  % Group 2016
  \nextgroupplot[title={2016}, ymin=-80, ymax=140, ytick={-80,-40,0,40,80,120}, extra x ticks={2016}]
  \PlotGTATT{2016}{0}{style cs}{cs}
  \PlotGTATT{2016}{-\off}{style dml}{dml}
  \PlotGTATT{2016}{\off}{style sdml}{sdml}
  % Group 2019
  \nextgroupplot[title={2019}, right axis, ymin=-45, ymax=25, ytick={-40,-20,0,20}, extra x ticks={2019}]
  \PlotGTATT{2019}{0}{style cs}{cs}
  \PlotGTATT{2019}{-\off}{style dml}{dml}
  \PlotGTATT{2019}{\off}{style sdml}{sdml}
\end{groupplot}

\node[anchor=north]
  at ($(GTATT c1r2.south)!0.5!(GTATT c2r2.south)+(0,-0.5cm)$) {\pgfplotslegendfromname{LegendGTATT}};
\end{tikzpicture}
%
%\vspace{0.1em}
%{\raggedright\footnotesize \textit{Notes.}\par}
\end{figure}

\pgfplotstableread{
event	cs	ci_low_cs	ci_high_cs	dml	ci_low_dml	ci_high_dml	sdml	ci_low_sdml	ci_high_sdml
-10	-15.350323	-33.067353	2.366706	-4.483858	-17.084838	8.117121	-11.639466	-26.600880	3.321948
-9	-25.791807	-42.620676	-8.962938	-12.418756	-25.533155	0.695643	-15.168911	-29.607958	-0.729864
-8	-17.258425	-34.714323	0.197472	-8.090339	-21.378802	5.198125	-13.943306	-26.423059	-1.463552
-7	-2.808912	-15.750252	10.132429	-3.750111	-13.716750	6.216527	-8.630821	-19.539082	2.277441
-6	5.068632	-10.275735	20.412999	-2.187134	-9.047862	4.673593	-4.448304	-11.360484	2.463876
-5	2.668480	-1.679572	7.016532	1.765245	-1.474290	5.004780	0.628801	-2.693216	3.950818
-4	2.133383	-2.924812	7.191578	-0.999102	-3.583249	1.585044	-1.703218	-4.719892	1.313455
-3	2.857439	-0.934498	6.649376	1.132792	-1.003286	3.268871	-0.060983	-2.338991	2.217026
-2	2.912967	-2.394760	8.220695	1.924381	0.019453	3.829310	1.071209	-1.197136	3.339554
-1	0.000000	-0.000000	0.000000	0.000000	0.000000	0.000000	0.000000	0.000000	0.000000
0	-1.492955	-6.716460	3.730549	-1.077010	-3.004868	0.850848	-0.524763	-2.715184	1.665658
1	-1.969392	-9.001513	5.062729	0.857938	-1.353896	3.069771	1.359657	-1.503312	4.222627
2	-2.725066	-12.488982	7.038849	2.859542	-0.656536	6.375619	2.533242	-0.942054	6.008538
3	-5.055663	-14.535884	4.424558	1.723326	-1.221129	4.667781	1.730479	-1.605308	5.066267
4	-4.716163	-16.157557	6.725231	2.725095	-0.174855	5.625045	3.072480	-0.559026	6.703987
5	2.477249	-13.582551	18.537050	5.376410	1.826298	8.926522	7.665728	3.451455	11.880001
}\EventATT

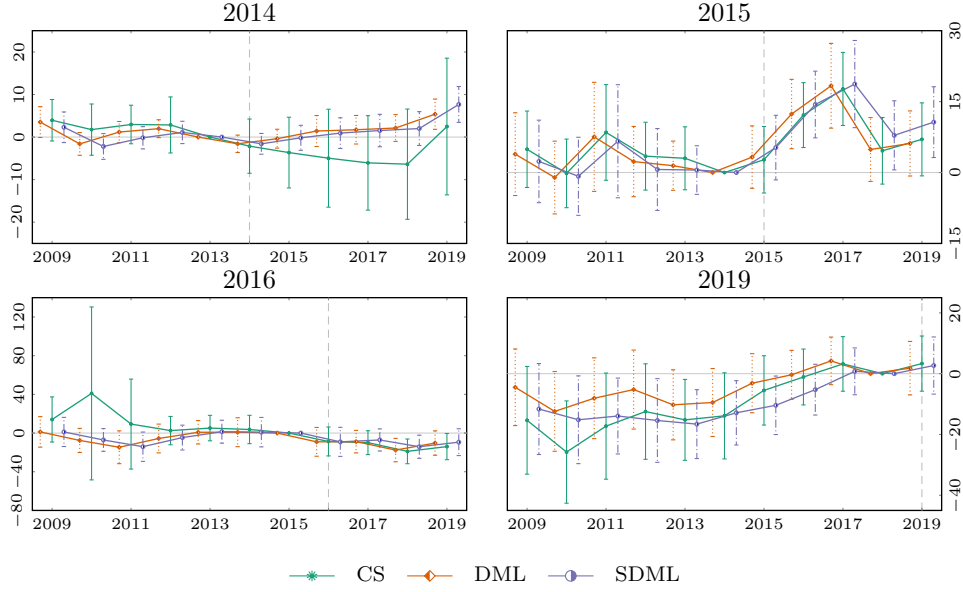
\begin{figure}[H]
\centering\scriptsize
\caption{GT-ATTs over Event Time for Each Expansion Group}
\label{Fig: GTATT Event}
\begin{tikzpicture}
\begin{groupplot}[event panel,
                  group style={group name=Event, group size=2 by 2, horizontal sep=1.8em, vertical sep=2.4em},
                 ]
  % Group 2014                 
  \nextgroupplot[title={2014}, legend to name=LegendEvent, xmin=-5.5, xmax=5.5]
  \PlotEventGTATT{2014}{style cs}{cs}
  \PlotEventGTATT{2014}{style dml}{dml}
  \PlotEventGTATT{2014}{style sdml}{sdml}
  \addlegendimage{style cs}\addlegendentry{CS}
  \addlegendimage{style dml}\addlegendentry{DML}
  \addlegendimage{style sdml}\addlegendentry{SDML}
  % Group 2015
  \nextgroupplot[title={2015}, right axis, xmin=-6.5, xmax=4.5]
  \PlotEventGTATT{2015}{style cs}{cs}
  \PlotEventGTATT{2015}{style dml}{dml}
  \PlotEventGTATT{2015}{style sdml}{sdml}
  % Group 2016
  \nextgroupplot[title={2016}, xmin=-7.5, xmax=3.5]
  \PlotEventGTATT{2016}{style cs}{cs}
  \PlotEventGTATT{2016}{style dml}{dml}
  \PlotEventGTATT{2016}{style sdml}{sdml}
  % Group 2019
  \nextgroupplot[title={2019}, right axis, xmin=-10.5, xmax=0.5]
  \PlotEventGTATT{2019}{style cs}{cs}
  \PlotEventGTATT{2019}{style dml}{dml}
  \PlotEventGTATT{2019}{style sdml}{sdml}
\end{groupplot}

\node[anchor=north]
  at ($(Event c1r2.south)!0.5!(Event c2r2.south)+(0,-0.5cm)$)
  {\pgfplotslegendfromname{LegendEvent}};
\end{tikzpicture}
%
%\vspace{0.75em}
%{\raggedright\footnotesize \textit{Notes.}\par}
\end{figure}

Figures \ref{Fig: GTATT calendar} and \ref{Fig: GTATT Event} present GT-ATTs over calendar time and event time respectively for each expansion group. For the 2014 cohort, the CS estimates are negative in most post-expansion years, whereas the DML and SDML estimates become positive after 2015 and are significantly positive by 2019. For the 2015 cohort, all three methods produce similar positive effects during the first two post-expansion years, but the SDML estimates are systematically larger and remain significant in 2018 and 2019, when the CS and DML confidence intervals include zero. For the 2016 cohort, all three methods produce negative estimates that are significantly negative in 2018; by 2019, only the CS estimate remains significant. The contemporaneous estimates for the 2019 cohort are similar and insignificant across methods, but their pre-treatment estimates differ substantially, with SDML generally producing more negative estimates than DML. These differences carry over to Figure \ref{Fig: Aggregated ATT Event Study}: the CS estimates are negative from event time zero through four, whereas the DML and SDML estimates are positive from event time one onward, although all corresponding confidence intervals include zero. The aggregate pre-treatment estimates are not uniformly close to zero: CS and SDML are significantly negative at event time $-9$, SDML remains significantly negative at event time $-8$, and DML is narrowly positive and significant at event time $-2$. At event time five, the DML and SDML estimates are positive and significant, while the CS estimate is smaller and considerably less precise. Overall, relative to DML, imposing shape constraints produces meaningful differences in the magnitude and statistical significance of the estimated effects across several cohorts and time periods.

% We stress that the discrepancies between the different methods should not be interpreted as evidence of one being superior than another, but rather than potential of shape constraints in aiding estimation and inference. 

% Overall, DML and SDML generally yield shorter confidence intervals than CS, but their interval lengths are comparable.

\begin{figure}[H]
\centering\scriptsize
\caption{Aggregated Event Study ATTs}
\label{Fig: Aggregated ATT Event Study}
\begin{tikzpicture}
\begin{axis}[
  event panel,
  width=0.65\textwidth,
  height=0.34\textwidth,
  ymin=-45, ymax=25, ytick={-40,-20,0,20},
  legend to name=LegendAggEvent,
]
  \PlotEventATT{0}{style cs}{cs}
  \PlotEventATT{-\off}{style dml}{dml}
  \PlotEventATT{\off}{style sdml}{sdml}
  \addlegendimage{style cs}\addlegendentry{CS}
  \addlegendimage{style dml}\addlegendentry{DML}
  \addlegendimage{style sdml}\addlegendentry{SDML}
\end{axis}

\node[anchor=north] 
  at ($(current axis.south)+(0,-0.5cm)$)
  {\pgfplotslegendfromname{LegendAggEvent}};
\end{tikzpicture}
%
%\vspace{0.75em}
%{\raggedright\footnotesize \textit{Notes.}\par}
\end{figure}
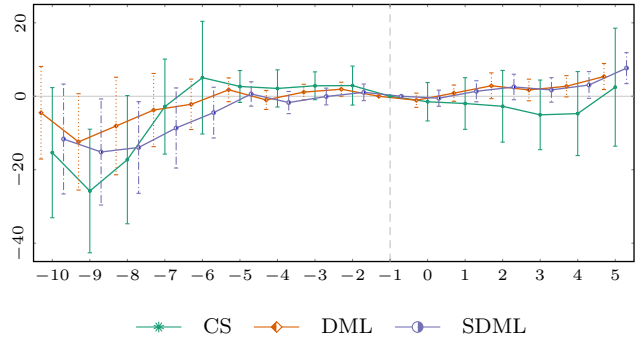

\subsection{Working Hours and Wage Growth}

How working hours translate into wage growth is central to understanding career progression, labor supply incentives, and the accumulation of earnings differences across workers. The theoretical case for allowing this relationship to be nonlinear goes back at least to \citet{Barzel1973Wage}, who treats hours and wages as jointly determined components of an employment arrangement rather than separate objects linked by a fixed hourly wage. Consistent with this perspective, subsequent empirical work finds that hourly earnings vary systematically with hours worked, pointing to nonlinear wage-hours relationships \citep{BiddleZarkin1989Wage,BickBlandinRogerson2022Wage}. Building on a dynamic complete information framework, \citet{Gicheva2013Working} develops a model that predicts a convex intertemporal relationship between working hours and wage growth, which is then informally validated based on a data set from the panel survey of registrants in the United States for the Graduate Management Admission Test (GMAT) between June 1990 and March 1991---see \citet{Fang2021Unifying} for a formal test.

To map this problem into our framework, consider the model:
\begin{align}
Y=\gamma_0(D,Z) + U~,
\end{align} 
where $Y$ is the annual wage growth rate, $D$ is weekly working hours, $Z$ is a vector of demographic variables (e.g., experience, education, gender, age, race, and family characteristics), and $E[U|D,Z]=0$. In the GMAT sample, there are $19$ control variables including experience, education, gender, age, race, and family characteristics. The parameter $\theta_0$ is the weighted average derivative: for $w$ a weight function,
\begin{align}
  \theta_0\equiv E[w(D)\nabla_d\gamma_0(D,Z)]~.
\end{align}

In the implementation below, we employ the same GMAT dataset as \citet{Gicheva2013Working} and set $w$ to be the normalized indicator of the empirical $[0.1,0.9]$ quantile region of $D$. The first step $\gamma_0$ is estimated by a two-branch neural network subject to convexity with respect to $D$. The Riesz representer $\alpha_0$ is characterized as the solution to
\begin{align}
\max_{\alpha\in\Gamma^\smallsmile-\Gamma^\smallsmile} -\frac{1}{2} E[\alpha(D,Z)^2] + E[w(D)\nabla_d\alpha(D,Z)]~,
\end{align}
where $\Gamma^\smallsmile\subset L^2(D,Z)$ is the class of functions $(d,z)\mapsto\gamma(d,z)$ convex in $d$. We keep the tuning of architectures consistent with the simulations for convexity in Section~\ref{Sec: Sims, ATE}. The final cleaned sample consists of $1,911$ survey respondents. We also consider men and women separately as in \citet{Gicheva2013Working}. % This gender-specific analysis is motivated by evidence that gender earnings gaps among high-skill professionals are closely linked to differences in weekly hours, career interruptions, and the disproportionate rewards to long hours \citep{BertrandGoldinKatz2010Dynamics,Goldin2014Grand}.

\begin{table}[!htbp]
\centering
\sisetup{round-mode=places, round-precision=4}
\makeatletter
\newcommand{\CI}[1]{\CI@aux#1\@nil}
\def\CI@aux#1,#2\@nil{[\num{#1},\,\num{#2}]}
\makeatother
\begin{talltblr}[
  caption = {DML and SDML Estimates of the Weighted Average Derivative.},
  label = {Tab: Wage results},
  remark{Notes} = {The last column reports the number of observations with positive weight. The original sample
  sizes are $1,911$ for all respondents, $1,103$ for men, and $808$ for women.},
]{
  colsep=5pt,
  rowsep=2pt,
  column{1}={leftsep=8pt},
  column{Z}={rightsep=8pt},
  colspec = {Q[c] Q[c] S[table-format=-1.4, table-number-alignment=center] Q[c] Q[c] Q[c]},
  cell{2-Z}{4} = {cmd=\num},
  cell{2-Z}{5} = {cmd=\CI},
  hline{1,2,Z} = {-}{},
}
Method & Sample &  {Point Estimate}  & Standard Error       & 95\% Confidence Interval & Sample Size  \\
DML    & All    & 0.003141     & 0.001165 & 0.000858,0.005425        & 1546 \\
       & Men    & 0.004586     & 0.001468 & 0.001709,0.007463        & 904 \\
       & Women  & 0.000588     & 0.002127 & -0.003581,0.004757       & 677 \\
\addlinespace       
SDML   & All    & 0.000790     & 0.001744 & -0.002627,0.004208       & 1546 \\
       & Men    & 0.002810     & 0.002604 & -0.002295,0.007914       & 904 \\ 
       & Women  & -0.000676    & 0.002371 & -0.005323,0.003971       & 677 \\ 
 
\end{talltblr}
\end{table}

Table~\ref{Tab: Wage results} reports estimates of the weighted average derivative. The unconstrained DML estimates are positive and significant at the $5\%$ level for the full sample and for men, while the estimate for women is smaller and statistically indistinguishable from zero. The SDML point estimates are also positive for the full sample and for men but are estimated less precisely, and the SDML estimate for women is slightly negative; all three SDML confidence intervals include zero. Overall, the table suggests a modest average marginal association, if any, between weekly working hours and subsequent wage growth. This is consistent with the broader wage-hours literature, which emphasizes that the relationship between hours and earnings is nonlinear and heterogeneous rather than uniformly positive at the margin. 

% For example, \citet{DenningJacobLefgrenLehn2022Wage} show that marginal returns to hours can be small within labor-market settings, even when hours are important for broader wage differences.

% Alternative application: Convex supply curves in \url{https://www.aeaweb.org/articles?id=10.1257/aer.20210811}; this is an endogenous model. 

\section{Conclusion}\label{Sec: Conclusion}

In this paper, we present a general framework of identification and estimation for DML where the parameter of interest $\theta_0$ is identified by a GMM model in the presence of a first step nuisance $\gamma_0$.  The parameter $\gamma_0$ may be a high dimensional parameter defined by a model with endogeneity and possibly subject to nonlinear shape constraints. In particular, we establish identification of the Riesz representer $\alpha_0$ under mild conditions and generalize the Riesz regression to accommodate a generic first step $\gamma_0$. Since machine learning alone is incapable of overcoming the curse of dimensionality, we incorporate shape constraints as a vehicle for bringing economic structures into DML. By exploiting linearity of the first step influence function with respect to $\alpha_0$, we are able to develop a theory that is conceptually simple, transparent, and unifying. We believe our framework benefits statistical inference on many causal and structural problems, in both high dimensional and classical semiparametric settings. 
 
\begin{appendices}
\titleformat{\section}{\Large\center}{{\sc Appendix}}{1em}{}
\renewcommand{\theequation}{\thesection.\arabic{equation}}
\numberwithin{equation}{section}

\bookmarksetup{numbered=false}
\section{Proofs of Main Results}\label{Sec: Main proofs}
\bookmarksetup{numbered=true}

\noindent{\sc Proofs of Theorems \ref{Thm: Riesz ID} and \ref{Thm: Characterization}:} By the Gateaux differentiability imposed in Assumption \ref{Ass: Riesz}, equation \eqref{Eqn: Local Orthogonal wtr gamma} is equivalent to: for all $\delta\in \Gamma-\gamma_0$,
\begin{align}\label{Thm: Riesz ID, aux1}
\Psi_0 (\Pi_0(\delta))+\Upsilon_0(\alpha_0,\Pi_0(\delta))=0~,
\end{align}
or equivalently, by the linearity of $\delta\mapsto\Pi_0(\delta)$, $b\mapsto\Psi_0(b)$, and $b\mapsto\Upsilon_0(\alpha_0,b)$,
\begin{align}\label{Thm: Riesz ID, aux2}
\Psi_0 (b)+\Upsilon_0(\alpha_0,b)=0~,
\end{align}
for all $b\in\Pi_0(\mathbf H_0)$ and hence for all $b\in\mathbf A_0$ (the closure of $\Pi_0(\mathbf H_0)$), by continuity of $b\mapsto \Psi_0 (b)$ and $b\mapsto \Upsilon_0(\alpha_0,b)$.

Since $\mathbf A_0$ is a Hilbert space as a closed subspace in $\mathbf A$, by Assumptions \ref{Ass: Riesz}(ii)(iv) and the Lax-Milgram theorem (see, e.g., \citet[pp.~13--17]{ErnGuermond2021FiniteII}), we may then conclude that there exists a unique $\alpha_0\in \mathbf A_0$ satisfying \eqref{Thm: Riesz ID, aux2} and hence \eqref{Eqn: Local Orthogonal wtr gamma} if $\Upsilon_0$ is coercive. Given result \eqref{Thm: Riesz ID, aux2}, the conclusion of Theorem \ref{Thm: Characterization} follows by Theorem 22.A in \citet{Zeidler1990IIA} if $\Upsilon_0$ is positive. If $\Upsilon_0$ is negative, then $-\Upsilon_0$ is positive, and Theorem 22.A in \citet{Zeidler1990IIA} applies by multiplying both sides of \eqref{Thm: Riesz ID, aux2} by $-1$.  \qed

\noindent{\sc Proof of Theorem \ref{Thm: Convergence rate}:} For simplicity, below we adopt some notation from empirical process theory.  Write $Pf\equiv E[f(X)]$ and $\mathbb P_nf\equiv \sum_{i=1}^{n}f(X_i)/n$ so that $\mathbb G_nf\equiv \sqrt n\{\mathbb P_nf-Pf\}$ for a generic function $f$. 

% Below we adopt the notation in empirical process theory. Thus, $Pl(a,b,\nu)\equiv E[l(X,a,b,\nu)]$, $\mathbb P_nl(a,b,\nu)\equiv \sum_{i=1}^n l(X_i,a,b,\nu)/n$, and $\mathbb G_n\equiv \sqrt n\{\mathbb P_n-P\}$. 

First, we note by the linearity and symmetry in Assumptions \ref{Ass: Convergence rate, population}(i)(ii) that
\begin{align}\label{Eqn: Convergence rate aux1}
Pf(\alpha,\nu) -Pf(\alpha_0,\nu)&= \frac{1}{2}Pl(\alpha-\alpha_0,\alpha-\alpha_0,\nu_0)+\frac{1}{2}Pl(\alpha-\alpha_0,\alpha-\alpha_0,\nu-\nu_0) \notag\\
&\quad + Pl(\alpha-\alpha_0,\alpha_0,\nu-\nu_0)+Ph(\alpha-\alpha_0,\nu-\nu_0)\notag\\
&\quad +Pl(\alpha-\alpha_0,\alpha_0,\nu_0)+Ph(\alpha-\alpha_0,\nu_0)~.
\end{align}
By Assumptions \ref{Ass: Convergence rate, population}(i)(ii)(iii), Theorem \ref{Thm: Riesz ID} implies that, for all $\alpha\in\mathbf A_0$,
\begin{align}\label{Eqn: Convergence rate aux2}
Pl(\alpha-\alpha_0,\alpha_0,\nu_0)+Ph(\alpha-\alpha_0,\nu_0)=0~.
\end{align}
By results \eqref{Eqn: Convergence rate aux1} and \eqref{Eqn: Convergence rate aux2}, we note that, for all $\alpha\in\mathbf A_0$ and $\nu\in\mathbf B$,
\begin{align}\label{Eqn: Convergence rate aux3}
Pf(\alpha,\nu) -Pf(\alpha_0,\nu)&= \frac{1}{2}Pl(\alpha-\alpha_0,\alpha-\alpha_0,\nu)  \notag\\
&\quad + Pl(\alpha-\alpha_0,\alpha_0,\nu-\nu_0)+Ph(\alpha-\alpha_0,\nu-\nu_0)~.
\end{align}
By Assumptions \ref{Ass: Convergence rate, population}(i), we note that, for all $a\in\mathbf A_0$, all $\nu$ in a (sufficiently small) neighborhood $V_0$ of $\nu_0$, and some constants $\varpi,\varpi_1>0$,
\begin{multline}\label{Eqn: Convergence rate aux4}
E[l(X,a,a,\nu)]=E[l(X,a,a,\nu-\nu_0)]+E[l(X,a,a,\nu_0)]\\
\ge -\varpi\|a\|_{\mathbf A}^2\|\nu-\nu_0\|_{\mathbf B} -\|\Upsilon_0\|_{op}\|a\|_{\mathbf A}^2\ge -\frac{\varpi_1}{2}\|a\|_{\mathbf A}^2~.
\end{multline}
By Assumptions \ref{Ass: Convergence rate, population}(i)(ii), results \eqref{Eqn: Convergence rate aux3} and \eqref{Eqn: Convergence rate aux4} thus imply that, for all $\alpha\in\mathbf A_0$, all $\nu\in V_0$, and some constant $c_1>0$,
\begin{align}\label{Eqn: Convergence rate aux5}
Pf(\alpha,\nu) -Pf(\alpha_0,\nu) &\ge -\frac{\varpi_1}{2}\|\alpha-\alpha_0\|_{\mathbf A}^2 - 2c_1^{1/2} \|\alpha-\alpha_0\|_{\mathbf A}\|\nu-\nu_0\|_{\mathbf B}\notag\\
&\ge -\varpi_1\|\alpha-\alpha_0\|_{\mathbf A}^2 - \frac{2c_1}{\varpi_1}  \|\nu-\nu_0\|_{\mathbf B}^2~,
\end{align}
where the second step follows by the inequality of arithmetic and geometric means. 

Next, by the continuity in Assumptions \ref{Ass: Convergence rate, population}(i)(ii), there exist constants $c_2,c_3>0$ such that, for all $\alpha\in\mathbf A_0$ and $\nu\in\mathbf B$,
\begin{align}\label{Eqn: Convergence rate aux6}
\frac{1}{2}&Pl(\alpha-\alpha_0,\alpha-\alpha_0,\nu-\nu_0)+Pl(\alpha-\alpha_0,\alpha_0,\nu-\nu_0)+Ph(\alpha-\alpha_0,\nu-\nu_0)\notag\\
&\le c_2\|\alpha-\alpha_0\|_{\mathbf A}^2 \|\nu-\nu_0\|_{\mathbf B} + c_3\|\alpha-\alpha_0\|_{\mathbf A} \|\nu-\nu_0\|_{\mathbf B}~.
\end{align}
By the triangle inequality and Assumptions \ref{Ass: Convergence rate, approximation}(i)(ii), we have $\|\alpha-\alpha_0\|_{\mathbf A}\le \|\alpha-\alpha_{0,n}\|_{\mathbf A}+\|\alpha_{0,n}-\alpha_0\|_{\mathbf A}\lesssim \varrho_n$ for all $\alpha\in\mathbf A_{0,n}$. It follows from \eqref{Eqn: Convergence rate aux6} that, for some $c_4>0$, 
\begin{multline}\label{Eqn: Convergence rate aux7}
\frac{1}{2}Pl(\alpha-\alpha_0,\alpha-\alpha_0,\nu-\nu_0)+Pl(\alpha-\alpha_0,\alpha_0,\nu-\nu_0)+Ph(\alpha-\alpha_0,\nu-\nu_0) \\
\le 2c_4^{1/2}\varrho_n\|\alpha-\alpha_0\|_{\mathbf A} \|\nu-\nu_0\|_{\mathbf B} \le   \frac{\varkappa}{2}\|\alpha-\alpha_0\|_{\mathbf A}^2 + \frac{2c_4\varrho_n^2}{\varkappa}\|\nu-\nu_0\|_{\mathbf B}^2~,
\end{multline}
where the last step is due to the inequality of arithmetic and geometric means. Combining Assumption \ref{Ass: Convergence rate, population}(iii) with results \eqref{Eqn: Convergence rate aux1}, \eqref{Eqn: Convergence rate aux2}, and \eqref{Eqn: Convergence rate aux7} then yields 
\begin{align}\label{Eqn: Convergence rate aux7aux}
Pf(\alpha,\nu) -Pf(\alpha_0,\nu)\le -\frac{\varkappa}{2}\|\alpha-\alpha_0\|_{\mathbf A}^2+\frac{2c_4\varrho_n^2}{\varkappa}\|\nu-\nu_0\|_{\mathbf B}^2~,
\end{align}
for all $\alpha\in\mathbf A_{0,n}$ and all $\nu\in\mathbf B$. Exploiting $\|\alpha-\alpha_{0,n}\|_{\mathbf A}^2\le 2\{\|\alpha-\alpha_0\|_{\mathbf A}^2+\|\alpha_{0,n}-\alpha_0\|_{\mathbf A}^2\}$, $\alpha_{0,n}\in\mathbf A_{0,n}$, and $\varrho_n\ge 1$, we obtain from results \eqref{Eqn: Convergence rate aux5} and \eqref{Eqn: Convergence rate aux7aux} that
\begin{align}\label{Eqn: Convergence rate aux8}
Pf(\alpha,\nu) &-Pf(\alpha_{0,n},\nu) = Pf(\alpha,\nu) -Pf(\alpha_0,\nu)+Pf(\alpha_0,\nu)-Pf(\alpha_{0,n},\nu)\notag\\
&\le -\frac{\varkappa}{2}\|\alpha-\alpha_0\|_{\mathbf A}^2+\frac{2c_4\varrho_n^2}{\varkappa}\|\nu-\nu_0\|_{\mathbf B}^2 + \varpi_1\|\alpha_{0,n}-\alpha_0\|_{\mathbf A}^2 + \frac{2c_1}{\varpi_1}  \|\nu-\nu_0\|_{\mathbf B}^2\notag\\
&\le -\frac{\varkappa}{4}\|\alpha-\alpha_{0,n}\|_{\mathbf A}^2+\varpi_2\|\alpha_{0,n}-\alpha_0\|_{\mathbf A}^2+c_5\varrho_n^2\|\nu-\nu_0\|_{\mathbf B}^2~,
\end{align}
for all $\alpha\in\mathbf A_{0,n}$, all $\nu\in V_0$, some constant $\varpi_2>0$, and some constant $c_5>0$.

Given Assumptions \ref{Ass: Convergence rate, approximation}(ii)(iii), we next notice that it suffices to show
\begin{align}\label{Eqn: Convergence rate aux9}
\|\hat\alpha_n-\alpha_{0,n}\|_{\mathbf A}+\{\lambda_nJ_n(\hat\alpha_n)\}^{1/2}=O_p(\delta_n+\varrho_n\|\hat\nu_n-\nu_0\|_{\mathbf B}+\{\lambda_nJ_n(\alpha_{0,n})\}^{1/2})~.
\end{align}
This is accomplished by employing a variant of the peeling device in the proof of Theorem 3.4.6 in \citet{VaartWellner2023Book}. To this end, for each $n\in\mathbf N$ (the set of natural numbers), $j\in\mathbf Z$ (the set of integers), and each constant $M>0$, we define a ``shell'' $S_{n,j,M}$ and an event $\mathcal E_{n,M}$ as follows:
\begin{align*}
S_{n,j,M}&\equiv \{(\alpha,\nu)\in \mathbf A_{0,n}\times\mathbf B_n\colon 2^{j-1}\delta_n <\|\alpha-\alpha_{0,n}\|_{\mathbf A} + \{\lambda_nJ_n(\alpha)\}^{\frac{1}{2}}\le 2^j\delta_n,\notag\\
&\hspace{1.2cm} 2^M\max\{\varrho_n\|\nu-\nu_0\|_{\mathbf B},\{\lambda_nJ_n(\alpha_{0,n})\}^{\frac{1}{2}}\}\le \|\alpha-\alpha_{0,n}\|_{\mathbf A} + \{\lambda_nJ_n(\alpha)\}^{\frac{1}{2}}\}~,\\
\mathcal E_{n,M}&\equiv\{\|\hat\alpha_n-\alpha_{0,n}\|_{\mathbf A}+\{\lambda_nJ_n(\hat\alpha_n)\}^{\frac{1}{2}}\ge 2^M[\delta_n+\varrho_n\|\hat\nu_n-\nu_0\|_{\mathbf B}+\{\lambda_nJ_n(\alpha_{0,n})\}^{\frac{1}{2}}]\}~.
\end{align*}
Without loss of generality, assume $\varkappa/4\le 1$. Thus, if $(\alpha,\nu)\in S_{n,j,M}$ with $\nu\in V_0$, then we may exploit \eqref{Eqn: Convergence rate aux8}, $J_n$ being nonnegative, the definition of $S_{n,j,M}$, and $\|\alpha-\alpha_{0,n}\|_{\mathbf A}^2+\lambda_nJ_n(\alpha)\ge [\|\alpha-\alpha_{0,n}\|_{\mathbf A} + \{\lambda_nJ_n(\alpha)\}^{1/2}]^2/2$ to obtain that
\begin{align}\label{Eqn: Convergence rate aux10}
Pf&(\alpha,\nu) -Pf(\alpha_{0,n},\nu)-\lambda_nJ_n(\alpha)+\lambda_nJ_n(\alpha_{0,n})\notag\\
& \le -\frac{\varkappa}{4}\{\|\alpha-\alpha_{0,n}\|_{\mathbf A}^2+\lambda_nJ_n(\alpha)\}+\varpi_2\|\alpha_{0,n}-\alpha_0\|_{\mathbf A}^2+c_5\varrho_n^2\|\nu-\nu_0\|_{\mathbf B}^2+\lambda_nJ_n(\alpha_{0,n})\notag\\
&\le -(\frac{\varkappa}{32}-\frac{c_6\varpi_2}{2^{2j}}-\frac{c_5}{2^{2M}}-\frac{1}{2^{2M}})(2^j\delta_n)^2~,
\end{align}
where $c_6>0$ is a constant such that $\|\alpha_{0,n}-\alpha_0\|_{\mathbf A}^2\le c_6 \delta_n^2$ (by Assumption \ref{Ass: Convergence rate, approximation}(ii)).

Define $R_n\equiv \sup_{a\in\mathbf A_{0,n}}\{\hat{\mathfrak C}_n(a)-\lambda_nJ_n(a)\}-[\hat{\mathfrak C}_n(\hat\alpha_n)-\lambda_nJ_n(\hat\alpha_n)]$. Then by $\alpha_{0,n}\in\mathbf A_{0,n}$ and the definition of $\hat\alpha_n\in\mathbf A_{0,n}$, we note that $0\le R_n\le \delta_n^2$ and that
\begin{align}\label{Eqn: Convergence rate aux11}
\mathbb P_nf(\hat\alpha_n,\hat\nu_n)-\mathbb P_nf(\alpha_{0,n},\hat\nu_n) \ge \lambda_nJ_n(\hat\alpha_n)-\lambda_nJ_n(\alpha_{0,n})-R_n~.
\end{align}
Combining \eqref{Eqn: Convergence rate aux10} and \eqref{Eqn: Convergence rate aux11} yields that, if $(\hat\alpha_n,\hat\nu_n)\in S_{n,j,M}$ and $\hat\nu_n\in V_0$, then
\begin{align}\label{Eqn: Convergence rate aux12}
\sup_{(\alpha,\nu)\in S_{n,j,M}} &\{\mathbb G_nf(\alpha,\nu)-\mathbb G_nf(\alpha_{0,n},\nu)\}\notag\\
&\ge -\sqrt n\{Pf(\hat\alpha_n,\hat\nu_n)-Pf(\alpha_{0,n},\hat\nu_n)-\lambda_nJ_n(\hat\alpha_n)+\lambda_nJ_n(\alpha_{0,n})+R_n\}\notag\\
&\ge \sqrt n \{(\frac{\varkappa}{32}-\frac{c_6\varpi_2}{2^{2j}}-\frac{c_5}{2^{2M}}-\frac{1}{2^{2M}})(2^j\delta_n)^2-\delta_n^2\}~.
\end{align}
Note that the right hand side of the last inequality is bounded below by $c_7\sqrt n (2^j\delta_n)^2$ for some constant $c_7>0$ and for all large $j$ and $M$. Moreover, observe that
\begin{multline}\label{Eqn: Convergence rate aux13}
\sup_{(\alpha,\nu)\in S_{n,j,M}} \mathbb G_n(f(\alpha,\nu)-f(\alpha_{0,n},\nu)) \\
\le \sup_{(\alpha,\nu)\in U_{0,n}(2^j\delta_n) }|\mathbb G_n(f(\alpha,\nu)-f(\alpha_{0,n},\nu))|~.
\end{multline}
Since $(\hat\alpha_n,\hat\nu_n)\in S_{n,j,M}$ for some $j\ge M$ whenever $\mathcal E_{n,M}$ occurs, \eqref{Eqn: Convergence rate aux12}, \eqref{Eqn: Convergence rate aux13}, Markov's inequality, and Assumptions \ref{Ass: Convergence rate, approximation}(iv)(v) imply that, for some constant $c_8>0$,
\begin{multline}\label{Eqn: Convergence rate aux14}
P(\mathcal E_{n,M}) \le P(\mathcal E_{n,M},\hat\nu_n\in V_0) + P(\hat\nu_n\notin V_0)\lesssim \sum_{j=M}^{\infty} \frac{\omega_n(2^j\delta_n)}{c_7\sqrt n (2^j\delta_n)^2}+P(\hat\nu_n\notin V_0)\\
\lesssim \sum_{j=M}^{\infty}\frac{2^{j\varsigma}\omega_n(\delta_n)}{c_7\sqrt n (2^j\delta_n)^2}+ P(\hat\nu_n\notin V_0)\lesssim \frac{1}{c_7}\sum_{j=M}^{\infty}(\frac{1}{2^{2-\varsigma}})^j+ P(\hat\nu_n\notin V_0)~,
\end{multline}
where the second inequality follows because $\omega_n(\delta)/\delta^\varsigma$ is nonincreasing in $\delta$. Since $\sum_{j=M}^{\infty}(\frac{1}{2^{2-\varsigma}})^j\to 0$ whenever $M=M_n\to\infty$ due to $\varsigma<2$ by assumption, \eqref{Eqn: Convergence rate aux9} now readily follows from \eqref{Eqn: Convergence rate aux14} and Assumption \ref{Ass: Convergence rate, sample}(ii). \qed

\noindent{\sc Proof of Theorem \ref{Thm: Convergence rate2}:} For each $\Pi\in\mathbf D_n$, define a seminorm $\|\cdot\|_{\Pi}$ on $\mathrm{lin}(\Gamma-\Gamma)$ by $\|\eta\|_{\Pi}=\|\Pi\eta\|_{\mathbf A}$. Since $\alpha_{0,n}\in\mathbf A_{0,n}$ and $\hat\alpha_n\in\hat{\mathbf A}_n$, we have $\alpha_{0,n}=\Pi_0(\eta_{0,n})$ and $\hat\alpha_n=\hat\Pi_n(\hat\eta_n)$ for some $\eta_{0,n},\hat\eta_n\in\Delta\Gamma_n$. By Assumption \ref{Ass: Convergence rate, approximation}(ii) and the triangle inequality, it suffices to show \eqref{Eqn: Convergence rate2 aux0} with $\alpha_{0,n}$ in place of $\alpha_0$. Note that $\|\hat\Pi_n(\eta_{0,n})-\Pi_0(\eta_{0,n})\|_{\mathbf A}\le\varrho_n\|\hat\Pi_n-\Pi_0\|_{op,n}$ by Assumptions \ref{Ass: Convergence rate, approximation}(i) and \ref{Ass: Convergence rate, Pi}(i)(ii). Thus, it suffices to show that
\begin{multline}\label{Eqn: Convergence rate2 aux1}
\|\hat\Pi_n(\hat\eta_n)-\hat\Pi_n(\eta_{0,n})\|_{\mathbf A}\equiv\|\hat\eta_n-\eta_{0,n}\|_{\hat\Pi_n}\\
=O_p(\delta_n+\varrho_n\|\hat\nu_n-\nu_0\|_{\mathbf B}+\varrho_n\|\hat\Pi_n-\Pi_0\|_{op,n})~.
\end{multline}
To this end, we rework the proof of Theorem \ref{Thm: Convergence rate} under the random seminorm $\|\cdot\|_{\hat\Pi_n}$ by treating $\Pi_0$ as a nuisance parameter. We shall also make use of the notation, constants, and definitions there without explicit references.

Suppose first that $\mathfrak C_0(\alpha)-\mathfrak C_0(\alpha_0) \le c\|\alpha-\alpha_0\|_{\mathbf A}^\tau$ for some $c>0$ and $\tau>1$ whenever $d_{\mathbf A}(\alpha,\mathbf A_0)\le\epsilon$ for some $\epsilon>0$. Fix $\alpha\in\mathbf A$ with $d_{\mathbf A}(\alpha,\mathbf A_0)\le\epsilon$. Let $\alpha'\equiv \alpha_0+t\{\alpha-\alpha_0\}$ with $t\in[0,1]$ which satisfies $d_{\mathbf A}(\alpha',\mathbf A_0)\le \epsilon$ by $\alpha_0\in\mathbf A_0$ and convexity of the distance function $a\mapsto d_{\mathbf A}(a,\mathbf A_0)$ due to the linearity of $\mathbf A_0$. Hence, for any $t\in[0,1]$,
\begin{multline}\label{Eqn: Convergence rate2 aux2}
ct^\tau\|\alpha-\alpha_0\|_{\mathbf A}^\tau\ge \mathfrak C_0(\alpha')-\mathfrak C_0(\alpha_0)\\
= t[\Upsilon_0(\alpha-\alpha_0,\alpha_0)+\Psi_0(\alpha-\alpha_0)]+\frac{1}{2}t^2\Upsilon_0(\alpha-\alpha_0,\alpha-\alpha_0)~,
\end{multline}
where the equality is due to the (bi)linearity and symmetry in Assumptions \ref{Ass: Convergence rate, population}(i)(ii). Since $\tau>1$, dividing \eqref{Eqn: Convergence rate2 aux2} by $t>0$ and letting $t\downarrow 0$ yield:
\begin{align}\label{Eqn: Convergence rate2 aux3}
\Upsilon_0(\alpha-\alpha_0,\alpha_0)+\Psi_0(\alpha-\alpha_0)\le 0~.
\end{align}
For $\alpha''\equiv 2\alpha_0-\alpha$, we note by $d_{\mathbf A}(-\alpha,\mathbf A_0)\le\epsilon$ (since $\mathbf A_0$ is linear), subadditivity of  $a\mapsto d_{\mathbf A}(a,\mathbf A_0)$, and $2\alpha_0\in\mathbf A_0$ that $d_{\mathbf A}(\alpha'',\mathbf A_0)\le \epsilon$. Since \eqref{Eqn: Convergence rate2 aux3} holds for any $\alpha\in\mathbf A$ with $d_{\mathbf A}(\alpha,\mathbf A_0)\le\epsilon$, it follows by setting $\alpha=\alpha''$ in \eqref{Eqn: Convergence rate2 aux3} that
\begin{align}\label{Eqn: Convergence rate2 aux4}
\Upsilon_0(\alpha_0-\alpha,\alpha_0)+\Psi_0(\alpha_0-\alpha)\le 0~.
\end{align}
Combining \eqref{Eqn: Convergence rate2 aux3} and \eqref{Eqn: Convergence rate2 aux4} with linearity of $b\mapsto \Upsilon_0(b,\alpha_0)+\Psi_0(b)$ we obtain
\begin{align}\label{Eqn: Convergence rate2 aux5}
\Upsilon_0(\alpha-\alpha_0,\alpha_0)+\Psi_0(\alpha-\alpha_0)= 0~,
\end{align}
whenever $d_{\mathbf A}(\alpha,\mathbf A_0)\le\epsilon$. Moreover, by Assumption \ref{Ass: Convergence rate, approximation}(i) and by making $c$ larger if necessary, $\|\alpha-\alpha_0\|_{\mathbf A}\le c\varrho_n$ whenever $d_{\mathbf A}(\alpha,\mathbf A_0)\le\epsilon$ and $\alpha=\Pi(\eta)$ for $\eta\in\Delta\Gamma_n$ and a linear $\Pi\in\mathbf D_n$ such that $\varrho_n\|\Pi-\Pi_0\|_{op,n}\le\epsilon$. This, together with \eqref{Eqn: Convergence rate2 aux5}, allows us to employ arguments analogous to those leading to \eqref{Eqn: Convergence rate aux5} and \eqref{Eqn: Convergence rate aux8} to obtain that
\begin{align}
Pf(\alpha,\nu) -Pf(\alpha_0,\nu) &\ge -\varpi_1\|\alpha-\alpha_0\|_{\mathbf A}^2 - \frac{2c_1}{\varpi_1}  \|\nu-\nu_0\|_{\mathbf B}^2~, \label{Eqn: Convergence rate2 aux6} \\
Pf(\alpha,\nu) -Pf(\alpha_0,\nu) &\le -\frac{\varkappa}{2}\|\alpha-\alpha_0\|_{\mathbf A}^2+\frac{2c_4\varrho_n^2}{\varkappa}\|\nu-\nu_0\|_{\mathbf B}^2~,\label{Eqn: Convergence rate2 aux7}
\end{align}
whenever $\nu\in V_0$ and $d_{\mathbf A}(\alpha,\mathbf A_0)\le\epsilon$ with $\alpha=\Pi(\eta)$ for $\eta\in\Delta\Gamma_n$ and a linear $\Pi\in\mathbf D_n$ such that $\varrho_n\|\Pi-\Pi_0\|_{op,n}\le\epsilon$, where the constants here and below may differ from the corresponding ones in the proof of Theorem \ref{Thm: Convergence rate}. If Assumption \ref{Ass: Convergence rate, Pi}(iv) holds with $\hat\Pi_n(\Delta\Gamma_n)\subset\mathbf A_0$ with probability approaching one, then \eqref{Eqn: Convergence rate2 aux6} and \eqref{Eqn: Convergence rate2 aux7} hold whenever $\alpha\in\mathbf A_0$ and $\nu\in V_0$ with probability approaching one by the same arguments as in the proof of Theorem \ref{Thm: Convergence rate}. For brevity, below we focus on the case with  $\mathfrak C_0(\alpha)-\mathfrak C_0(\alpha_0) \lesssim \|\alpha-\alpha_0\|_{\mathbf A}^\tau$ whenever $d_{\mathbf A}(\alpha,\mathbf A_0)\le\epsilon$.

As in the proof of Theorem \ref{Thm: Convergence rate}, \eqref{Eqn: Convergence rate2 aux6} and \eqref{Eqn: Convergence rate2 aux7} imply that, whenever $\eta\in\Delta\Gamma_n$, $\varrho_n\|\Pi-\Pi_0\|_{op,n}\le\epsilon$, $d_{\mathbf A}(\Pi(\eta),\mathbf A_0)\vee d_{\mathbf A}(\Pi(\eta_{0,n}),\mathbf A_0)\le\epsilon$, and $\nu\in V_0$,
\begin{multline}\label{Eqn: Convergence rate2 aux8}
Pf(\Pi(\eta),\nu)-Pf(\Pi(\eta_{0,n}),\nu)\\
\le -\frac{\varkappa}{4}\|\Pi(\eta)-\Pi(\eta_{0,n})\|_{\mathbf A}^2+\varpi_2\|\Pi(\eta_{0,n})-\alpha_0\|_{\mathbf A}^2+c_5\varrho_n^2\|\nu-\nu_0\|_{\mathbf B}^2~.
\end{multline}
Assume $\varkappa\in(0,1)$ and $\varpi_2>1$ without loss of generality. For any linear $\Pi\colon\mathbf H_0\to\mathbf A$ and $\eta\in\Delta\Gamma_n$, Assumption \ref{Ass: Convergence rate, approximation}(i) and the triangle inequality imply that
\begin{align}\label{Eqn: Convergence rate2 aux9}
\|\Pi(\eta_{0,n})-\alpha_0\|_{\mathbf A}\le \varrho_n\|\Pi-\Pi_0\|_{op,n}+\|\alpha_{0,n}-\alpha_0\|_{\mathbf A}~.
\end{align}
Results \eqref{Eqn: Convergence rate2 aux8} and \eqref{Eqn: Convergence rate2 aux9} together yield that, whenever $\eta\in\Delta\Gamma_n$, $\varrho_n\|\Pi-\Pi_0\|_{op,n}\le\epsilon$, $d_{\mathbf A}(\Pi(\eta),\mathbf A_0)\vee d_{\mathbf A}(\Pi(\eta_{0,n}),\mathbf A_0)\le\epsilon$, and $\nu\in V_0$,
\begin{multline}\label{Eqn: Convergence rate2 aux10}
Pf(\Pi(\eta),\nu)-Pf(\Pi(\eta_{0,n}),\nu)\le -\frac{\varkappa}{4}\|\eta-\eta_{0,n}\|_{\Pi}^2+2\varpi_2\varrho_n^2\|\Pi-\Pi_0\|_{op,n}^2\\
+2\varpi_2\|\alpha_{0,n}-\alpha_0\|_{\mathbf A}^2+c_5\varrho_n^2\|\nu-\nu_0\|_{\mathbf B}^2~.
\end{multline}

Similar to the proof of Theorem \ref{Thm: Convergence rate}, it suffices to show that
\begin{multline}\label{Eqn: Convergence rate2 aux11}
\|\hat\eta_n-\eta_{0,n}\|_{\hat\Pi_n}+\{\lambda_n\bar J_n(\hat\eta_n)\}^{1/2}\\
=O_p(\delta_n+\varrho_n\|\hat\nu_n-\nu_0\|_{\mathbf B}+\varrho_n\|\hat\Pi_n-\Pi_0\|_{op,n}+\{\lambda_n\bar J_n(\eta_{0,n})\}^{1/2})~.
\end{multline}
For each $n\in\mathbf N$, $j\in\mathbf Z$, and each $M>0$, define $F_n\equiv \{(\eta,\nu,\Pi)\in \Delta\Gamma_n\times\mathbf B_n\times \mathbf D_n\colon \eta\in\Delta\Gamma_n,\varrho_n\|\Pi-\Pi_0\|_{op,n}\le\epsilon,d_{\mathbf A}(\Pi(\eta),\mathbf A_0)\le\epsilon, d_{\mathbf A}(\Pi(\eta_{0,n}),\mathbf A_0)\le\epsilon,\nu\in V_0\}$ and 
\begin{align*}
S_{n,j,M}&\equiv \{(\eta,\nu,\Pi)\in \Delta\Gamma_n\times\mathbf B_n\times \mathbf D_n\colon 2^{j-1}\delta_n <\|\eta-\eta_{0,n}\|_{\Pi} + \{\lambda_n\bar J_n(\eta)\}^{\frac{1}{2}}\le 2^j\delta_n,\notag\\
& \hspace{-0.8cm} 2^M\max\{\varrho_n\|\nu-\nu_0\|_{\mathbf B},\varrho_n\|\Pi-\Pi_0\|_{op,n},\{\lambda_n\bar J_n(\eta_{0,n})\}^{\frac{1}{2}}\}\le \|\eta-\eta_{0,n}\|_\Pi + \{\lambda_n\bar J_n(\eta)\}^{\frac{1}{2}}\}~,\\
\mathcal E_{n,M}&\equiv\{\|\hat\eta_n-\eta_{0,n}\|_{\hat\Pi_n}+\{\lambda_n\bar J_n(\hat\eta_n)\}^{\frac{1}{2}}\ge 2^M[\delta_n\notag\\
&\hspace{4.25cm}+\varrho_n\|\hat\Pi_n-\Pi_0\|_{op,n}+\varrho_n\|\hat\nu_n-\nu_0\|_{\mathbf B}+\{\lambda_n\bar J_n(\eta_{0,n})\}^{\frac{1}{2}}]\}~.
\end{align*}
By arguments analogous to those leading to \eqref{Eqn: Convergence rate aux10}, it follows from \eqref{Eqn: Convergence rate2 aux8} that, whenever $(\eta,\nu,\Pi)\in S_{n,j,M}\cap F_n$, we have that, for some constant $c_6>0$,
\begin{multline}\label{Eqn: Convergence rate2 aux12}
Pf(\Pi(\eta),\nu) -Pf(\Pi(\eta_{0,n}),\nu)-\lambda_n\bar J_n(\eta)+\lambda_n\bar J_n(\eta_{0,n})\\
\le -(\frac{\varkappa}{32}-\frac{2c_6\varpi_2}{2^{2j}}-\frac{2\varpi_2}{2^{2M}}-\frac{c_5}{2^{2M}}-\frac{1}{2^{2M}})(2^j\delta_n)^2~.
\end{multline}

Define $R_n\equiv \sup_{a\in\hat{\mathbf A}_n}\{\hat{\mathfrak C}_n(a)-\lambda_nJ_n(a)\}-[\hat{\mathfrak C}_n(\hat\alpha_n)-\lambda_nJ_n(\hat\alpha_n)]$. Then by the definition of $\hat\eta_n$ we note that 
$0\le R_n\le \delta_n^2$ and that
\begin{align}\label{Eqn: Convergence rate2 aux13}
\mathbb P_nf(\hat\Pi_n(\hat\eta_n),\hat\nu_n)-\mathbb P_nf(\hat\Pi_n(\eta_{0,n}),\hat\nu_n) \ge \lambda_n\bar J_n(\hat\eta_n)-\lambda_n\bar J_n(\eta_{0,n})-R_n~.
\end{align}
Combining \eqref{Eqn: Convergence rate2 aux12} and \eqref{Eqn: Convergence rate2 aux13} yields that, whenever $(\hat\eta_n,\hat\nu_n,\hat\Pi_n)\in S_{n,j,M}\cap F_n$,
\begin{multline}\label{Eqn: Convergence rate2 aux14}
\sup_{(\eta,\nu,\Pi)\in S_{n,j,M}} \mathbb G_n(f(\Pi(\eta),\nu)-f(\Pi(\eta_{0,n}),\nu))\\
\ge \sqrt n \{(\frac{\varkappa}{32}-\frac{2c_6\varpi_2}{2^{2j}}-\frac{2\varpi_2}{2^{2M}}-\frac{c_5}{2^{2M}}-\frac{1}{2^{2M}})(2^j\delta_n)^2-\delta_n^2\}~.
\end{multline}
Note that the right hand side of the above inequality is bounded below by $c_7\sqrt n (2^j\delta_n)^2$ for some constant $c_7>0$ and for all large $j$ and $M$. Moreover, observe that
\begin{multline}\label{Eqn: Convergence rate2 aux15}
\sup_{(\eta,\nu,\Pi)\in S_{n,j,M}} \mathbb G_n(f(\Pi(\eta),\nu)-f(\Pi(\eta_{0,n}),\nu)) \\
\le \sup_{(\eta,\nu,\Pi)\in U_{0,n}(2^j\delta_n) }|\mathbb G_n(f(\Pi(\eta),\nu)-f(\Pi(\eta_{0,n}),\nu))|~.
\end{multline}
By Assumptions \ref{Ass: Convergence rate, approximation}(i), \ref{Ass: Convergence rate, sample}(ii), and \ref{Ass: Convergence rate, Pi}(i)(ii), we know that $(\hat\eta_n,\hat\nu_n,\hat\Pi_n)\in F_n$ with probability approaching one. This, together with the results \eqref{Eqn: Convergence rate2 aux14} and \eqref{Eqn: Convergence rate2 aux15} and Assumption \ref{Ass: Convergence rate, Pi}(iii), allows us to conclude as in the proof of Theorem \ref{Thm: Convergence rate}. \qed

\end{appendices}

\phantomsection
\addcontentsline{toc}{section}{References}
\putbib[bibliography]

\end{bibunit}
\makeatletter
\gdef\@extra@b@citeb{}
\gdef\@extra@binfo{}
\makeatother

\newpage
\pagenumbering{arabic}

\null
\vskip 2em
\pdfbookmark[1]{Supplement Title}{Supplements}
\begin{center}
{\LARGE Supplement to ``Debiased Machine Learning: Identification, Estimation, and Shape Constraints'' \par}
\vskip 1.5em
{\large
\lineskip .5em
\begin{tabular}[t]{c}
Qihui Chen\\ School of Management and Economics \\CUHK--Shenzhen\\ qihuichen@cuhk.edu.cn
\end{tabular}
\hskip 1em plus .17fil
\begin{tabular}[t]{c}
Ka Yan Cheng \\ Department of Economics \\ Emory University\\ ka.yan.cheng@emory.edu
\end{tabular}
\hskip 1em plus .17fil
\begin{tabular}[t]{c}
Zheng Fang \\ Department of Economics \\ Emory University\\ zheng.fang@emory.edu
\end{tabular}
\par}
\end{center}
\par
\vskip 1.5em

\begin{appendices}
\setcounter{section}{0}
\renewcommand{\thesection}{S.\arabic{section}}
\titleformat{\section}{\normalfont\Large\bfseries}{\thesection}{1em}{}
\renewcommand{\theequation}{\thesection.\arabic{equation}}
\numberwithin{equation}{section}

This supplement collects additional technical details and supporting results. 

\section{Additional Technical Details}\label{Sec: Details of examples}

In this section, we discuss an extension of Example \ref{Ex: NPIV} and illustrate verifications of some assumptions in Section \ref{Sec: Riesz regression}.

\subsection{Extension of Example \ref{Ex: NPIV}}\label{Sec: NPIV extensions}

Let $\gamma_0\in\mathbf H=L^2(Z)$ be a solution to $E[Y|W]=E[\gamma_0(Z)|W]$ that is not necessarily unique. Let $\theta_0=E[m(X,\gamma_0)]$ be an identified functional of $\gamma_0$, where $\gamma\mapsto E[m(X,\gamma)]$ is continuous and linear. By the Riesz representation theorem, there exists a unique $\nu_m\in\overline{\mathrm{lin}}(\Gamma-\gamma_0)$ such that $E[m(X,\gamma)]=E[\nu_m(Z)\gamma(Z)]$ for all $\gamma\in\mathbf H$. The debiasing term for orthogonalization continues to take the same form as in \eqref{Eqn: NPIVinfluence} \citep{BennettKallusMaoNeweySyrgkanisUehara2025StrongID}: for all $\gamma\in\mathbf H$ and $\alpha\in\mathbf A$,
\begin{align}
\phi(X,\theta,\gamma,\alpha) = \alpha(W)\{Y-\gamma(Z)\}~.
\end{align}
In particular, $\phi(\cdot,\theta,\gamma,\alpha)$ is linear in $\alpha\in\mathbf A$, allowing the Riesz representer $\alpha_0$ to be characterized through a quadratic optimization problem via Theorem \ref{Thm: Characterization}.

Specifically, assume that there exists some $\bar{\nu}_m\in\mathbf A$ such that  $E[\bar{\nu}_{m}(W)|Z] = \nu_{m}(Z)$. Then Assumption \ref{Ass: Riesz} holds for every $\gamma_0$ in the identified set, with the same $\Pi_0$, $\Psi_0$, and $\Upsilon_0$ as in Example \ref{Ex: NPIV} by the arguments in \eqref{Eqn: NPIV, Psi} and \eqref{Eqn: NPIV, Upsilon}. By Theorem \ref{Thm: Characterization}, there is a unique $\alpha_0\in\mathbf A_0$ satisfying \eqref{Eqn: Local Orthogonal wtr gamma} for all $\delta\in\Gamma-\gamma_0$ which also uniquely maximizes
\begin{align}
\mathfrak C_0(\alpha)
=
-\frac12 E[\alpha(W)^2]
+
E[\bar{\nu}_m(W)\alpha(W)]
\end{align}
over $\mathbf A_0$. Moreover, by Lemma \ref{Lem: linear span}, $\mathbf A_0=\overline{\Pi_0(\mathbf H_0)}$ with $\mathbf H_0=\mathrm{lin}(\Gamma-\gamma_0)$ is independent of the particular choice of $\gamma_0$ within the identified set. Consequently, the characterization of $\alpha_0$ depends only on $\nu_m$ and $\mathbf A_0$, and therefore does not rely on whether $\gamma_0$ is point identified. Hence, the same Riesz regression procedure applies under both point and partial identification of $\gamma_0$, a conclusion that is consistent with the findings of \citet{Chen2021RobustPLM} and \citet{BennettKallusMaoNeweySyrgkanisUehara2025StrongID} in the case without shape constraints.

Our framework differs from \citet{BennettKallusMaoNeweySyrgkanisUehara2025StrongID} in several important respects. First, we focus on a more general parameter $\theta_0$ defined by a moment condition and a more general first step, while \citet{BennettKallusMaoNeweySyrgkanisUehara2025StrongID} study a strongly identified linear functional $\theta_0$ so as to focus on a partially identified first step defined by a class of linear conditional moment restrictions. Second, we accommodate general nonlinear shape constraints on $\gamma_0$. While \citet{BennettKallusMaoNeweySyrgkanisUehara2025StrongID} allow the parameter space $\Gamma$ for $\gamma_0$ to be convex, they require $\Gamma$ to have nonempty interior, thereby ruling out common shape constraints.
%\footnote{In settings with exogenous first steps, \citet{LaanBibautKallusLuedtke2026AutoDML} also mentioned the possibility that their framework accommodates nonlinear constraints on $\gamma_0$. However, they require a first order optimality condition that may not hold under common shape constraints.} 
Third, specializing to the setup of \citet{BennettKallusMaoNeweySyrgkanisUehara2025StrongID}, we establish identification of the Riesz representer $\alpha_0$ and develop a general sieve-based Riesz regression procedure (see Remark \ref{Rem: linear term}). In contrast, \citet{BennettKallusMaoNeweySyrgkanisUehara2025StrongID} directly impose a ``strong identification'' condition to guarantee the existence of $\alpha_0$. Our analysis complements theirs by showing that a Riesz regression procedure can be developed without imposing this strong identification condition at the expense of a suitable sieve approximation. 

Although the same Riesz regression procedure applies under both point and partial identification, the estimation of the nuisance function $\gamma_0$ itself requires additional care when the NPIV model is only partially identified. In particular, one must select a representative element from the identified set of $\gamma_0$ before constructing the debiased estimator; see Section 3.1 of \citet{Chen2021RobustPLM} for a concrete example. Existing approaches typically achieve this by introducing a strictly convex penalty that uniquely selects an element from the identified set. From this perspective, shape restrictions can also be viewed as a form of regularization. By restricting the admissible set of solutions, shape constraints may substantially shrink the identified set and sharpen identification under partial identification; see, e.g., \citet{FreybergerHorowitz2015ShapeID}. Consequently, they may facilitate the selection of a stable representative element of $\gamma_0$ for subsequent debiasing and inference. Developing the estimation theory for $\gamma_0$ and establishing the convergence rates required in Section \ref{Sec: AN} under partial identification and shape restrictions are beyond the scope of the present paper and are left for future research.
  
\iffalse
We stress that, consistent with \citet{BennettKallusMaoNeweySyrgkanisUehara2025StrongID}, we do not require the identification of the first step $\gamma_0$ for $\sqrt n$-estimation and debiased machine learning of $\theta_0$---see also \citet{Santos2011Recover} and \citet{Chen2020RobustPLM}. However, unlike \citet{BennettKallusMaoNeweySyrgkanisUehara2025StrongID}, we do not require the functional strong identification condition either because we work with the instrument space $\mathbf A$ directly: the existence and uniqueness of the Riesz representer $\alpha_0$ is guaranteed under our assumptions. {\color{red} NEED TO LOOK INTO THE CONNECTIONS MORE CLOSELY.}
\fi 

\subsection{The Curvature Condition on \texorpdfstring{$\mathfrak C_0$}{C0}}\label{Sec: curvature}

Next, we show how the curvature condition in Assumption \ref{Ass: Convergence rate, Pi}(iv) holds automatically in Example \ref{Ex: NPIV}. We consider the general case when $\gamma\mapsto E[m(X,\gamma)]$ is not necessarily linear. Given equation \eqref{Eqn: NPIV Psi general}, we may assume $\bar\nu_m\in\mathbf A_0$ without loss of generality by the projection theorem. Moreover, we may rewrite $\mathfrak C_0$ in \eqref{Eqn: NPIV C0 general} by simple algebra:
\begin{align}\label{Eqn: NPIV C0 rewritten}
\mathfrak C_0(\alpha)=-\frac{1}{2}E[(\alpha(W)-\bar\nu_m(W))^2]+\frac{1}{2}E[\bar\nu_m(W)^2]~,
\end{align}
for all $\alpha\in\mathbf A$. Note that $\mathfrak C_0$ is trivially well defined on the whole of $\mathbf A=L^2(W)$ (and hence any enlargement of $\mathbf A_0$ in $\mathbf A$). Equation \eqref{Eqn: NPIV C0 rewritten} implies that $\alpha_0=\bar\nu_m$, $\mathfrak C_0(\alpha_0)=E[\bar\nu_m(W)^2]/2$, and hence, for all $\alpha\in\mathbf A$,
\begin{align}
\mathfrak C_0(\alpha)-\mathfrak C_0(\alpha_0)=-\frac{1}{2}E[(\alpha(W)-\alpha_0(W))^2]=-\frac{1}{2}\|\alpha-\alpha_0\|_{\mathbf A}^2~.
\end{align}

\subsection{The Entropy Conditions}\label{Sec: Entropies}

Finally, we discuss the modulus of continuity condition in Section \ref{Sec: Riesz regression}. In general, the modulus of continuity $\delta_n$ may be derived via symmetrization or bracketing. For $U_{0,n}(\delta)$ in \eqref{Eqn: sieve local size} and i.i.d.\ random variables $\{\varepsilon_i\}_{i=1}^n$ uniformly distributed on $\{-1,1\}$ and independent of $\{X_i\}_{i=1}^n$, symmetrization yields
\begin{multline}\label{Eqn: local Rademacher}
\frac{1}{\sqrt n}E[\sup_{(\eta,\nu,\Pi)\in U_{0,n}(\delta)}|\mathbb G_n(f(\Pi(\eta),\nu)-f(\Pi(\eta_{0,n}),\nu))|]\\
\lesssim E[\sup_{f\in\mathcal F_n(\delta)}|\frac{1}{n}\sum_{i=1}^{n}\varepsilon_if(X_i)|]\equiv \mathscr R_n(\mathcal F_n(\delta))~,
\end{multline}
where $\mathcal F_n(\delta)\equiv \{f(\Pi(\eta),\nu)-f(\Pi(\eta_{0,n}),\nu)\colon (\eta,\nu,\Pi)\in U_{0,n}(\delta)\}$ \citep{VaartWellner1996Book}. Let $\mathrm{star}(\mathcal F_n(\delta))\equiv\{af\colon f\in \mathcal F_n(\delta), a\in[0,1]\}$ be the star hull of $\mathcal F_n(\delta)$. Since $\mathscr R_n(\mathcal F_n(\delta))\le \mathscr R_n(\mathrm{star}(\mathcal F_n(\delta)))$ and $\mathscr R_n(\mathrm{star}(\mathcal F_n(\delta)))/\delta$ is nonincreasing in $\delta\in(0,\infty)$ \citep{Wainwright2019HighStats}, the critical radius $\delta_n^*\equiv\inf\{\delta>0\colon \mathscr R_n(\mathrm{star}(\mathcal F_n(\delta)))\le\delta^2\}$ serves as a valid candidate for $\delta_n$ satisfying Assumption \ref{Ass: Convergence rate, approximation}(v). Fortunately, upper bounds of critical radii have been derived for a variety of sieve spaces including neural networks \citep{FosterSyrgkanis2023Orthogonal,ChernozhukovNeweySinghSyrgkanis2024Adversarial}. 

 % , and $\mathscr R_n(\mathcal F_n(\delta))$ may be viewed as a version of the local Rademacher complexity \citep{KoltchinskiiPanchenko2002Classifier,BartlettBousquetMendelson2005Rademacher,RaskuttiWainwrightYu2012Minimax}  

To illustrate, recall that, for $U_{0,n}(\delta)$ defined in \eqref{Eqn: sieve local size},
\begin{align}
\mathcal F_n(\delta)=\{f(\Pi(\eta),\nu)-f(\Pi(\eta_{0,n}),\nu)\colon (\eta,\nu,\Pi)\in U_{0,n}(\delta)\}~.
\end{align}
The size of $\mathcal F_n(\delta)$ may be measured by uniform entropy or bracketing entropy. Let $N_{[\,]}(\epsilon,\mathcal F_n(\delta),\|\cdot\|_{P,2})$ be the minimum number of $\epsilon$-brackets (measured by $\|\cdot\|_{P,2}$) needed to cover $\mathcal F_n(\delta)$ and define the bracketing entropy integral 
\begin{align}
J_{[\,]}(\delta,\mathcal F_n(\delta),\|\cdot\|_{P,2})\equiv \int_0^\delta \sqrt{1+\log N_{[\,]}(\epsilon,\mathcal F_n(\delta),\|\cdot\|_{P,2})}\,\mathrm d\epsilon~.
\end{align}
Let $N(\epsilon,\mathcal F_n(\delta),\|\cdot\|_{Q,2})$  be the minimal number of $\|\cdot\|_{Q,2}$-balls of radius $\epsilon$ needed to cover $\mathcal F_n(\delta)$ and define the uniform entropy integral 
\begin{align}
J(\delta,\mathcal F_n(\delta),\|\cdot\|_2)\equiv \sup_{Q}\int_0^\delta\sqrt{1+\log N(\epsilon,\mathcal F_n(\delta),\|\cdot\|_{Q,2})}\,\mathrm d\epsilon~,
\end{align}
where the supremum is taken over all discrete probability measures $Q$. Also, define $\mathbf A_n\equiv \{\Pi(\eta)\colon \eta\in\Delta\Gamma_n,\Pi\in\mathbf D_n\}$ and recall that $\mathbf A_{0,n}\equiv \{\Pi_0(\eta)\colon \eta\in\Delta\Gamma_n\}$.

To formalize our discussions, we impose the following assumption. 

\begin{ass}\label{Ass: entropy1}
(i) $\mathcal F_n(\delta)$ is bounded under the uniform norm by some $\tau_n\ge 1$ and $\sup_{f\in\mathcal F_n(\delta)}\|f\|_{P,2}\lesssim \tau_n\delta$; (ii) $\sup_{f\in\mathcal F_n(\delta)}|\sum_{i=1}^{n}e_if(X_i)|$ and $\sup_{f\in\mathcal F_n(\delta)}|\sum_{i=1}^{n}e_if^2(X_i)|$ are measurable for every $n\in\mathbf N$, $(e_1,\ldots,e_n)\in\mathbf R^n$, and $\delta>0$.
\end{ass}

Assumption \ref{Ass: entropy1}(ii) is satisfied by neural networks and other common sieves, which may be verified by appealing to Example 2.3.4 in \citet{VaartWellner2023Book}. The uniform boundedness condition in Assumption \ref{Ass: entropy1}(i) is mild and allows for the bound to diverge. The moment condition $\sup_{f\in\mathcal F_n(\delta)}\|f\|_{P,2}\lesssim \tau_n\delta$ in Assumption \ref{Ass: entropy1}(i) follows by exploiting the particular structures of $\mathcal F_n(\delta)$. Specifically, by Assumption \ref{Ass: Convergence rate, population} and linearity of $\Pi\in\mathbf D_n$, we may obtain that, for $(\eta,\nu,\Pi)\in U_{0,n}(\delta)$,
\begin{multline}\label{Eqn: moment bound of sieve}
f(X,\Pi(\eta),\nu)-f(X,\Pi(\eta_{0,n}),\nu)\\
 = \frac{1}{2}[l(X,\Pi(\eta-\eta_{0,n}),\Pi(\eta+\eta_{0,n}),\nu)]+h(X,\Pi(\eta-\eta_{0,n}),\nu)~.
\end{multline}
In Examples \ref{Ex: average linear effects}--\ref{Ex: NPIV}, if $\mathbf A_n$ and $\mathbf B_n$ are bounded by $\sqrt{\tau_n}$ under the uniform norm for some $\tau_n\ge  1$, then \eqref{Eqn: moment bound of sieve} implies $\sup_{f\in\mathcal F_n(\delta)}\|f\|_{P,2}\lesssim \tau_n\delta$.

The following well known lemma (see, e.g., Section 3.4.2 in \citet{VaartWellner2023Book}) links the moment $E[\sup_{f\in\mathcal F_n(\delta)}|\mathbb G_nf|]$ to the sizes of $\mathcal F_n(\delta)$ measured by the uniform entropy and the bracketing entropy. 

\begin{lem}\label{Lem: entropy general}
If Assumptions \ref{Ass: Convergence rate, sample}(i) and \ref{Ass: entropy1}(i) hold, then $\delta_n>0$ satisfies the inequality $\omega_n(\delta_n)/\sqrt n\le \delta_n^2$ in Assumption \ref{Ass: Convergence rate, Pi}(iii) whenever
\begin{align}\label{Eqn: moment bound bracketing}
J_{[\,]}(\tau_n\delta_n,\mathcal F_n(\delta_n),\|\cdot\|_{P,2})\le \sqrt n \delta_n^2~. 
\end{align}
If in addition Assumption \ref{Ass: entropy1}(ii) holds, then the same holds whenever
\begin{align}\label{Eqn: moment bound symmetrization} % 
J(\tau_n\delta_n,\mathcal F_n(\delta_n),\|\cdot\|_2)\le \sqrt n \delta_n^2~. 
\end{align}
% \footnote{It suffices that the two suprema are measurable with respect to the completion of the underlying probability space for every $n\in\mathbf N$, $(e_1,\ldots,e_n)\in\mathbf R^n$, and $\delta>0$.} 
% The link of rates of convergence to the modulus of continuity of empirical processes: \citet{Geer1995Sieve}. 
\end{lem}

Given Lemma \ref{Lem: entropy general}, one may then verify Assumption \ref{Ass: Convergence rate, Pi}(iii) by leveraging recent entropy results of deep neural networks---see, e.g., \citet{KajiManresaPouliot2023Adversarial}, \citet{OuBolcskei2026Covering}, and references therein. We omit the details for brevity.

\iffalse
In deep learning, one may further bound the bracketing numbers of $\mathcal F_n(\delta)$. Suppose that $|l(X,a,b,\nu)|\le \{\|a\|_{\mathbf A,s}\vee \|\nu\|_{\mathbf B,s}\} F_n(X)$ and  $|h(X,a,\nu)|\le \{\|a\|_{\mathbf A,s}\vee \|\nu\|_{\mathbf B,s}\} F_n(X)$ for all $a,b\in \mathbf A_n-\mathbf A_n$, all $\nu\in \mathbf B_n-\mathbf B_n$, and some function $F_n$, where $\|\cdot\|_{\mathbf A,s}$ and $\|\cdot\|_{\mathbf B,s}$ are some norms (e.g., the uniform norm). By simple algebra we may then obtain that, for any $(\eta,\nu,\Pi),(\eta',\nu',\Pi')\in U_{0,n}(\delta)$,

\begin{align}
|[f(X,\Pi(\eta),\nu)-f(X,\Pi(\eta_{0,n}),\nu)]-[f(X,\Pi'(\eta'),\nu')-f(X,\Pi'(\eta_{0,n}),\nu')]|\le 
\end{align}

\textcolor{red}{MORE DETAILS FOR OUR ARCHITECTURE OF THE  NEURAL NETWORKS.}

% https://arxiv.org/abs/2410.06378
% Kaji, Shimotsu & Suzuki: https://arxiv.org/pdf/2007.06169
\fi

\noindent{\sc Proof of Lemma \ref{Lem: entropy general}:}  By Theorem 2.14.17' in \citet{VaartWellner2023Book} (with $\delta$ there identified with $\tau_n\delta$), we note that
\begin{multline}\label{Eqn: moment bound bracketing1}
E[\sup_{f\in\mathcal F_n(\delta)}|\mathbb G_nf|]\lesssim J_{[\,]}(\tau_n\delta,\mathcal F_n(\delta),\|\cdot\|_{P,2})(1+\frac{J_{[\,]}(\tau_n\delta,\mathcal F_n(\delta),\|\cdot\|_{P,2})}{\sqrt n\tau_n^2\delta^2}\tau_n)\\
\le J_{[\,]}(\tau_n\delta,\mathcal F_n(\delta),\|\cdot\|_{P,2})(1+\frac{J_{[\,]}(\tau_n\delta,\mathcal F_n(\delta),\|\cdot\|_{P,2})}{\sqrt n\delta^2})~,
\end{multline}
where the second inequality is due to $\tau_n\ge 1$. If \eqref{Eqn: moment bound bracketing} holds, then \eqref{Eqn: moment bound bracketing1} implies
\begin{align}\label{Eqn: moment bound bracketing2}
E[\sup_{f\in\mathcal F_n(\delta)}|\mathbb G_nf|]\lesssim J_{[\,]}(\tau_n\delta,\mathcal F_n(\delta),\|\cdot\|_{P,2})~.
\end{align}
The conclusion of the lemma then follows by taking the right hand side of \eqref{Eqn: moment bound bracketing2} to be $\omega_n(\delta)$. If the measurability condition holds, then by Theorem 2.14.2 in \citet{VaartWellner2023Book}, rescaling, and change of variables, we may obtain that
\begin{multline}\label{Eqn: moment bound symmetrization1}
E[\sup_{f\in\mathcal F_n(\delta)}|\mathbb G_nf|]\lesssim J(\tau_n \delta,\mathcal F_n(\delta),\|\cdot\|_2)(1+\frac{J(\tau_n \delta,\mathcal F_n(\delta),\|\cdot\|_2)}{\sqrt n\tau_n\delta^2})\\
\le J(\tau_n \delta,\mathcal F_n(\delta),\|\cdot\|_2)(1+\frac{J(\tau_n \delta,\mathcal F_n(\delta),\|\cdot\|_2)}{\sqrt n\delta^2})~.
\end{multline}
Thus, if \eqref{Eqn: moment bound symmetrization} holds, then arguments analogous to those leading to the first conclusion of the lemma allow us to obtain the second conclusion. \qed

\section{Supporting Results}\label{Sec: Supporting}

The theorem below is concerned with conditions under which there exists for each $\Psi\in\mathbf A_0^*$ (the topological dual of $\mathbf A_0$) a unique $\alpha=\alpha_0$ such that, for all $b\in\mathbf A_0$ ,
\begin{align}\label{Eqn: Riesz well-posedness}
\Psi (b)+\Upsilon_0(\alpha_0,b)=0~.
\end{align}

\begin{thm}\label{Thm: Riesz ID2}
Let $\Upsilon_0\colon \mathbf A_0\times\mathbf A_0\to\mathbf R$ be continuous, bilinear, symmetric, and positive/negative for a Hilbert space $\mathbf A_0$ with norm $\|\cdot\|_{\mathbf A}$. Then $\Upsilon_0$ is coercive if and only if there exists for each $\Psi\in\mathbf A_0^*$ a unique $\alpha_0\in\mathbf A_0$ satisfying \eqref{Eqn: Riesz well-posedness} for all $b\in\mathbf A_0$.
\end{thm}
\noindent{\sc Proof:} This is essentially Exercise 25.7 in \citet{ErnGuermond2021FiniteII} and we include a proof here for completeness. Without loss of generality, assume that $\Upsilon_0\colon \mathbf A_0\times\mathbf A_0\to\mathbf R$ is continuous, bilinear, symmetric, and positive; the case with negative $\Upsilon_0$ can be dealt with by considering $-\Upsilon_0$. Define $\Phi_0\colon \mathbf A_0\to \mathbf A_0^*$ by $\Phi_0(a)(b)=\Upsilon_0(a,b)$. It is clear that the bijectivity of $\Phi_0$ is equivalent to the existence and uniqueness of $\alpha_0$ for each $\Psi\in\mathbf A_0^*$ such that \eqref{Eqn: Riesz well-posedness} holds for all $b\in\mathbf A_0$. By $\Upsilon_0$ being bilinear, continuous, symmetric, and positive and Lemma 4.2.1 in \citet{BoffiBrezziFortin2013MixedFE}, we have: 
\begin{align}\label{Thm: Riesz ID2, aux1}
|\Upsilon_0(a,b)|^2\le  \Upsilon_0(a,a) \Upsilon_0(b,b)~,
\end{align}
for all $a,b\in\mathbf A_0$. Moreover, there exists some constant $C>0$ such that $\Upsilon_0(b,b)\le C\|b\|_{\mathbf A}^2$. Therefore, we may deduce from \eqref{Thm: Riesz ID2, aux1} that
\begin{align}\label{Thm: Riesz ID2, aux2}
\Upsilon_0(a,a)\ge C^{-1} \sup_{b\in\mathbf A_0\backslash\{0\}}\frac{|\Upsilon_0(a,b)|^2}{\|b\|_{\mathbf A}^2}=C^{-1} \|\Phi_0(a)\|_{op}^2~,
\end{align}
where $\|\Phi_0(a)\|_{op}$ is the operator norm of $\Phi_0(a)\in\mathbf A_0^*$. 

Note that $\Phi_0$ is continuous and linear since $\Upsilon_0$ is continuous and bilinear; see, e.g., \citet[p.~14]{ErnGuermond2021FiniteII}. Suppose that $\Phi_0\colon \mathbf A_0\to \mathbf A_0^*$ is bijective. Then by Theorem C.49 in \citet{ErnGuermond2021FiniteII}, there is some $C'>0$ such that 
\begin{align}\label{Thm: Riesz ID2, aux3}
\|\Phi_0(a)\|_{op}\ge C'\|a\|_{\mathbf A}~.
\end{align}
The coercivity of $\Upsilon_0$ then follows by combining results \eqref{Thm: Riesz ID2, aux2} and \eqref{Thm: Riesz ID2, aux3}. The converse follows directly by the Lax-Milgram theorem.\qed 

\begin{lem}\label{Lem: linear span}
Let $\Gamma$ be a nonempty set in a vector space and $\gamma_0\in\Gamma$. Then $\mathrm{lin}(\Gamma-\gamma_0)=\mathrm{lin}(\Gamma-\Gamma)$. If in addition $0\in\Gamma$, then $\mathrm{lin}(\Gamma-\gamma_0)=\mathrm{con}(\Gamma)-\mathrm{con}(\Gamma)$.
\end{lem}
\noindent{\sc Proof:} Clearly, $\mathrm{lin}(\Gamma-\gamma_0)\subset\mathrm{lin}(\Gamma-\Gamma)$ since $\gamma_0\in\Gamma$. The reverse is also obvious since $\gamma_1-\gamma_2=(\gamma_1-\gamma_0)-(\gamma_2-\gamma_0)$ for any $\gamma_1,\gamma_2\in\Gamma$. For the second claim, suppose that $0\in\Gamma$ and let $\gamma\in\mathrm{lin}(\Gamma-\gamma_0)$. Then, for some $a_j\in\mathbf R$, $\gamma_j\in\Gamma$, and $J<\infty$, we have:   
\begin{align}\label{Eqn: linear span, aux1}
\gamma=\sum_{j=1}^{J} a_j\{\gamma_j-\gamma_0\}=\sum_{j=1}^{J} a_j^+\gamma_j+\sum_{j=1}^{J} a_j^-\gamma_0 -[\sum_{j=1}^{J} a_j^-\gamma_j+\sum_{j=1}^{J} a_j^+\gamma_0]~,
\end{align}
where $a^+\equiv \max\{a,0\}$ and $a^-\equiv \max\{-a,0\}$ so that $a=a^+-a^-$ for any $a\in\mathbf R$. It follows from \eqref{Eqn: linear span, aux1} and $\gamma_0\in\Gamma$ that $\gamma\in\mathrm{con}(\Gamma)-\mathrm{con}(\Gamma)$. Conversely, suppose $\gamma\in \mathrm{con}(\Gamma)-\mathrm{con}(\Gamma)$. Then for some $a_j,b_j\in\mathbf R_+$, $\gamma_j\in\Gamma$, and $J<\infty$, we may write
\begin{align}\label{Eqn: linear span, aux2}
\gamma=\sum_{j=1}^{J}a_j\gamma_j-\sum_{j=1}^{J}b_j\gamma_j=\sum_{j=1}^{J}(a_j-b_j)\{\gamma_j-\gamma_0\}-\sum_{j=1}^{J}(a_j-b_j)\{0-\gamma_0\}~,
\end{align}
which implies $\gamma\in\mathrm{lin}(\Gamma-\gamma_0)$ since $0\in\Gamma$. This completes the proof of the lemma. \qed

\iffalse
\begin{lem}\label{Lem: Riesz consistency}
Let Assumption \ref{Ass: Convergence rate, population} hold. If $J_n\colon\mathbf A_{0,n}\to\mathbf R$ convex and $\alpha\mapsto l(X,\alpha,\alpha,\nu)$ concave whenever $\|\nu-\nu_0\|_{\mathbf B}<\delta$ for some $\delta>0$, then it follows that $\|\hat\alpha_n-\alpha_{0,n}\|_{\mathbf A}=o_p(1)$.
\end{lem}
\noindent{\sc Proof:} By Proposition 25.25 in \citet{Zeidler1990IIB}, $\mathfrak C_0(\alpha)$ is concave in $\alpha\in\mathbf A_{0,n}$.

and $\hat{\mathfrak C}_n(\alpha)-\lambda_nJ(\alpha)$ are both

This lemma is a variant of Problem 3.2.6 in \citet{VaartWellner2023Book}---note that $\hat\alpha_n$ is an approximate maximizer.

The proof strategy outlined by van der Vaart remains valid, we thus omit the proof for simplicity. \qed
\fi

\begin{lem}\label{Lem: error bound}
Let the conditions of Theorem \ref{Thm: Riesz ID} hold with $\Upsilon_0$ negative. If $\alpha_{0,n}\in\mathbf A_{0,n}\subset\mathbf A_0$ maximizes $\mathfrak C_0$ over $\mathbf A_{0,n}$ and $\mathbf A_{0,n}$ is starshaped at $\alpha_{0,n}$,\footnote{A set $A$ in a vector space is starshaped at $a_0$ if $a_0+t(a-a_0)\in A$ for any $a\in A$ and $t\in[0,1]$.} then it follows that 
\begin{align}
\|\alpha_{0,n}-\alpha_0\|_{\mathbf A}\le \frac{\|\Upsilon_0\|_{op}}{\varkappa}d_{\mathbf A}(\alpha_0,\mathbf A_{0,n})~.
\end{align}
\end{lem}
\noindent{\sc Proof:} This is a slight generalization of C\'{e}a's lemma (see, e.g., Lemma 26.13 in \citet{ErnGuermond2021FiniteII}). By coercivity of $\Upsilon_0$,  we note that, for all $\alpha\in\mathbf A_{0,n}$, 
\begin{align}\label{Eqn: error bound, aux1}
\varkappa \|\alpha_{0,n}-\alpha_0\|_{\mathbf A}^2 & \le -\Upsilon_0(\alpha_{0,n}-\alpha_0,\alpha_{0,n}-\alpha_0)\notag\\
&\le -\Upsilon_0(\alpha_{0,n}-\alpha_0,\alpha-\alpha_0)-\Upsilon_0(\alpha_{0,n}-\alpha_0,\alpha_{0,n}-\alpha)~,
\end{align}
where the second inequality is due to the linearity of $b\mapsto \Upsilon_0(a,b)$. Since $\mathbf A_{0,n}$ is starshaped at $\alpha_{0,n}$ and $\alpha_{0,n}$ maximizes $\mathfrak C_0$ over $\mathbf A_{0,n}$, we obtain by the linearity and symmetry in Assumption \ref{Ass: Riesz} that, for any $b\in\mathbf A_{0,n}-\alpha_{0,n}$,
\begin{align}\label{Eqn: error bound, aux2}
0\ge\lim_{t\downarrow 0} \{\mathfrak C_0(\alpha_{0,n}+tb)-\mathfrak C_0(\alpha_{0,n})\}=\Upsilon_0(\alpha_{0,n},b)+\Psi_0(b)~.
\end{align}
Moreover, by Theorem \ref{Thm: Riesz ID}, $\alpha_0$ satisfies $\Psi_0(b)+\Upsilon_0(\alpha_0,b)=0$ for all $b\in\mathbf A_{0,n}$. This, together with \eqref{Eqn: error bound, aux2} and linearity of $a\mapsto \Upsilon_0(a,b)$, implies that, for all $\alpha\in\mathbf A_{0,n}$,
\begin{multline}\label{Eqn: error bound, aux3}
-\Upsilon_0(\alpha_{0,n}-\alpha_0,\alpha_{0,n}-\alpha)=\Upsilon_0(\alpha_{0,n},\alpha-\alpha_{0,n})-\Upsilon_0(\alpha_0,\alpha-\alpha_{0,n})\\
=\Upsilon_0(\alpha_{0,n},\alpha-\alpha_{0,n})+\Psi_0(\alpha-\alpha_{0,n})\le 0~.
\end{multline}
Combining results \eqref{Eqn: error bound, aux1} and \eqref{Eqn: error bound, aux3} we may then obtain that
\begin{align}\label{Eqn: error bound, aux4}
\varkappa \|\alpha_{0,n}-\alpha_0\|_{\mathbf A}^2\le -\Upsilon_0(\alpha_{0,n}-\alpha_0,\alpha-\alpha_0)\le \|\Upsilon_0\|_{op}\|\alpha_{0,n}-\alpha_0\|_{\mathbf A}\|\alpha-\alpha_{0,n}\|_{\mathbf A}~,
\end{align}
where $\|\Upsilon_0\|_{op}<\infty$ since $\Upsilon_0$ is continuous and bilinear. In turn, \eqref{Eqn: error bound, aux4} implies 
\begin{align}\label{Eqn: error bound, aux5}
\varkappa \|\alpha_{0,n}-\alpha_0\|_{\mathbf A}\le \|\Upsilon_0\|_{op}\|\alpha-\alpha_{0,n}\|_{\mathbf A}
\end{align}
for all $\alpha\in\mathbf A_{0,n}$. The conclusion of the lemma follows by taking the infimum of the right hand side of \eqref{Eqn: error bound, aux5} over $\alpha\in\mathbf A_{0,n}$ and then dividing both sides by $\varkappa>0$. \qed

The next lemma establishes the consistency of the debiased GMM estimator. 

\begin{ass}\label{Ass: Model for consistency}
(i) $\theta_0\in\Theta$ uniquely satisfies \eqref{Eqn: GMM} for a compact $\Theta\subset\mathbf R^{d_\theta}$ and $\gamma_0\in\Gamma$ a nuisance parameter with $\Gamma$ a closed set in a normed space $\mathbf H$; (ii) $\phi\colon\mathcal X\times\Theta\times\mathbf H\times\mathbf A\to\mathbf R^{d_g}$ for $\mathbf A$ a normed space satisfies $E[\phi(X,\theta_0,\gamma_0,\alpha_0)]=0$ for some $\alpha_0\in\mathbf A$; (iii) there exists a neighborhood $U_0$ of $\gamma_0$ such that $\|g(X,\theta,\gamma)-g(X,\theta',\gamma)\|\le c(X,\gamma)\|\theta-\theta'\|^\varpi$ for all $\gamma\in U_0$, all $\theta,\theta'\in\Theta$, and some $\varpi\in (0,1]$ where $E[\sup_{\gamma\in U_0}c(X,\gamma)]<\infty$; (iv) $E[g(X,\theta,\cdot)]$ is continuous at $\gamma_0$ for all $\theta\in\Theta$; (v) $E[\phi(X,\cdot,\cdot,\cdot)]$ is continuous at $(\theta_0,\gamma_0,\alpha_0)$; (vi) $E[\sup_{\gamma\in U_0}\|g(X,\theta,\gamma)\|^{1+\varsigma}]<\infty$ for all $\theta\in\Theta$ and some $\varsigma>0$; (vii) there is a neighborhood $V_0$ of $(\theta_0,\gamma_0,\alpha_0)$ such that $E[\sup_{(\theta,\gamma,\alpha)\in V_0}\|\phi(X,\theta,\gamma,\alpha)\|^{1+\varsigma}]<\infty$.
\end{ass}

\begin{ass}\label{Ass: Data for consistency}
(i) $\{X_i\}_{i=1}^n$ is an i.i.d.\ sample of $X$; (ii) $\hat\gamma_{-k}\overset{p}{\to}\gamma_0$,  $\hat\alpha_{-k}\overset{p}{\to}\alpha_0$, and $\tilde\theta_{-k}\overset{p}{\to}\theta_0$ as $n\to\infty$ for each $k=1,\ldots, K$; (iii) $\min_{k=1}^K n_k\to\infty$ as $n\to\infty$; (iv) $\hat\Omega_n\overset{p}{\to}\Omega_0$ as $n\to\infty$ for some positive definite $\Omega_0$.
\end{ass}

\begin{lem}\label{Lem: DGMM consistency}
Let $\hat\theta_n\colon\{X_i\}_{i=1}^n\to\Theta$ be defined as in \eqref{Eqn: Debiased GMM}. If Assumptions \ref{Ass: Model for consistency} and \ref{Ass: Data for consistency} hold, then it follows that $\hat\theta_n\overset{p}{\to}\theta_0$ as $n\to\infty$.
\end{lem} % \citet{Rao2009Conditional}, \citet{NiuChakrabortyDukesKatsevich2024X}, \citet{BulinskiAlexander2017ConditionalCLT}
\noindent{\sc Proof:} First, we observe by Assumption \ref{Ass: Data for consistency}(i) the simple identity:
\begin{align}\label{Lem: DGMM consistency, aux1}
\frac{1}{n}\sum_{k=1}^{K}\sum_{i\in I_k }g(X_i,\theta,\hat\gamma_{-k})  &- E[g(X,\theta,\gamma_0)]\notag\\
& =\sum_{k=1}^{K}\frac{n_k}{n}\frac{1}{n_k}\sum_{i\in I_k }\{g(X_i,\theta,\hat\gamma_{-k}) - E[g(X_i,\theta,\hat\gamma_{-k})|\{X_j\colon j\notin I_k\}]\}\notag\\
& \qquad +\sum_{k=1}^{K} \frac{n_k}{n} \{E_X[g(X,\theta,\hat\gamma_{-k})]- E[g(X,\theta,\gamma_0)]\}~,
\end{align}
where $E_X$ denotes expectation taken with respect to $X$ with $\hat\gamma_{-k}$ held fixed. By Assumptions \ref{Ass: Model for consistency}(iv) and \ref{Ass: Data for consistency}(ii), the continuous mapping theorem gives
\begin{align}\label{Lem: DGMM consistency, aux2}
E_X[g(X,\theta,\hat\gamma_{-k})]- E[g(X,\theta,\gamma_0)] \overset{p}{\to} 0~,
\end{align}
for each $k=1,\ldots,K$. Result \eqref{Lem: DGMM consistency, aux2}, $ n_k/n\le 1$, and the triangle inequality then imply:
\begin{align}\label{Lem: DGMM consistency, aux3}
\|\sum_{k=1}^{K} \frac{n_k}{n} \{E_X[g(X,\theta,\hat\gamma_{-k})]- E[g(X,\theta,\gamma_0)]\}\|\overset{p}{\to} 0~,
\end{align}
for all $\theta\in\Theta$. Moreover, we have by Assumption \ref{Ass: Data for consistency}(i) that
\begin{align}\label{Lem: DGMM consistency, aux4}
\frac{1}{n_k^{1+\varsigma}}\sum_{i\in I_k} E[\|g(X_i,\theta,\hat\gamma_{-k})\|^{1+\varsigma}|\{X_j\colon j\notin I_k\}]=\frac{1}{n_k^{\varsigma}} E_X[\|g(X,\theta,\hat\gamma_{-k})\|^{1+\varsigma}]~,
\end{align}
and by Assumptions \ref{Ass: Model for consistency}(vi) and \ref{Ass: Data for consistency}(ii)(iii) that, for any $\epsilon>0$, 
\begin{align}\label{Lem: DGMM consistency, aux5}
P(\frac{1}{n_k^{\varsigma}} E_X[\|g(X,\theta,\hat\gamma_{-k})\|^{1+\varsigma}]>\epsilon) &\le P(\frac{1}{n_k^{\varsigma}} E[\sup_{\gamma\in U_0}\|g(X,\theta,\gamma)\|^{1+\varsigma}]>\epsilon)+P(\hat\gamma_{-k}\notin U_0)\notag\\
&\overset{p}{\to} 0~.
\end{align}
Combining results \eqref{Lem: DGMM consistency, aux4} and \eqref{Lem: DGMM consistency, aux5} then yields that
\begin{align}\label{Lem: DGMM consistency, aux6}
\frac{1}{n_k^{1+\varsigma}}\sum_{i\in I_k} E[\|g(X_i,\theta,\hat\gamma_{-k})\|^{1+\varsigma}|\{X_j\colon j\notin I_k\}]\overset{p}{\to} 0~.
\end{align}
By Assumption \ref{Ass: Data for consistency}(i)(iii) and result \eqref{Lem: DGMM consistency, aux6}, we thus obtain by Theorem B.4 in \citet[p.~2]{NiuChakrabortyDukesKatsevich2024X} that, for all $k=1,\ldots,K$ and all $\theta\in\Theta$,
\begin{align}\label{Lem: DGMM consistency, aux7}
\frac{1}{n_k}\sum_{i\in I_k }\{g(X_i,\theta,\hat\gamma_{-k}) - E[g(X_i,\theta,\hat\gamma_{-k})|\{X_j\colon j\notin I_k\}]\} \overset{p}{\to} 0
\end{align}
conditionally on $\{X_j\colon j\notin I_k\}$ and hence unconditionally by Lemma D.5 in \citet[p.~7]{NiuChakrabortyDukesKatsevich2024X}. Since $0\le n_k/n\le 1$, \eqref{Lem: DGMM consistency, aux7} and the triangle inequality imply
\begin{align}\label{Lem: DGMM consistency, aux8}
\|\sum_{k=1}^{K}\frac{n_k}{n}\frac{1}{n_k}\sum_{i\in I_k }\{g(X_i,\theta,\hat\gamma_{-k}) - E[g(X_i,\theta,\hat\gamma_{-k})|\{X_j\colon j\notin I_k\}]\}\|\overset{p}{\to}0~.
\end{align}
By results \eqref{Lem: DGMM consistency, aux3} and \eqref{Lem: DGMM consistency, aux8}, it follows from \eqref{Lem: DGMM consistency, aux1} that,  for each fixed $\theta\in\Theta$,
\begin{align}\label{Lem: DGMM consistency, aux9}
\frac{1}{n}\sum_{k=1}^{K}\sum_{i\in I_k }g(X_i,\theta,\hat\gamma_{-k})  - E[g(X,\theta,\gamma_0)]\overset{p}{\to}0~.
\end{align}
By Assumptions \ref{Ass: Model for consistency}(ii)(v)(vii) and \ref{Ass: Data for consistency}(i)(ii)(iii), we may apply arguments identical to those leading to \eqref{Lem: DGMM consistency, aux9} to deduce that
\begin{align}\label{Lem: DGMM consistency, aux10}
\frac{1}{n}\sum_{k=1}^{K}\sum_{i\in I_k }\phi(X_i,\tilde\theta_{-k},\hat\gamma_{-k},\hat\alpha_{-k})  \overset{p}{\to}  E[\phi(X,\theta_0,\gamma_0,\alpha_0)]=0~.
\end{align}

Next, by Assumption \ref{Ass: Model for consistency}(iii) and the triangle inequality, we have:
\begin{align}\label{Lem: DGMM consistency, aux11}
\|\frac{1}{n}\sum_{k=1}^{K}\sum_{i\in I_k }g(X_i,\theta,\gamma) -\frac{1}{n}\sum_{k=1}^{K}\sum_{i\in I_k }g(X_i,\theta',\gamma)\|\le \frac{1}{n}\sum_{k=1}^{K}\sum_{i\in I_k }c(X_i,
\gamma)\|\theta-\theta'\|^\varpi~,
\end{align}
whenever $\gamma\in U_0$. Fix $\epsilon,\eta>0$. By Assumptions \ref{Ass: Data for consistency}(i)(ii), result \eqref{Lem: DGMM consistency, aux11}, and Markov's inequality, we may obtain that, for $\delta\equiv\{\eta\epsilon/E[1\vee\sup_{\gamma\in U_0}c(X,\gamma)]\}^{1/\varpi}$,
\begin{multline}\label{Lem: DGMM consistency, aux12}
\limsup_{n\to\infty} P(\sup_{\|\theta-\theta'\|<\delta}\|\frac{1}{n}\sum_{k=1}^{K}\sum_{i\in I_k }g(X_i,\theta,\hat\gamma_{-k}) -\frac{1}{n}\sum_{k=1}^{K}\sum_{i\in I_k }g(X_i,\theta',\hat\gamma_{-k})\|>3\epsilon)\\
\le \limsup_{n\to\infty} \frac{\sum_{k=1}^{K}\frac{n_k}{n}E[E_X[c(X,\hat\gamma_{-k})1\{\hat\gamma_{-k}\in U_0\}]]\delta^{\varpi}}{\epsilon}=\eta~.
\end{multline}
By Assumption \ref{Ass: Model for consistency}(iii), the map $\theta\mapsto E[g(X,\theta,\gamma_0)]$ is uniformly continuous and hence, by making $\delta$ smaller if necessary, it follows that
\begin{align}\label{Lem: DGMM consistency, aux13}
\sup_{\|\theta-\theta'\|<\delta} \|E[g(X,\theta,\gamma_0)-g(X,\theta',\gamma_0)]\|\le\epsilon~.
\end{align}
By Assumption \ref{Ass: Model for consistency}(i), there exist $\theta_1,\ldots,\theta_J\in\Theta$ with $J<\infty$ such that $\min_{j=1}^J\|\theta-\theta_j\|<\delta$ for all $\theta\in\Theta$. Thus, the triangle inequality and \eqref{Lem: DGMM consistency, aux13} imply
\begin{align}\label{Lem: DGMM consistency, aux14}
\sup_{\theta\in\Theta} \|\frac{1}{n}\sum_{k=1}^{K}\sum_{i\in I_k }&g(X_i,\theta,\hat\gamma_{-k})-E[g(X,\theta,\gamma_0)]\|\notag\\
&\le \sup_{\|\theta-\theta'\|<\delta}\|\frac{1}{n}\sum_{k=1}^{K}\sum_{i\in I_k }g(X_i,\theta,\hat\gamma_{-k}) -\frac{1}{n}\sum_{k=1}^{K}\sum_{i\in I_k }g(X_i,\theta',\hat\gamma_{-k})\| \notag\\
&\quad + \max_{j=1}^J \|\frac{1}{n}\sum_{k=1}^{K}\sum_{i\in I_k }g(X_i,\theta_j,\hat\gamma_{-k})  - E[g(X,\theta_j,\gamma_0)]\|+\epsilon~.
\end{align}
Combining results \eqref{Lem: DGMM consistency, aux9}, \eqref{Lem: DGMM consistency, aux12}, and \eqref{Lem: DGMM consistency, aux14}, we may then conclude that 
\begin{align}\label{Lem: DGMM consistency, aux15}
\sup_{\theta\in\Theta} \|\frac{1}{n}\sum_{k=1}^{K}\sum_{i\in I_k }g(X_i,\theta,\hat\gamma_{-k})-E[g(X,\theta,\gamma_0)]\| = o_p(1)~.
\end{align}

For notational simplicity, define $\psi_0(\theta)\equiv E[g(X,\theta,\gamma_0)]$ and
\begin{align*}
\hat\psi_n(\theta)\equiv\frac{1}{n}\sum_{k=1}^{K}\sum_{i\in I_k }\{g(X_i,\theta,\hat\gamma_{-k}) + \phi(X_i,\tilde\theta_{-k},\hat\gamma_{-k}, \hat\alpha_{-k})\}~.
\end{align*}
By results \eqref{Lem: DGMM consistency, aux10} and \eqref{Lem: DGMM consistency, aux15}, we note by the triangle inequality that
\begin{align}\label{Lem: DGMM consistency, aux16}
\sup_{\theta\in\Theta}\|\hat\psi_n(\theta)-\psi_0(\theta)\|=o_p(1)~.
\end{align}
Moreover, $\sup_{\theta\in\Theta}\|\psi_0(\theta)\|<\infty$ by compactness of $\Theta$ in Assumption \ref{Ass: Model for consistency}(i) and continuity of $\theta\mapsto\psi_0(\theta)$ implied by Assumption \ref{Ass: Model for consistency}(iii). Therefore, it follows from result \eqref{Lem: DGMM consistency, aux16} that $\sup_{\theta\in\Theta}\|\hat\psi_n(\theta)\|=O_p(1)$. These facts, together with the triangle inequality, Assumption \ref{Ass: Data for consistency}(iv), and result \eqref{Lem: DGMM consistency, aux15}, imply that
\begin{align}\label{Lem: DGMM consistency, aux17}
\sup_{\theta\in\Theta}|\hat\psi_n(\theta)^\intercal \hat\Omega_n \hat\psi_n(\theta) -\psi_0(\theta)^\intercal \Omega_0\psi_0(\theta)| = o_p(1)~.
\end{align}
By Assumptions \ref{Ass: Model for consistency}(i) and \ref{Ass: Data for consistency}(iv), the map $\theta\mapsto \psi_0(\theta)^\intercal \Omega_0\psi_0(\theta)$ is uniquely minimized at $\theta_0$. This, together with the compactness of $\Theta$ in Assumption \ref{Ass: Model for consistency}(i), continuity of $\theta\mapsto E[g(X,\theta,\gamma_0)]$ and hence of $\theta\mapsto \psi_0(\theta)^\intercal \Omega_0\psi_0(\theta)$, and result \eqref{Lem: DGMM consistency, aux17} allows us to conclude by Theorem 2.1 in \citet{NeweyMcFadden1994Handbook} that $\hat\theta_n\overset{p}{\to}\theta_0$. \qed

The lemma below is a law of large numbers for cross-fitted sample averages where, as usual, $\hat\eta_{-k}$ depends on observations not in the $k$th fold of a partition $\{1,\ldots,n\}=\bigcup_{k=1}^K I_k$ for some fixed $K>1$ and $n_k$ is the sample size of the $k$th fold. 
% Note that in addition to cross-fitting, Lemma \ref{Lem: LLN} generalizes Lemma 4.3 in \citet{NeweyMcFadden1994Handbook} in that $f(X,\eta)$ may be discontinuous in $\eta$. 

\begin{lem}\label{Lem: LLN}
Let $f\colon\mathcal X\times\mathbf D_n\to\mathbf R^m$ for a normed space $\mathbf D_n$. If (i) $\{X_i\}_{i=1}^n$ is an i.i.d.\ sample of $X\in\mathcal X$; (ii) $\hat\eta_{-k}\overset{p}{\to}\eta_0$ as $n\to\infty$ for each $k=1,\ldots, K$; (iii) $E[f(X,\cdot)]$ is continuous at $\eta_0$; (iv) there is a neighborhood $U_0$ of $\eta_0$ such that $E[\sup_{\eta\in U_0}\|f(X,\eta)\|^{1+\varsigma}]<\infty$ for some $\varsigma>0$; and (v) $\min_{k=1}^K n_k\to\infty$ as $n\to\infty$, then
\begin{align}
\frac{1}{n}\sum_{k=1}^{K}\sum_{i\in I_k} f(X_i,\hat\eta_{-k})\overset{p}{\to} E[f(X,\eta_0)]~.
\end{align}
\end{lem}
\noindent{\sc Proof:} Without loss of generality, assume that $f$ is scalar-valued (i.e., $m=1$). By conditions (ii)(iii) and the continuous mapping theorem, we have:
\begin{align}\label{Lem: LLN, aux1}
E_X[f(X,\hat\eta_{-k})] \overset{p}{\to}E[f(X,\eta_0)]~,
\end{align}
for all $k=1,\ldots,K$, where $E_X$  indicates expectation taken with respect to $X$ holding $\hat\eta_{-k}$ fixed. By condition (i), we also note that
\begin{align}\label{Lem: LLN, aux2}
\frac{1}{n_k^{1+\varsigma}}\sum_{i\in I_k} E[|f(X_i,\hat\eta_{-k})|^{1+\varsigma}|\{X_j\colon j\notin I_k\}]=\frac{1}{n_k^{\varsigma}} E_X[|f(X,\hat\eta_{-k})|^{1+\varsigma}]~.
\end{align}
In turn, observe that conditions (ii)(iv)(v) imply, for any $\epsilon>0$, 
\begin{align}\label{Lem: LLN, aux3}
P(\frac{1}{n_k^{\varsigma}} E_X[|f(X,\hat\eta_{-k})|^{1+\varsigma}]>\epsilon) &\le P(\frac{1}{n_k^{\varsigma}} E[\sup_{\eta\in U_0}|f(X,\eta)|^{1+\varsigma}]>\epsilon)+P(\hat\eta_{-k}\notin U_0)\notag\\
&\overset{p}{\to} 0~.
\end{align}
Combining results \eqref{Lem: LLN, aux2} and \eqref{Lem: LLN, aux3} then yields that
\begin{align}\label{Lem: LLN, aux4}
\frac{1}{n_k^{1+\varsigma}}\sum_{i\in I_k} E[|f(X_i,\hat\eta_{-k})|^{1+\varsigma}|\{X_j\colon j\notin I_k\}]\overset{p}{\to} 0~.
\end{align}
By conditions (i)(v) and result \eqref{Lem: LLN, aux4}, we thus obtain by Theorem B.4 in \citet[p.~2]{NiuChakrabortyDukesKatsevich2024X} that, for all $k=1,\ldots,K$,
\begin{align}\label{Lem: LLN, aux5}
\frac{1}{n_k}\sum_{i\in I_k }\{f(X_i,\hat\eta_{-k}) - E[f(X_i,\hat\eta_{-k})|\{X_j\colon j\notin I_k\}]\} \overset{p}{\to} 0
\end{align}
conditionally on $\{X_j\colon j\notin I_k\}$ and hence also unconditionally by Lemma D.5 in \citet[p.~7]{NiuChakrabortyDukesKatsevich2024X}. Observing the simple identity
\begin{align}\label{Lem: LLN, aux6}
\frac{1}{n}\sum_{k=1}^{K}\sum_{i\in I_k} \{f(X_i,\hat\eta_{-k})&-E[f(X,\eta_0)]\}\notag\\
&=\sum_{k=1}^{K}\frac{n_k}{n}\frac{1}{n_k}\sum_{i\in I_k} \{f(X_i,\hat\eta_{-k})-E[f(X_i,\hat\eta_{-k})|\{X_j\colon j\notin I_k\}]\}\notag\\
&\quad+ \sum_{k=1}^{K}\frac{n_k}{n}\{ E_X[f(X,\hat\eta_{-k})]-E[f(X,\eta_0)]\}~,
\end{align}
we may then conclude the proof of the lemma by combining results \eqref{Lem: LLN, aux1} and \eqref{Lem: LLN, aux5} with the triangle inequality and the fact $0\le n_k/n\le 1$. \qed

\iffalse
\section{Discussions of Examples \ref{Ex: average linear effects}--\ref{Ex: synthetic control}}

\begin{ass}\label{Ass: average effects}
(i) $\gamma_0\in\Gamma\subset L^2(Z)$ for a nonempty closed set $\Gamma$; (ii) $E[m(X,\gamma)]$ is Gateaux differentiable at $\gamma=\gamma_0$ tangentially to $\Gamma-\gamma_0$; (iii) $\nabla_\gamma E[m(X,\gamma_0)][\delta]=\Psi_0(\delta)$ for a map $\Psi_0\colon \overline{\mathrm{lin}}(\Gamma-\gamma_0)\to \mathbf R$ that is continuous and linear; 
\end{ass}

\begin{pro}
Assumption \ref{Ass: average effects} implies Assumption \ref{Ass: Riesz} with $\Psi_0\colon $, $\Upsilon_0\colon $ is coercive and symmetric and
\begin{align}
J_n(\alpha)= 
\end{align}
\end{pro}

\fi

\noindent{\sc Proof of Proposition \ref{Pro: AN}:} By Assumptions \ref{Ass: Model}(i)(ii) and \ref{Ass: Data}(i), we may apply the mean value theorem to obtain for each $j=1,\ldots,d_g$ some $\check\theta_{n,j}\in\Theta^\circ$ lying between $\theta_0$ and $\hat\theta_n$ such that, whenever $(\hat\theta_n,\hat\gamma_{-k})\in U_0\times V_0$ for all $k=1,\ldots,K$, 
\begin{align}\label{Pro: AN, aux1}
\frac{1}{n}\sum_{k=1}^{K}\sum_{i\in I_k} g(X_i,\hat\theta_n,\hat\gamma_{-k}) =  \frac{1}{n}\sum_{k=1}^{K}\sum_{i\in I_k} g(X_i,\theta_0,\hat\gamma_{-k}) + \check{G}_n\{\hat\theta_n-\theta_0\}~,
\end{align}
where, for $f_j$ the $j$th coordinate of a vector-valued function $f$ here and in what follows,
\begin{align*}
\check{G}_n\equiv \frac{1}{n}\sum_{k=1}^{K}\sum_{i\in I_k} \begin{bmatrix}
\frac{\partial}{\partial\theta^\intercal} g_1(X_i,\theta,\hat\gamma_{-k}) \big|_{\theta=\check\theta_{n,1}}\\
\vdots\\
\frac{\partial}{\partial\theta^\intercal} g_{d_g}(X_i,\theta,\hat\gamma_{-k}) \big|_{\theta=\check\theta_{n,d_g}}\\
\end{bmatrix}~.
\end{align*}
By Assumption \ref{Ass: Data}(ii), $(\hat\theta_n,\hat\gamma_{-k})\in U_0\times V_0$ with probability approaching one and hence we obtain the first order condition that, for $\hat\phi_{-k}(X_i)\equiv \phi(X_i,\tilde\theta_{-k},\hat\gamma_{-k}, \hat\alpha_{-k})$,
\begin{align}\label{Pro: AN, aux2}
\{\frac{1}{n}\sum_{k=1}^{K}\sum_{i\in I_k} \frac{\partial}{\partial \theta^\intercal}g(X_i,\theta,\hat\gamma_{-k})\Big|_{\theta=\hat\theta_n}\}^\intercal\hat\Omega_n \frac{1}{n}\sum_{k=1}^{K}\sum_{i\in I_k} \{g(X_i,\hat\theta_n,\hat\gamma_{-k})+\hat\phi_{-k}(X_i)\}=0~.
\end{align}
Combining results \eqref{Pro: AN, aux1} and \eqref{Pro: AN, aux2} we thus obtain that
\begin{align}\label{Pro: AN, aux3}
\frac{1}{n}\sum_{k=1}^{K}&\sum_{i\in I_k} \nabla_\theta g(X_i,\hat\theta_n,\hat\gamma_{-k})^\intercal\hat\Omega_n \check{G}_n\sqrt n\{\hat\theta_n-\theta_0\}\notag\\
&= -\frac{1}{n}\sum_{k=1}^{K}\sum_{i\in I_k} \nabla_\theta g(X_i,\hat\theta_n,\hat\gamma_{-k})^\intercal\hat\Omega_n \frac{1}{\sqrt n}\sum_{k=1}^{K}\sum_{i\in I_k}\{ g(X_i,\theta_0,\hat\gamma_{-k})+\hat\phi_{-k}(X_i)\}\notag\\
&\qquad +o_p(1)~.
\end{align}
Since each $\check\theta_{n,j}$ lies between $\theta_0$ and $\hat\theta_n$, we must have $\check\theta_{n,j}\overset{p}{\to} \theta_0$ by Assumption \ref{Ass: Data}(ii). By Assumptions \ref{Ass: Model}(iii)(iv) and \ref{Ass: Data}(i)(ii)(iii)(iv), Lemma \ref{Lem: LLN} then implies
\begin{align}\label{Pro: AN, aux4}
\frac{1}{n}\sum_{k=1}^{K}\sum_{i\in I_k} \nabla_\theta g(X_i,\hat\theta_n,\hat\gamma_{-k}) \overset{p}{\to} G_0~,\quad \check{G}_n\overset{p}{\to} G_0~.
\end{align}
By Assumption \ref{Ass: Data}(v), it follows from result \eqref{Pro: AN, aux4} that
\begin{gather}
\frac{1}{n}\sum_{k=1}^{K}\sum_{i\in I_k} \nabla_\theta g(X_i,\hat\theta_n,\hat\gamma_{-k}) \hat\Omega_n \overset{p}{\to} G_0^\intercal\Omega_0~, \label{Pro: AN, aux5}\\
\frac{1}{n}\sum_{k=1}^{K}\sum_{i\in I_k} \nabla_\theta g(X_i,\hat\theta_n,\hat\gamma_{-k}) \hat\Omega_n \check{G}_n\overset{p}{\to} G_0^\intercal\Omega_0 G_0~,\label{Pro: AN, aux6}
\end{gather}
where the limit $G_0^\intercal\Omega_0 G_0$ is nonsingular by Assumption \ref{Ass: Model}(iii). 

By simple algebra, we next observe the identity:
\begin{align}\label{Pro: AN, aux7}
\frac{1}{\sqrt n}&\sum_{k=1}^{K}\sum_{i\in I_k}  \{ g(X_i,\theta_0,\hat\gamma_{-k}) +\hat\phi_{-k}(X_i)\}-\frac{1}{\sqrt n}\sum_{i=1}^{n} \psi_0(X_i)  \notag\\
& = \sum_{k=1}^{K}\frac{1}{\sqrt n}\sum_{i\in I_k}\{ \psi(X_i,\theta_0,\hat\gamma_{-k},\alpha_0)-\psi_0(X_i)\}\notag\\
& \quad + \sum_{k=1}^{K}\frac{1}{\sqrt n}\sum_{i\in I_k}\{ \phi(X_i,\tilde\theta_{-k},\gamma_0,\hat\alpha_{-k})-\phi_0(X_i)\} + \sum_{k=1}^{K}\frac{1}{\sqrt n}\sum_{i\in I_k}\hat\Delta_{-k}(X_i)~,
\end{align}
where $\phi_0(X_i)\equiv \phi(X_i,\theta_0,\gamma_0,\alpha_0)$, and
\begin{align*}
\hat\Delta_{-k}(X_i) \equiv \hat\phi_{-k}(X_i) - \phi(X_i,\tilde\theta_{-k},\gamma_0,\hat\alpha_{-k}) - \phi(X_i,\theta_0,\hat\gamma_{-k},\alpha_0) +\phi_0(X_i)~.
\end{align*}
Since $n_k\le n$, we may obtain by Assumptions \ref{Ass: FSIF}(i)(ii) and \ref{Ass: Data}(i)(iii) that
\begin{align}\label{Pro: AN, aux8}
\|E[\frac{1}{\sqrt n}\sum_{i\in I_k}&\{ \psi(X_i,\theta_0,\hat\gamma_{-k},\alpha_0)-\psi_0(X_i)\}\big|\{X_j\colon j\notin I_k\}]\|\notag\\
&\quad =\frac{n_k}{\sqrt n} \|E_X[\psi(X_i,\theta_0,\hat\gamma_{-k},\alpha_0)]\|\lesssim \sqrt{n} \|\hat{\gamma}_{-k}-\gamma_0\|_{\mathbf H}^2=o_p(1)~.
\end{align}
Moreover, we have by Assumptions \ref{Ass: Data}(i) and $n_k\le n$ that, for each $j=1,\ldots,d_g$,
\begin{multline}\label{Pro: AN, aux9}
\mathrm{Var}(\frac{1}{\sqrt n}\sum_{i\in I_k}\{\psi_j(X_i,\theta_0,\hat\gamma_{-k},\alpha_0)-\psi_{0,j}(X_i)\}\big|\{X_j\colon j\notin I_k\})\\
\le E_X[(\psi_j(X,\theta_0,\hat\gamma_{-k},\alpha_0)-\psi_{0,j}(X_i))^2]=o_p(1)~,
\end{multline}
where the last step follows by Assumptions \ref{Ass: FSIF}(iv) and \ref{Ass: Data}(iii). Combining results \eqref{Pro: AN, aux8} and \eqref{Pro: AN, aux9}, we thus have: for each $j=1,\ldots,d_g$,
\begin{align}\label{Pro: AN, aux10}
\frac{1}{\sqrt n}\sum_{i\in I_k}\{\psi_j(X_i,\theta_0,&\hat\gamma_{-k},\alpha_0)-\psi_{0,j}(X_i)\}=o_p(1)
\end{align}
conditional on $\{X_j\colon j\notin I_k\}$ and hence unconditionally by the Fubini theorem and the dominated convergence theorem. Since $K$ is fixed, it follows from result \eqref{Pro: AN, aux10} that
\begin{align}\label{Pro: AN, aux11}
\frac{1}{\sqrt n}\sum_{k=1}^{K}\sum_{i\in I_k}\{\psi(X_i,\theta_0,&\hat\gamma_{-k},\alpha_0)-\psi_0(X_i)\} = o_p(1)~.
\end{align}

By Assumptions \ref{Ass: FSIF}(iii) and \ref{Ass: Data}(i), we note that, almost surely,
\begin{align}\label{Pro: AN, aux12}
E[\frac{1}{\sqrt n}\sum_{i\in I_k}\{ \phi(X_i,\tilde\theta_{-k},\gamma_0,\hat\alpha_{-k})-\phi_0(X_i)\}|\{X_j\colon j\notin I_k\}]=0~.
\end{align}
By arguments analogous to those leading to \eqref{Pro: AN, aux11}, it follows from \eqref{Pro: AN, aux12} that 
\begin{align}\label{Pro: AN, aux13}
\frac{1}{\sqrt n}\sum_{k=1}^{K}\sum_{i\in I_k}\{ \phi(X_i,\tilde\theta_{-k},\gamma_0,\hat\alpha_{-k})-\phi_0(X_i)\} = o_p(1)~.
\end{align}
By Assumptions \ref{Ass: FSIF}(vi)(vii) and \ref{Ass: Data}(i)(iii), we may also employ arguments analogous to those leading to result \eqref{Pro: AN, aux11} to deduce that
\begin{align}\label{Pro: AN, aux14}
\sum_{k=1}^{K}\frac{1}{\sqrt n}\sum_{i\in I_k}\hat\Delta_{-k}(X_i) = o_p(1)~.
\end{align}
Combining results \eqref{Pro: AN, aux7}, \eqref{Pro: AN, aux11}, \eqref{Pro: AN, aux13}, and \eqref{Pro: AN, aux14}, we thus arrive at
\begin{align}\label{Pro: AN, aux15}
\frac{1}{\sqrt n}\sum_{k=1}^{K}\sum_{i\in I_k}  \{ g(X_i,\theta_0,\hat\gamma_{-k}) +\hat\phi_{-k}(X_i)\}=\frac{1}{\sqrt n}\sum_{i=1}^{n} \psi_0(X_i) +o_p(1)~.
\end{align}
In turn, \eqref{Pro: AN, aux3}, \eqref{Pro: AN, aux5}, \eqref{Pro: AN, aux6}, and \eqref{Pro: AN, aux15}, together with Assumptions \ref{Ass: FSIF}(i) and \ref{Ass: Data}(i) and the continuous mapping theorem, allow us to conclude the proof of part (i).

For part (ii), Lemma \ref{Lem: LLN} gives by Assumptions \ref{Ass: Data}(i)(iii)(iv) and \ref{Ass: variance estimation}(i)(ii) that
\begin{align}\label{Pro: AN, aux16}
\hat\Sigma_n\overset{p}{\to} \Sigma_0~,
\end{align} 
and by Assumptions \ref{Ass: variance estimation}(iii) and \ref{Ass: Data}(i)(iii)(iv) that
\begin{align}\label{Pro: AN, aux17}
\tilde G_n\equiv \frac{1}{n}\sum_{k=1}^{K}\sum_{i\in I_k} \nabla_\theta g(X_i,\theta_0,\hat\gamma_{-k})\overset{p}{\to} G_0~.
\end{align}
By Assumption \ref{Ass: variance estimation}(iii), we have that whenever $(\hat\theta_n,\hat\gamma_{-k})\in U_0\times V_0$,
\begin{multline}\label{Pro: AN, aux18}
\|\hat G_n-\tilde G_n\| = \|\frac{1}{n}\sum_{k=1}^{K}\sum_{i\in I_k} \{\nabla_\theta g(X_i,\hat\theta_n,\hat\gamma_{-k})-\nabla_\theta g(X_i,\theta_0,\hat\gamma_{-k})\}\|\\
\le \frac{1}{n}\sum_{k=1}^{K}\sum_{i\in I_k} b(X_i,\hat\gamma_{-k})\|\hat\theta_n-\theta_0\|~.
\end{multline}
By Assumption \ref{Ass: Data}(iii), we note that, for any $\epsilon>0$,
\begin{align}\label{Pro: AN, aux19}
\limsup_{n\to\infty}P(\frac{1}{n}\sum_{k=1}^{K}&\sum_{i\in I_k} b(X_i,\hat\gamma_{-k})>\epsilon)\notag\\
&\le \limsup_{n\to\infty}P(\frac{1}{n}\sum_{k=1}^{K}\sum_{i\in I_k} b(X_i,\hat\gamma_{-k})>\epsilon,\hat\gamma_{-k}\in V_0\text{ for all }k)\notag\\
&\le \limsup_{n\to\infty}P(\frac{1}{n}\sum_{k=1}^{K}\sum_{i\in I_k}\sup_{\gamma\in V_0} b(X_i,\gamma)>\epsilon)~.
\end{align}
By Markov's inequality, Assumptions \ref{Ass: Data}(i) and \ref{Ass: variance estimation}(iii), and $\epsilon>0$ being arbitrary, result \eqref{Pro: AN, aux19} implies $\frac{1}{n}\sum_{k=1}^{K}\sum_{i\in I_k} b(X_i,\hat\gamma_{-k})=O_p(1)$. This, together with Assumption \ref{Ass: Data}(ii) and results \eqref{Pro: AN, aux17} and \eqref{Pro: AN, aux18}, implies that
\begin{align}\label{Pro: AN, aux20}
\hat G_n\overset{p}{\to} G_0~.
\end{align}
The conclusion $\hat V_{\theta,n}\overset{p}{\to}V_{\theta,0}$ of part (ii) then follows from \eqref{Pro: AN, aux16}, \eqref{Pro: AN, aux20}, Assumptions \ref{Ass: Model}(iii) and \ref{Ass: Data}(v), and the continuous mapping theorem. \qed
 
\end{appendices}

\phantomsection
\addcontentsline{toc}{section}{References}
\bibliographystyle{ecta}
\bibliography{bibliography}
%\bibliography{D:/Dropbox/Common/LaTeX/mybibliography}
\end{document}